\RequirePackage{filecontents}

\documentclass[dissertation,final,vertlayout,pdfa,nologo,math]{aaltoseries}

\makeatletter
\@ifpackageloaded{inputenc}{%
	\inputencoding{utf8}}{%
	\usepackage[utf8]{inputenc}}
\makeatother

\hypersetup{bookmarks=true, colorlinks=false, pagebackref=true, hypertexnames=true, hidelinks}

\usepackage[hyperpageref]{backref}
\renewcommand*{\backref}[1]{}
\renewcommand*{\backrefalt}[4]{{
		\ifcase #1 Not cited.%
		\or Cited on page~#2.%
		\else Cited on pages #2.%
		\fi%
	}}

\usepackage[english]{babel}
\usepackage{amsmath,amsthm,amssymb,bm}
\usepackage{emptypage}

\usepackage[SchoolofEngineering]{aaltologo}

\usepackage{CJKutf8}
\usepackage[round, authoryear]{natbib}
\usepackage{graphicx}
\usepackage{mathtools}
\usepackage[shortlabels]{enumitem}

\usepackage{tikz}
\usetikzlibrary{fadings}
\usetikzlibrary{patterns}
\usetikzlibrary{shadows.blur}
\usetikzlibrary{shapes}

\newcommand{\thmenvcounter}{chapter}
\newcommand{\cu}[1]{
	\ifcat\noexpand#1\relax
	\bm{#1}
	\else
	\mathbf{#1}
	\fi
}

\newcommand{\tash}[2]{\frac{\partial #1}{\partial #2}}
\newcommand{\tashh}[3]{\frac{\partial^2 #1}{\partial #2 \, \partial #3}}

\newcommand{\diff}{\mathop{}\!\mathrm{d}}

\newcommand{\cond}{{\;|\;}}
\newcommand{\condbig}{{\;\big|\;}}

\newcommand{\condBigg}{{\;\Bigg|\;}}

\let\sup\relax
\let\inf\relax
\let\lim\relax
\DeclareMathOperator*{\argmin}{arg\,min\,}  
\DeclareMathOperator*{\sup}{sup\,}  
\DeclareMathOperator*{\inf}{inf\,}  
\DeclareMathOperator*{\lim}{lim\,}  

\newcommand{\sgn}{\operatorname{sgn}}           

\newcommand{\expecsym}{\operatorname{\mathbb{E}}}     
\newcommand{\covsym}{\operatorname{Cov}}     
\newcommand{\varrsym}{\operatorname{Var}}     
\newcommand{\diagsym}{\operatorname{diag}}     
\newcommand{\tracesym}{\operatorname{tr}}           

\let\expec\relax
\let\cov\relax
\let\varr\relax
\let\diag\relax
\let\trace\relax

\makeatletter
\newcommand{\expec}{\@ifstar{\@expecauto}{\@expecnoauto}}
\newcommand{\@expecauto}[1]{\expecsym \left[ #1 \right]}
\newcommand{\@expecnoauto}[1]{\expecsym \, [#1]}
\newcommand{\expecbig}[1]{\expecsym \big[ #1 \big]}
\newcommand{\expecBig}[1]{\expecsym \Big[ #1 \Big]}

\newcommand{\cov}{\@ifstar{\@covauto}{\@covnoauto}}
\newcommand{\@covauto}[1]{\covsym \left[ #1 \right]}
\newcommand{\@covnoauto}[1]{\covsym \, [#1]}
\newcommand{\covbig}[1]{\covsym \big[ #1 \big]}

\newcommand{\varr}{\@ifstar{\@varrauto}{\@varrnoauto}}
\newcommand{\@varrauto}[1]{\varrsym \left[ #1 \right]}
\newcommand{\@varrnoauto}[1]{\varrsym \, [#1]}
\newcommand{\varrbig}[1]{\varrsym \big[ #1 \big]}

\newcommand{\diag}{\@ifstar{\@diagauto}{\@diagnoauto}}
\newcommand{\@diagauto}[1]{\diagsym \left( #1 \right)}
\newcommand{\@diagnoauto}[1]{\diagsym \, (#1)}

\newcommand{\trace}{\@ifstar{\@traceauto}{\@tracenoauto}}
\newcommand{\@traceauto}[1]{\tracesym \left( #1 \right)}
\newcommand{\@tracenoauto}[1]{\tracesym \, (#1)}

\makeatother

\newcommand{\A}{\mathcal{A}}           
\newcommand{\Am}{\overline{\mathcal{A}}}           

\newcommand*{\trans}{{\mkern-1.5mu\mathsf{T}}}

\newcommand*{\T}{\mathbb{T}} 
\newcommand*{\R}{\mathbb{R}} 
\newcommand*{\N}{\mathbb{N}} 

\newcommand*{\FF}{\mathcal{F}} 
\newcommand*{\PP}{\mathbb{P}} 
\newcommand*{\GP}{\mathrm{GP}} 

\newcommand{\mineig}{\lambda_{\mathrm{min}}}
\newcommand{\maxeig}{\lambda_{\mathrm{max}}}

\let\norm\relax
\DeclarePairedDelimiter{\normbracket}{\lVert}{\rVert}
\newcommand{\norm}{\normbracket}
\newcommand{\normbig}[1]{\big \lVert #1 \big \rVert}
\newcommand{\normBig}[1]{\Big \lVert #1 \Big\rVert}

\let\innerp\relax
\DeclarePairedDelimiter{\innerpbracket}{\langle}{\rangle}
\newcommand{\innerp}{\innerpbracket}

\let\abs\relax
\DeclarePairedDelimiter{\absbracket}{\lvert}{\rvert}
\newcommand{\abs}{\absbracket}
\newcommand{\absbig}[1]{\big \lvert #1 \big \rvert}
\newcommand{\absBig}[1]{\Big \lvert #1 \Big\rvert}

\newcommand{\mBesselsec}{\operatorname{K}_\nu}
\newcommand{\jacob}{\operatorname{J}}
\newcommand{\hessian}{\operatorname{H}}

\def\matern{Mat\'{e}rn }

\makeatletter
\@ifundefined{thmenvcounter}{}
{%
	\newtheorem{envcounter}{EnvcounterDummy}[\thmenvcounter]
	\newtheorem{theorem}[envcounter]{Theorem}
	
	\newtheorem{lemma}[envcounter]{Lemma}
	\newtheorem{corollary}[envcounter]{Corollary}
	\newtheorem{remark}[envcounter]{Remark}
	\newtheorem{example}[envcounter]{Example}
	\newtheorem{definition}[envcounter]{Definition}
	\newtheorem{algorithm}[envcounter]{Algorithm}
	\newtheorem{assumption}[envcounter]{Assumption}
}
\makeatother

\author{Zheng Zhao}
\title{State-Space Deep Gaussian Processes with Applications}

\begin{document}

\includepdf[pages=-]{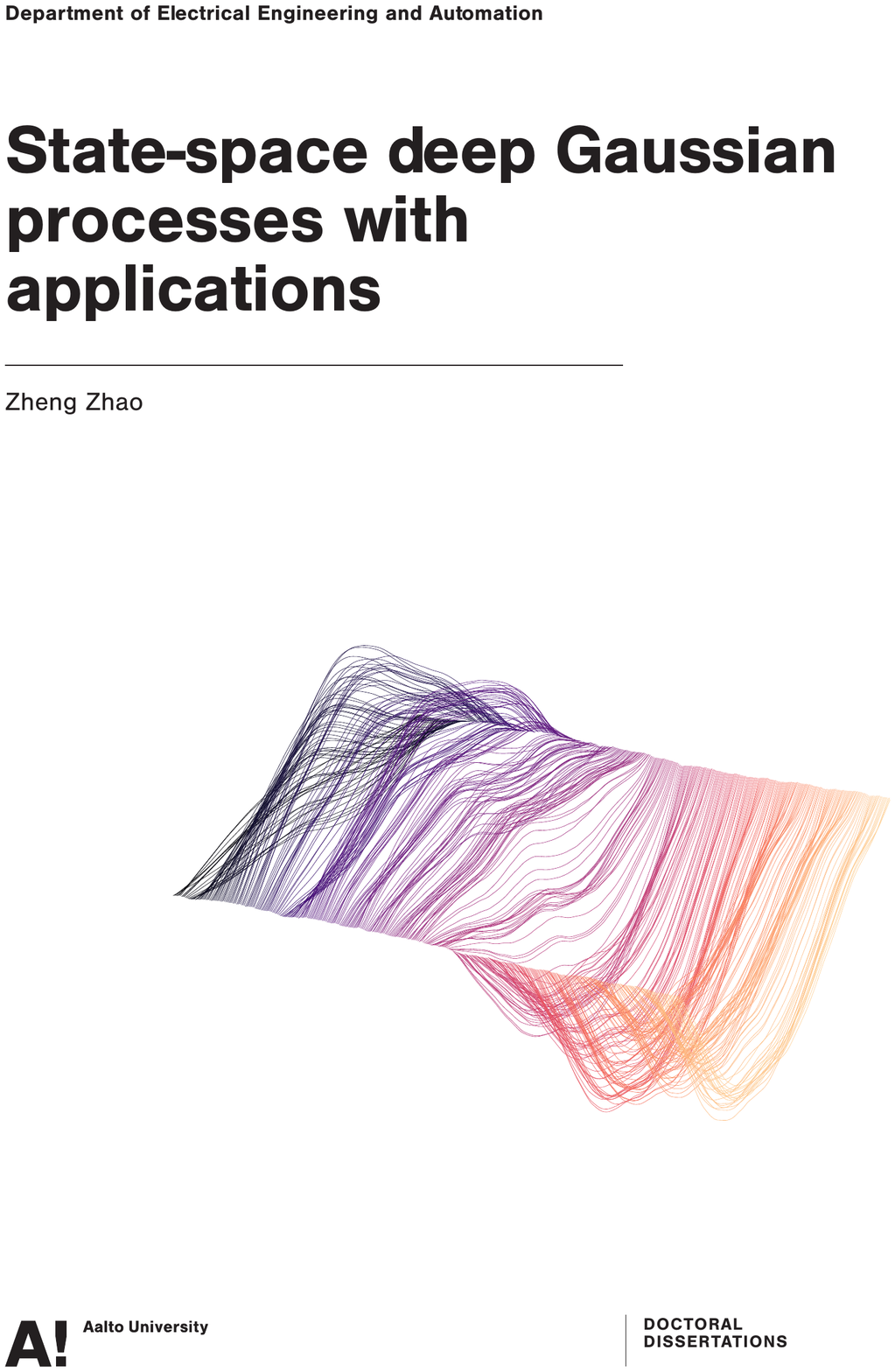}

\draftabstract{This thesis is mainly concerned with state-space approaches for solving deep (temporal) Gaussian process (DGP) regression problems. More specifically, we represent DGPs as hierarchically composed systems of stochastic differential equations (SDEs), and we consequently solve the DGP regression problem by using state-space filtering and smoothing methods. The resulting state-space DGP (SS-DGP) models generate a rich class of priors compatible with modelling a number of irregular signals/functions. Moreover, due to their Markovian structure, SS-DGPs regression problems can be solved efficiently by using Bayesian filtering and smoothing methods. The second contribution of this thesis is that we solve continuous-discrete Gaussian filtering and smoothing problems by using the Taylor moment expansion (TME) method. This induces a class of filters and smoothers that can be asymptotically exact in predicting the mean and covariance of stochastic differential equations (SDEs) solutions.  Moreover, the TME method and TME filters and smoothers are compatible with simulating SS-DGPs and solving their regression problems. Lastly, this thesis features a number of applications of state-space (deep) GPs. These applications mainly include, (i) estimation of unknown drift functions of SDEs from partially observed trajectories and (ii) estimation of spectro-temporal features of signals.}

\setcounter{page}{0}

\begin{preface}[Helsinki]{\large{\begin{CJK}{UTF8}{bkai}\\趙~正\end{CJK}}
}
The research work in this thesis has been carried out in the Department of Electrical Engineering and Automation, Aalto University, during the years 2018-2021. My doctoral studies officially started in April of 2018, while most of the pivotal work came in 2020-2021. During this time, my doctoral research was financially supported by Academy of Finland and Aalto ELEC Doctoral School. The Aalto Scientific Computing team and the Aalto Learning Center also provided useful computational and literature resources for my studies. I particularly enjoyed the Spring, Autumn, and Winter in Finland, which allowed me to find inner peace and focus on my research.

I would like to offer my greatest gratitude to Prof. Simo S\"{a}rkk\"{a} who is my supervisor and mentor, and without whom this work would never have been possible. After finishing my master studies in Beijing University of Technology in 2017, I found myself lost in finding a ``meaningful'' way of life in the never-sleeping metropolis that is Beijing. This quest was fulfilled when Simo offered me the opportunity of pursuing a doctoral degree under his supervision. Disregarding my bewilderment on the research path in the beginning, Simo's patience and valuable guidance led me to a research area that I am fascinated in. Over the years, Simo's help, support, and friendship have helped me become a qualified and independent researcher. I think very highly of Simo's supervision, and I almost surely could not have found a better supervisor. 

During my years in the campus, I owe a great thanks to Rui Gao (\begin{CJK}{UTF8}{bkai}高~睿\end{CJK}) who is a brilliant, learnt, and erudite researcher.

I would like to thank these few people that have accompanied me through joy and sorrow, I name: Adrien Corenflos and Christos Merkatas. I thank you for the friendship and relieving me from solitude\footnote{This was written under constraint.}.  

During my years in Aalto university, I have shared my office with Marco Soldati, Juha Sarmavuori, Janne Myll\"{a}rinen, Fei Wang (\begin{CJK}{UTF8}{bkai}王~斐\end{CJK}), Jiaqi Liu (\begin{CJK}{UTF8}{bkai}劉~佳琦\end{CJK}), Ajinkya Gorad, Masaya Murata (\begin{CJK}{UTF8}{bkai}村田~真哉\end{CJK}), and Otto Kangasmaa. I thank them all for filling the office with happiness and joy. I especially thank Marco Soldati who offered me honest friendship, lasagne, and taught me many useful Italian phrases. My thanks also go to Lauri Palva, Zenith Purisha, Joel Jaskari, Sakira Hassan, Fatemeh Yaghoobi, Abubakar Yamin, Zaeed Khan, Xiaofeng Ma (\begin{CJK}{UTF8}{bkai}馬 曉峰\end{CJK}), Prof. Ivan Vujaklija, Dennis Yeung, Wendy Lam, Prof. Ilkka Laakso, Marko Mikkonen, Noora Matilainen, Juhani Kataja, Linda Srbova, and Tuomas Turunen. All these amazing people made working at Aalto a real pleasure. I would also like to give my thanks to Laila Aikala who kindly offered me a peaceful place to stay in Espoo.

I warmly thank Prof. Leo K\"{a}rkk\"{a}inen for the collaboration on the AI in Health Technology course and our inspiring discussions on many Thursdays and Fridays. I particularly enjoyed the collaboration with Muhammad Fuady Emzir who offered me knowledge generously and with no reservations. Many thanks go to my coauthors Prof. Roland Hostettler, Prof. Ali Bahrami Rad, Filip Tronarp, and Toni Karvonen. I also appreciated the collaboration with Sarang Thombre and Toni Hammarberg from Finnish Geospatial Research Institute, Prof. Ville V. Lehtola from University of Twente, and Tuomas Lumikari from Helsinki University Hospital. I also thank Prof. Lassi Roininen and Prof. Arno Solin for their time and valuable advice. 

Lastly, I would like to thank my parents and sister who support me persistently as always. 

\end{preface}

\clearpage
\tableofcontents

\languagecheck{Adrien Corenflos, Christos Merkatas, and Dennis Yeung}

%
\def\authorscontributionname{Author's contribution}
\listofpublications


\abbreviations

\begin{description}[style=multiline,leftmargin=3cm]
\item[CD-FS] Continuous-discrete filtering and smoothing
\item[DGP] Deep Gaussian process
\item[GFS] Gaussian approximated density filter and smoother
\item[GMRF] Gaussian Markov random field
\item[GP] Gaussian process
\item[It\^{o}-1.5] It\^{o}--Taylor strong order 1.5
\item[LCD] Locally conditional discretisation
\item[MAP] Maximum a posteriori
\item[MCMC] Markov chain Monte Carlo
\item[MLE] Maximum likelihood estimation
\item[NSGP] Non-stationary Gaussian process
\item[ODE] Ordinary differential equation
\item[PDE] Partial differential equation
\item[RBF] Radial basis function
\item[R-DGP] Regularised (batch) deep Gaussian process
\item[R-SS-DGP] Regularised state-space deep Gaussian process
\item[RTS] Rauch--Tung--Striebel
\item[SDE] Stochastic differential equation
\item[SS-DGP] State-space deep Gaussian process
\item[SS-GP] State-space Gaussian process
\item[TME] Taylor moment expansion
\end{description}

\symbols

\begin{description}[style=multiline,leftmargin=3cm]
\item[$a$] Drift function of SDE
\item[$A$] Drift matrix of linear SDE
\item[$\A$] Infinitesimal generator
\item[$\Am$] Multidimensional infinitesimal generator
\item[$b$] Dispersion function of SDE
\item[$B$] Dispersion matrix of linear SDE
\item[$c$] Constant
\item[$\mathcal{C}^k(\Omega; \Pi)$] Space of $k$ times continuously differentiable functions on $\Omega$ mapping to $\Pi$
\item[$C(t,t')$] Covariance function
\item[$C_{\mathrm{Mat.}}(t,t')$] \matern covariance function
\item[$C_{\mathrm{NS}}(t,t')$] Non-stationary \matern covariance function
\item[$C_{1:T}$] Covariance/Gram matrix by evaluating the covariance function $C(t, t')$ on Cartesian grid $(t_1,\ldots, t_T) \times (t_1,\ldots, t_T)$
\item[$\covsym$] Covariance
\item[$\cov{X \mid Y}$] Conditional covariance of random variable $X$ given another random variable $Y$
\item[$\cov{X \mid y}$] Conditional covariance of random variable $X$ given the realisation $y$ of random variable $Y$
\item[$d$] Dimension of state variable
\item[$d_i$] Dimension of the $i$-th GP element
\item[$d_y$] Dimension of measurement variable
\item[$\det$] Determinant
\item[$\diagsym$] Diagonal matrix 
\item[$\expecsym$] Expectation
\item[$\expec{X \mid \mathcal{F}}$] Conditional expectation of $X$ given sigma-algebra $\mathcal{F}$
\item[$\expec{X \cond Y}$] Conditional expectation of $X$ given the sigma-algebra generated by random variable $Y$
\item[$\expec{X \cond y}$] Conditional expectation of $X$ given the realisation $y$ of random variable $Y$
\item[$f$] Approximate transition function in discrete state-space model
\item[$f^M$] $M$-order TME approximated transition function in discrete state-space model
\item[$\check{f}$] Exact transition function in discrete state-space model
\item[$\mathring{f}_j$] $j$-th frequency component
\item[$\FF$] Sigma-algebra
\item[$\FF_t$] Filtration
\item[$\FF_t^W$] Filtration generated by $W$ and initial random variable
\item[$g$] Transformation function
\item[$\mathrm{GP}(0, C(t,t'))$] Zero-mean Gaussian process with covariance function $C(t,t')$.
\item[$h$] Measurement function
\item[$H$] Measurement matrix
\item[$\hessian_x f$] Hessian matrix of $f$ with respect to $x$
\item[$I$] Identity matrix
\item[$J$] Set of conditional dependencies of GP elements
\item[$\jacob_x f$] Jacobian matrix of $f$ with respect to $x$
\item[$K$] Kalman gain
\item[$\mBesselsec$] Modified Bessel function of the second kind with parameter $\nu$
\item[$\ell$] Length scale parameter
\item[$\mathcal{L}^\mathrm{A}$] Augmented Lagrangian function
\item[$\mathcal{L}^\mathrm{B}$] MAP objective function of batch DGP
\item[$\mathcal{L}^\mathrm{B-REG}$] $L^1$-regularisation term for batch DGP
\item[$\mathcal{L}^\mathrm{S}$] MAP objective function of state-space DGP
\item[$\mathcal{L}^\mathrm{S-REG}$] $L^1$-regularisation term for state-space DGP
\item[$m(t)$] Mean function
\item[$m^-_k$] Predictive mean at time $t_k$
\item[$m^f_k$] Filtering mean at time $t_k$
\item[$m^s_k$] Smoothing mean at time $t_k$ 
\item[$M$] Order of Taylor moment expansion
\item[$N$] Order of Fourier expansion
\item[$\mathrm{N}(x\mid m, P)$] Normal probability density function with mean $m$ and covariance $P$
\item[$\N$] Set of natural numbers
\item[$O$] Big $O$ notation
\item[$p_X(x)$] Probability density function of random variable $X$
\item[$p_{X \cond Y}(x\cond y)$] Conditional probability density function of $X$ given $Y$ taking value $y$
\item[$P^-_k$] Predictive covariance at time $t_k$
\item[$P^f_k$] Filtering covariance at time $t_k$
\item[$P^s_k$] Smoothing covariance at time $t_k$
\item[$P^{i,j}_k$] Filtering covariance of the $i$ and $j$-th state elements at time $t_k$
\item[$\mathbb{P}$] Probability measure
\item[$q_k$] Approximate process noise in discretised state-space model at time $t_k$
\item[$\check{q}_k$] Exact process noise in discretised state-space model at time $t_k$
\item[$Q_k$] Covariance of process noise $q_k$
\item[$R_{M, \phi}$] Remainder of $M$-order TME approximation for target function $\phi$
\item[$\R$] Set of real numbers
\item[$\R_{>0}$] Set of positive real numbers
\item[$\R_{<0}$] Set of negative real numbers
\item[$\sgn$] Sign function
\item[$\mathcal{S}_{m, P}$] Sigma-point approximation of Gaussian integral with mean $m$ and covariance $P$
\item[$t$] Temporal variable
\item[$\tracesym$] Trace
\item[$t_0$] Initial time
\item[$T$] Number of measurements
\item[$\T$] Temporal domain $\T\coloneqq [t_0, \infty)$
\item[$U$] (State-space) GP
\item[$U^i_{j_i}$] (State-space) GP element in $\mathcal{V}$ indexed by $i$, and it is also a parent of the $j_i$-th GP element in $\mathcal{V}$
\item[$U_{1:T}$] Collection of $U(t_1), U(t_2),\ldots, U(t_T)$
\item[$\mathcal{U}^i$] Collection of parents of $U^i_{j_i}$
\item[$V$] (State-space) deep GP
\item[$V_k$] Shorthand of $V(t_k)$
\item[$V_{1:T}$] Collection of $V(t_1), V(t_2),\ldots, V(t_T)$
\item[$\mathcal{V}$] Collection of GP elements
\item[$\varrsym$] Variance
\item[$w$] Dimension of Wiener process
\item[$W$] Wiener process
\item[$X$] Stochastic process
\item[$X_0$] Initial random variable
\item[$X_k$] Shorthand of $X(t_k)$
\item[$Y_k$] Measurement random variable at time $t_k$
\item[$Y_{1:T}$] Collection of $Y_1, Y_2,\ldots, Y_T$

\item[$\gamma$] Dimension of the state variable of Mat\'{e}rn GP
\item[$\Gamma$] Shorthand of $b(x) \, b(x)^\trans$
\item[$\varGamma$] Gamma function
\item[$\Delta t$] Time interval $t-s$
\item[$\Delta t_k$] Time interval $t_k-t_{k-1}$
\item[$\eta$] Multiplier for augmented Lagrangian function
\item[$\theta$] Auxiliary variable used in augmented Lagrangian function
\item[$\Theta_{r}$] $r$-th polynomial coefficient in TME covariance approximation
\item[$\mineig$] Minimum eigenvalue
\item[$\maxeig$] Maximum eigenvalue
\item[$\Lambda(t)$] Solution of a matrix ordinary differential equation
\item[$\cu{\Lambda}(t, s)$] Shorthand of $\Lambda(t) \, (\Lambda(s))^{-1}$
\item[$\xi_k$] Measurement noise at time $t_k$
\item[$\Xi_k$] Variance of measurement noise $\xi_k$
\item[$\rho$] Penalty parameter in augmented Lagrangian function
\item[$\sigma$] Magnitude (scale) parameter
\item[$\Sigma_M$] $M$-order TME covariance approximant
\item[$\phi$] Target function
\item[$\phi_{ij}$] $i,j$-th element of $\phi$
\item[$\phi^\mathrm{I}$] $\phi^\mathrm{I}(x) \coloneqq x$
\item[$\phi^\mathrm{II}$] $\phi^\mathrm{II}(x) \coloneqq x \, x^\trans$
\item[$\Phi$] Sparsity inducing matrix
\item[$\chi(\Delta t)$] Polynomial of $\Delta t$ associated with TME covariance approximation
\item[$\Omega$] Sample space

\item[$(\Omega, \FF, \FF_t, \PP)$] Filtered probability space with sample space $\Omega$, sigma-algebra $\FF$, filtration $\FF_t$, and probability measure $\PP$
\item[$\abs{\cdot}$] Absolute value
\item[$\norm{\cdot}_p$] $L^p$ norm or $L^p$-induced matrix norm
\item[$\norm{\cdot}_G$] Euclidean norm weighted by a non-singular matrix $G$
\item[$\nabla_x f$] Gradient of $f$ with respect to $x$
\item[$\binom{\cdot}{\cdot}$] Binomial coefficient
\item[$\innerp{\cdot, \cdot}$] Inner product
\item[$\circ$] Mapping composition
\item[$\coloneqq$] By definition
\item[$\times$] Cartesian product
\item[$a \, \wedge \, b$] Minimum of $a$ and $b$
\end{description}

\chapter{Introduction}
\label{chap:intro}
In signal processing, statistics, and machine learning, it is common to consider that noisy measurements/data are generated from a latent, unknown, function. In statistics, this is often regarded as a regression problem over the space of functions. Specifically, Bayesian statistics impose a prior belief over the latent function of interest in the form of a probability distribution. It is therefore of vital importance to choose the prior appropriately, since it will encode the characteristics of the underlying function. In recent decades, Gaussian processes\footnote{In the statistics and applied probability literature, Gaussian processes can also be found under the name of Gaussian fields, in particular when they are multidimensional in the input. Depending on the context, we may use one or the other terminology interchangeably.}~\citep[GPs,][]{Carl2006GPML} have become a popular family of prior distributions over functions, and they have been used successfully in numerous applications~\citep{Roberts2013, Hennig2015, Kocijan2016}.

Formally, GPs are function-valued random variables that have Gaussian distributions fully determined by their mean and covariance functions. The choice of mean and covariance functions is in itself arbitrary, which allows for representing functions with various properties. As an example, \matern covariance functions are used as priors to functions with different degrees of differentiability~\citep{Carl2006GPML}. However, the use of GPs in practice usually involves two main challenges.

The first challenge lies in the expensive \textit{computational cost} of training and parameter estimation. Due to the necessity of inverting covariance matrices during the learning phase, the computational complexity of standard GP regression and parameter estimation is cubic in the number of measurements. This makes GP computationally infeasible for large-scale datasets. Moreover, when the sampled data points are densely located, the covariance matrices that need inversion may happen to be numerically singular or close to singular, making the learning process unstable.

The second challenge is related to modelling of irregular functions, such as piecewise smooth functions, or functions that have time-varying features (e.g., frequency or volatility). Many commonly-used GPs (e.g., with \matern covariance functions) fail to cover these irregular functions mainly because their probability distributions are invariant under translation (i.e., they are said to be \textit{stationary}). This behaviour is illustrated in Figure~\ref{fig:gp-fail}, where we show that a \matern GP poorly fits two irregular functions (i.e., a rectangular signal and a composite sinusoidal signal), because the GP's parameters/features are assumed to be constant over time. Specifically, in the rectangular signal example, in order to model the discontinuities, the \matern GP recovers a small global length scale ($\ell \approx 0.04$) which results in poor fitting in the continuous and flat parts. Similarly, in the composite sinusoidal signal example, the GP learns a small global length scale ($\ell \approx 0.01$) in order to model the high-frequency sections of the signal. This too results in poor fitting the low-frequency section of the signal.

\begin{figure}[t!]
	\centering
	\includegraphics[width=.95\linewidth]{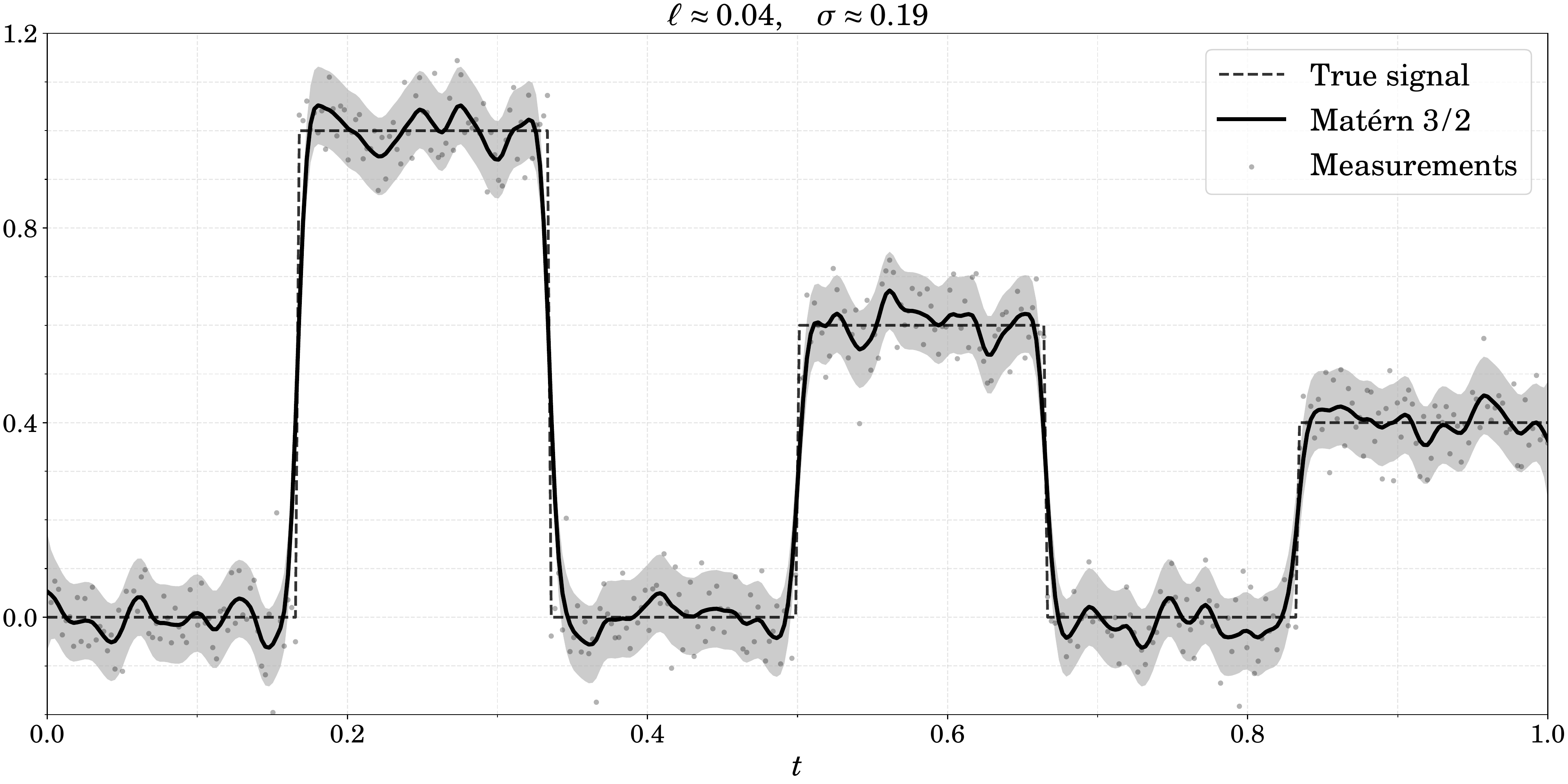}\\
	\includegraphics[width=.95\linewidth]{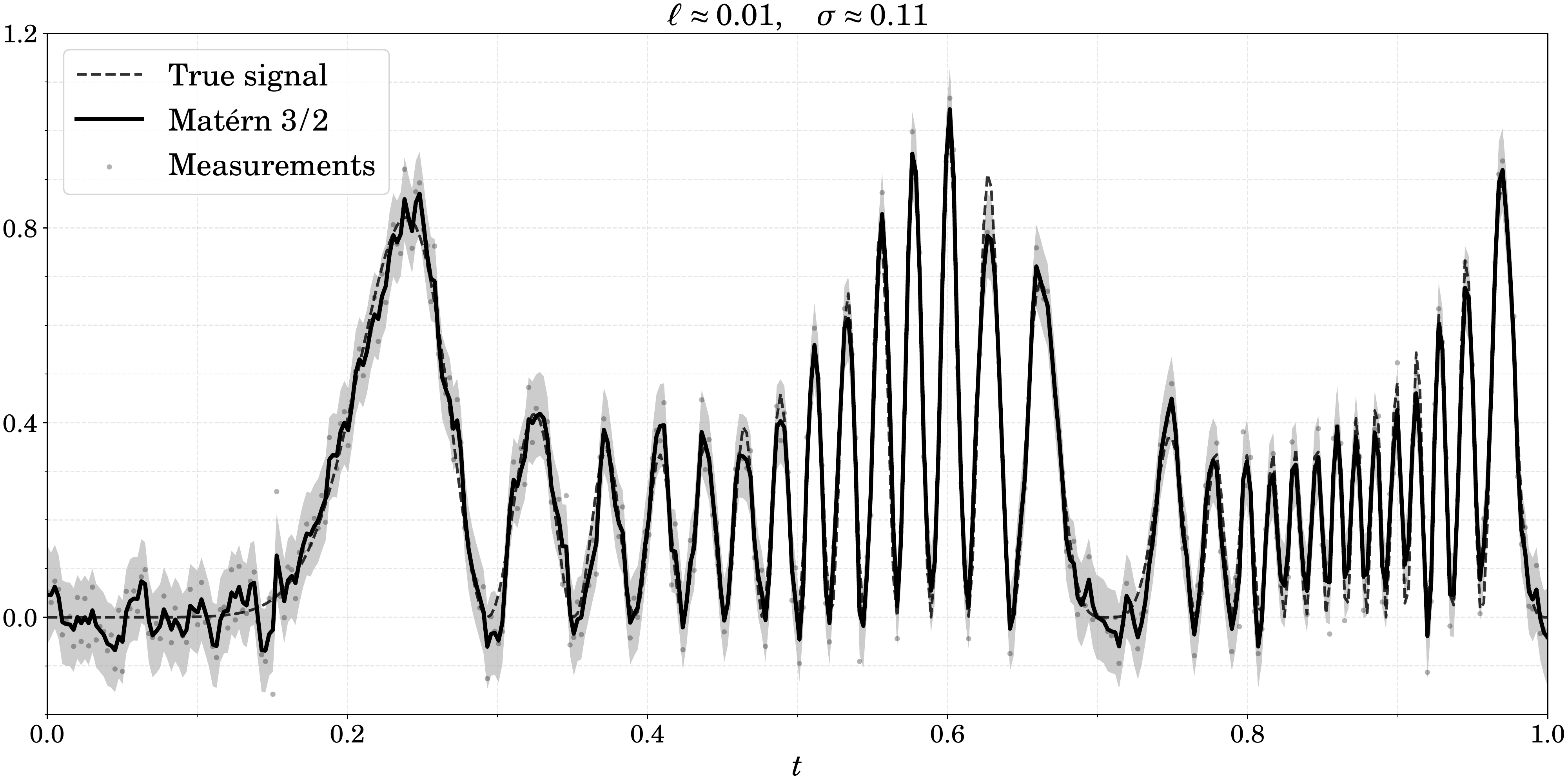}
	\caption{Mat\'{e}rn $\nu=3\,/\,2$ GP regression on a magnitude-varying rectangular signal (top) and a composite sinusoidal signal (bottom). The parameters $\ell$ and $\sigma$ are learnt by maximum likelihood estimation. The figures are taken from~\citet{Zhao2020SSDGP}.}
	\label{fig:gp-fail}
\end{figure}

The main aim of this thesis is thus to introduce a new class of non-stationary (Gaussian) Markov processes, that we name \textit{state-space deep Gaussian processes (SS-DGPs)}\footnote{Please note that although the name includes the term Gaussian, SS-DGPs are typically not Gaussian distributed, but instead hierarchically conditionally Gaussian, hence the name.}. These are able to address the computational and non-stationarity challenges aforementioned, by hierarchically composing the state-space representations of GPs. Indeed, SS-DGPs are computationally efficient models due to their Markovian structure. More precisely, this means that the resulting regression problem can be solved in linear computational time (with respect to the number of measurements) by using Bayesian filtering and smoothing methods. Moreover, due to their hierarchical nature, SS-DGPs are capable of changing their features/characteristics (e.g., length scale) over time, thereby inducing a rich class of priors compatible with irregular functions. The thesis ends with a collection of applications of state-space (deep) GPs.

\section{Bibliographical notes}
\label{sec:literature-review}
In this section we provide a short and non-exhaustive review of related works in the GP literature. In particular we will focus on works that consider specifically reducing their computational complexity and allowing the non-stationarity in GPs.

\subsection*{Scalable Gaussian processes}
We now give a list of GP methods and approximations that are commonly used to reduce the computational costs of GP regression and parameter learning.

\subsubsection{Sparse approximations of Gaussian processes}
Sparse GPs approximate full-rank GPs with sparse representations by using, for example, inducing points~\citep{Snelson2006}, subsets of data~\citep{Snelson2007, Csato2002}, or approximations of marginal likelihoods~\citep{Titsias2009}, mostly relying on so-called pseudo-inputs. These approaches can reduce the computational complexity to quadratic in the number of pseudo-inputs and linear in the number of data points. In practice, the number and position of pseudo-inputs used in sparse representation must either be assigned by human experts or learnt from data~\citep{Hensman2013}. For more comprehensive reviews of sparse GPs, see, for example,~\citet{Quinonero2005unifying, Chalupka2013, LiuHaitao2020}.

\subsubsection*{Gaussian Markov random fields}
Gaussian Markov random fields~\citep[GMRFs,][]{Rue2005Book} are indexed collections of Gaussian random variables that have a Markov property (defined on graph). They are computationally efficient models because their precision matrices are sparse by construction. Methodologies for solving the regression and parameter learning problems on GMRFs can be found, for example, in~\citet{Rue2007, Rue2009N}. However, GMRFs are usually only approximations of Gaussian fields~\citep[see, e.g.,][Chapter 5]{Rue2005Book}, although explicit representations exist for some specific Gaussian fields~\citep{Lindgren2011}.

\subsubsection*{State-space representations of Gaussian processes}
State-space Gaussian processes (SS-GPs) are (temporal) Markov GPs that are solutions of stochastic differential equations~\citep[SDEs,][]{Simo2013SSGP, Sarkka2019}. Due to their Markovian structure, probability distributions of SS-GPs factorise sequentially in the time dimension. The regression problem can therefore be solved efficiently in linear time with respect to the number of data points. Moreover, leveraging the sparse structure of the precision matrix \citep{Grigorievskiy2017}, or leveraging the associativity of the Kalman filtering and smoothing operations \citep{Corenflos2021SSGP} can lead to a sublinear computational complexity.

\subsubsection*{Other data-scalable Gaussian processes}
\citet{Rasmussen2002GPexperts, Meeds2006} form mixtures of GPs by splitting the dataset into batches resulting in a computational complexity that is cubic in the batch size. This methodology can further be made parallel \citep{ZhangMinyi2019}. \citet{Lazaro2010} approximate stationary GPs with sparse spectral representations (i.e., trigonometric expansions). \citet{Gardner2018} and \citet{KeWang2019} use conjugate gradients and stochastic trace estimation to efficiently compute the marginal log-likelihood of standard GPs, as well as their gradients with respect to parameters, resulting in a quadratic computational complexity in the number of data points.

\subsection*{Non-stationary Gaussian processes}
In the below we give a list of methods that are introduced in order to induce non-stationarity in GPs.

\subsubsection*{Non-stationary covariance function-based Gaussian processes}
Non-stationary covariance functions can be constructed by making their parameters (e.g., length scale or magnitude) depend on the data position. For instance, \citet{Gibbs} and \citet{Higdon1999non} present specific examples of covariance functions where the length scale parameter depends on the spatial location. On the other hand, \citet{Paciorek2004, Paciorek2006} generalise these constructions to turn any stationary covariance function into a non-stationary one. There also exist some other non-stationary covariance functions, such as the polynomial or neural network covariance functions~\citep{Williams1998, Carl2006GPML} that can also give non-stationary GPs, but we do not review them here as they are not within the scope of this thesis. 

\subsubsection*{Composition-based Gaussian processes}
\citet{Sampson1992, Schmidt2003, Carl2006GPML} show that it is possible to construct a non-stationary GP as the pullback of an existing stationary GP by a non-linear transformation. Formally, given a stationary GP $U\colon E \to \R$, one can find a suitable transformation $\Upsilon\colon \T \to E$, such that the composition $U \circ \Upsilon\colon \T\to\R$ is a non-stationary GP on $\T$. For example, \citet{Calandra2016ManifoldGP} and~\citet{Wilson2016DeepKernel} choose $\Upsilon$ as neural networks.

\subsubsection{Warping-based Gaussian processes}
Conversely to the composition paradigm above, it is also possible to transform GPs the other way around, that is, to consider that GPs are the transformations of some non-Gaussian processes by non-linear functions~\citep{Snelson2004}. Computing the marginal log-likelihood function of these warped GPs is then done by leveraging the change-of-variables formula for Lebesgue integrals (when it applies). However, the warping can be computationally demanding as the change-of-variables formula requires computing the inverse determinant of the transformation Jacobian. This issue can be mitigated, for example, by writing the the warping scheme with multiple layers of elementary functions which have explicit inverses~\citep{Rios2019}.

\subsection*{Deep Gaussian processes}
The deterministic constructions for introducing non-stationarity GPs can be further extended in order to give a class of non-stationary non-Gaussian processes that can also represent irregular functions. While they are different in structure, the three subclasses of models presented below are usually all referred as deep Gaussian processes (DGPs) in literature.

\subsubsection*{Composition-based deep Gaussian processes}
\citet{Gredilla2012} extends the aforementioned pullback idea by taking $\Upsilon\colon \T\to E$ to be a GP instead of a deterministic mapping in order to overcome the overfitting problem. Resulting compositions of the form $U \circ \Upsilon\colon \T\to\R$ may not necessarily be GPs anymore but may provide a more flexible family of priors than that of deterministic compositions. This construction can be done recursively leading to a subclass of DGPs~\citep{Damianou2013}. However, the training of these DGPs is found to be challenging and requires approximate inference methods~\citep{Bui2016, Salimbeni2017Doubly}. Moreover, \citet{Duvenaud2014Thesis, Duvenaud2014} show that increasing the depth of DGPs can lead to a representation pathology, where samples of DGPs tend to be flat in high probability and exhibit sudden jumps. This problem can be mitigated by making their latent GP components explicitly depend on their original inputs~\citep{Duvenaud2014}.

\subsubsection*{Hierarchical parametrisation-based deep Gaussian processes}
A similar idea to compositional DGPs is to model the parameters of GPs as latent GPs. The posterior distribution of the joint model can then be computed by successive applications of Bayes' rule. As an example,~\citet{Roininen2016} consider putting a GP prior on the length scale parameter of a \matern GP and use Metropolis-within-Gibbs to sample from the posterior distribution. Similarly,~\citet{Salimbeni2017ns} model the length scale parameter of the non-stationary covariance function introduced by~\citet{Paciorek2004} as a GP, but use a variational approximation to approximate its posterior distribution. Other sampling techniques to recover the posterior distribution of these models can be found, for example, in~\citet{Heinonen2016, Karla2020}. 

\citet{Zhao2020SSDGP} and~\citet{Emzir2020} show that this hierarchy in parametrisation can be done recursively, leading to another subclass of DGPs that can be represented by stochastic (partial) differential equations. The relationship between the composition-based and parametrisation-based DGPs is also briefly discussed in~\citet{Dunlop2018JMLR}.

\section{Reproducibility}
\label{sec:codes}
In order to allow for reproducibility of our work, we provide  the following implementations.

\begin{itemize}
	\item Taylor moment expansion (Chapter~\ref{chap:tme}). Python and Matlab codes for it are available at \href{https://github.com/zgbkdlm/tme}{https://github.com/zgbkdlm/tme}.
	\item State-space deep Gaussian processes (Chapter~\ref{chap:dssgp}). Python and Matlab codes for it are available at \href{https://github.com/zgbkdlm/ssdgp}{https://github.com/zgbkdlm/ssdgp}.
	\item The Python codes for reproducing Figures~\ref{fig:gp-fail}, \ref{fig:disc-err-dgp-m12}, \ref{fig:ssdgp-vanishing-cov}, and~\ref{fig:spectro-temporal-demo}, as well as the simulations in Examples~\ref{example-tme-benes}, \ref{example:tme-softplus}, \ref{example:ssdgp-m12}, and~\ref{example:ssdgp-m32} are available at \href{https://github.com/zgbkdlm/dissertation}{https://github.com/zgbkdlm/dissertation}.
\end{itemize}

\section{Outline of the thesis}
\label{sec:outline}
This thesis consists of seven publications and overviews of them, and the thesis is organised as follows.

In Chapter~\ref{chap:cd-smoothing} we review stochastic differential equations (SDEs) and Bayesian continuous-discrete filtering and smoothing (CD-FS) problems. This chapter lays out the preliminary definitions and results that are needed in the rest of the thesis. 

Chapter~\ref{chap:tme} (related to Publication~\cp{paperTME}) shows how to solve Gaussian approximated CD-FS problems by using the Taylor moment expansion (TME) method. This chapter also features some numerical demonstrations and analyses the positive definiteness of TME covariance approximations as well as the stability of TME Gaussian filters.

Chapter~\ref{chap:dssgp} (related to Publications~\cp{paperSSDGP} and~\cp{paperRNSSGP}) introduces SS-DGPs. In particular, after defining DGPs formally, we introduce their state-space representations. Secondly, we present how to sample from SS-DGPs by combining discretisation and numerical integration. Thirdly, we illustrate the construction of SS-DGPs in the \matern sense. Fourhtly, we represent SS-DGP regression problems as CD-FS problems that we can then solve using the methods introduced in Chapter~\ref{chap:cd-smoothing}. Finally, we explain how DGPs can be regularised in the $L^1$ sense, in particular to promote sparsity at any level of the DGP component hierarchy.

Chapter~\ref{chap:apps} (related to Publications~\cp{paperDRIFT},~\cp{paperKFSECG},~\cp{paperKFSECGCONF},~\cp{paperSSDGP}, and~\cp{paperMARITIME}) introduces various applications of state-space (deep) GPs. These include estimation of the drift functions in SDEs, probabilistic spectro-temporal signal analysis, as well as modelling real-world signals (from astrophysics, human motion, and maritime navigation) with SS-DGPs.

Finally, Chapter~\ref{chap:summary} offers a summary of the contributions of the seven publications presented in this thesis, and concludes with a discussion of unsolved problems and possible future extensions. 

\chapter{Preliminaries}
\label{chap:cd-smoothing}
The main scope of this thesis is to reduce deep Gaussian process (DGP) regression problems into continuous-discrete filtering and smoothing problems by representing DGPs as stochastic differential equations. In this chapter we focus on introducing the technical materials that will be necessary in constructing and solving such representations. Section~\ref{sec:sde} is concerned with introducing stochastic differential equations and their properties. Section~\ref{sec:cd-smoothing} focuses on continuous-discrete filtering and smoothing problems as well as algorithms to solve them. Additionally, for the sake of completeness, we list several intermediate results that will be used in the course of this thesis in Section~\ref{sec:useful-theorem}.

\section{Stochastic differential equations (SDEs)}
\label{sec:sde}
Solutions to stochastic differential equations (SDEs) are a large class of continuous-time Markov processes that are commonly used to model physical, biological, or financial dynamic systems~\citep{Kloeden1992, Braumann2019}. In this section, we introduce SDEs via their stochastic integral equation interpretations, and we thereupon present a few important concepts and results, including, the notion of existence and uniqueness of their solutions, their Markovian nature, and It\^{o}'s formula. For more comprehensive reviews of SDEs, we refer the reader to, for example, \citet{Chung1990, Karatzas1991, Ikeda1992, Oksendal2003}.

\subsection{Stochastic integral equations}
One may think of SDEs as ordinary differential/integral equations with additional stochastic driving terms. Wiener processes, which are also known as Brownian motions, are the \textit{de facto} choice for modelling these driving terms as they allow to represent a rich class of stochastic processes with varying characteristics. 

\begin{definition}[Wiener process]
	\label{def:wiener-process}
	A stochastic process $W \colon \T\times \Omega \to \R$ on some probability space $(\Omega, \FF, \PP)$ is called an $\R$-valued Wiener process on $\T\coloneqq [t_0, \infty)$, if
	\begin{itemize}
		\item $W(t_0) = 0$ almost surely,
		\item $t \mapsto W(t)$ is continuous almost surely,
		\item for every integer $k\geq 1$ and real numbers $t_1 \leq t_2 \leq \cdots \leq t_k\in\T$, the increments $W(t_k) - W(t_{k-1}), \ldots, W(t_2) - W(t_1)$ are mutually independent,
		\item and, for every $t>s\in\T$, the increment $W(t) - W(s) \sim \mathrm{N}(0, t-s)$ is Gaussian distributed of mean zero and covariance $t-s$,
	\end{itemize}
	where $W(t)$ is a shorthand for the random variable $\omega \mapsto W(t, \omega)$.
\end{definition}

There are several ways to construct Wiener processes. The first rigorous construction of Wiener processes is due to Nobert Wiener~\citep{Wiener1923} who construct the Wiener process by considering the space of real-valued continuous functions on an interval (i.e., $\mathcal{C}([0,1];\R)$), and equipping it with a canonical measure (called Wiener measure) that corresponds to the law of the Wiener process~\citep{ReneBrownianBook2012, Kuo1975Book, Kuo2006Book}. The space of continuous functions equipped with the Wiener measure is called the classical/canonical Wiener space.

Nobert Wiener and Raymond Paley also show that one can construct the Wiener process by representing it with a trigonometric orthonormal basis on $[0, 1]$, and independent identically distributed Gaussian random variables~\citep[][Chapter IX]{Paley1934}. This approach was further generalised by Paul L\'{e}vy and Zbigniew Ciesielski for any orthonormal basis of the Hilbert space of square integrable functions $L^2([0, 1])$. This is known as the L\'{e}vy--Ciesielski's construction~\citep{Karatzas1991}. For more comprehensive reviews on the existence/construction of Wiener processes, see, for example,~\citet[][Chapter 3]{ReneBrownianBook2012} or~\citet{Morters2010book}.

Definition~\ref{def:wiener-process} defines scalar-valued Wiener processes. In order to generalise Wiener processes to $\R^w$, it is common to think of $\R^w$-valued Wiener processes as vectors that are a collection of $w$ mutually independent Wiener processes~\citep[][Definition 18.5]{Koralov2007}. As for function-valued Wiener processes, such as $Q$-Wiener processes\footnote{The cover of the thesis illustrates a realisation of a $Q$-Wiener process taking value in a Sobolev space with homogenous Dirichlet boundary condition.}, the generalisation often leverages infinite-dimensional Gaussian measures~\citep{Kuo1975Book,Bogachev1998, Giuseppe2014, Lord_powell_shardlow_2014}.

The key ingredient to defining solutions of SDEs are stochastic integrals of the form
\begin{equation}
	\int^t_{t_0} b(s, \omega) \diff W(s, \omega),
	\label{equ:ito-integral}
\end{equation}
where $b$ is any suitable adapted process in the sense that $\omega \mapsto b(t, \omega)$ is measurable with respect to a filtration to which the Wiener process is adapted~\citep[][Chapter 4]{Kuo2006Book}. However, due to the fact that $t \mapsto W(t)$ has infinite first order variation almost surely~\citep[][Chapter 3]{Oksendal2003}, one cannot define the integral above in the classical Stieltjes sense. There exist multiple interpretations of such stochastic integral, and the two most popular constructions are due to \citet{Ito1944} and \citet{Stratonovich1966}. In It\^{o}'s construction, this leads to an integral being a (local) martingale with respect to the filtration that $W$ is adapted to~\citep{Kuo2006Book}. In particular, when the integrand $b$ does not depend on $\omega$ (i.e., is non-random), the integral~\eqref{equ:ito-integral} reduces to a Gaussian process~\citep{Kuo2006Book}.
\begin{remark}
	This thesis is only concerned with It\^{o}'s construction of stochastic integrals.
\end{remark}
The multidimensional extension of It\^{o} integrals is defined as follows. Suppose that $W$ is a $w$-dimensional Wiener process, and $b$ is an $\R^{d\times w}$-valued process, then the $i$-th element of a $d$-dimensional It\^{o} integral is defined as
\begin{equation}
	\sum^w_{j=1} \int^t_{t_0} b_{ij}(s,\omega) \diff W_{j}(s,\omega),
\end{equation}
where $ij$ and $j$ above stand for the usual element selection notations \citep[][Page 283]{Karatzas1991}.

With It\^{o} integrals defined, we can then formally interpret SDEs. Consider a $w$-dimensional Wiener process $W$ and a stochastic process $X\colon \T\to\R^d$ that satisfies the stochastic integral equation (SIE)
\begin{equation}
	\begin{split}
		X(t) &= X(t_0) + \int^t_{t_0} a(X(s),s) \diff s + \int^t_{t_0} b(X(s), s) \diff W(s),\\
		X(t_0) &= X_0,
	\end{split}
	\label{equ:SIE}
\end{equation}
on some probability space. The differential shorthand 
\begin{equation}
	\begin{split}
		\diff X(t) &= a(X(t), t) \diff t + b(X(t), t)\diff W(t), \\
		X(t_0) &= X_0,
	\end{split}
	\label{equ:SDE}
\end{equation}
of the SIE in Equation~\eqref{equ:SIE} is called a \emph{stochastic differential equation}. The SDE coefficients $a\colon \R^d\times\T\to \R^d$ and $b\colon\R^d\times\T \to \R^{d\times w} $ are called the \textit{drift} and \textit{dispersion} functions, respectively. 

\subsection{Existence and uniqueness of SDEs solutions}
\label{sec:sde-solution-existence-markov}
One fundamental question is whether an SDE admits a solution and, if so, what the properties (e.g., uniqueness and continuity) of the solution(s) are. In literature, the solution analysis of SDEs is usually described in the sense of strong or weak solutions. In this thesis we are mostly concerned with strong solutions that we detail in the following definition.
\begin{definition}[Strong solution]
	\label{def:strong-solution}
	Let $(\Omega, \FF, \FF_t, \mathbb{P})$ be a filtered probability space, $W\colon \T \to\R^w$ be a $w$-dimensional Wiener process defined on this space, and let $X_0\in\R^d$ be a random variable independent of $W$. Also let $\FF_t^W$ be the filtration generated by $W(t)$ and $X_0$. Then a continuous process $X\colon\T\to\R^d$ is said to be a strong solution of the SDE~\eqref{equ:SDE} if the following four conditions are satisfied.
	\begin{enumerate}[I.]
		\item $X(t)$ is adapted to $\FF_t^W$.
		\item $\PP$-almost surely $X(t)$ solves Equation~\eqref{equ:SIE} for all $t\in\T$.
		\item $\PP$-almost surely $\int^t_{t_0} \abs*{a_{i}(X(s), s)} + (b_{ij}(X(s), s))^2 \diff s < \infty$ holds for all $i=1,2,\ldots, d$, $j=1,2,\ldots, w$, and $t\in\T$.
		\item $\PP$-almost surely $X(t_0) = X_0$.
	\end{enumerate}
\end{definition}
The above definition is found in~\citet[][Definition 2.1]{Karatzas1991} or~\citet[][Section 10.4]{Chung1990}, but for simplicity, here we omit to augment $\FF_t^W$ with the null sets of $\Omega$. This definition means that if we are given a probability space which carries $W$ and $X_0$, the solution $X(t)$ must be adapted to the generated filtration $\FF_t^W$. In other words, $W$ and $X_0$ should completely characterise $X$, and one can write the strong solution as a function of $W$ and $X_0$ only. 

The third condition in Definition~\ref{def:strong-solution} is important to keep in mind as it makes the solutions continuous semimartingales~\citep{Chung1990, Williams2000Vol2}.

The notion of strong solution might not always be useful because the condition of being adapted to the generated filtration is sometimes too strict. For example, in Tanaka's equation~\citep[][Example 5.3.2]{Oksendal2003}, one cannot find such an $\FF_t^W$-adapted solution therefore, the equation does not admit a strong solution. To relax this restriction, we can allow flexibility on the Wiener process, and seek pairs $(X, W)$ solutions of the SDE~\eqref{equ:SDE}, instead of simply seeking $X$~\citep{Chung1990, Oksendal2003}. Such pairs are called weak solutions and are closely related to the martingale problem~\citep{Stroock1969, Strock1979, Williams2000Vol2}. Moreover, strong solutions are weak solutions but the converse is not true. However, since this thesis is not concerned with weak solutions, we refer the reader to, for example, \citet[][]{Chung1990} or~\citet[][Definition 3.1]{Karatzas1991} for technical expositions of these. 

\begin{remark}
	\label{remark:a-solution-of-sde}
	In the remainder of this thesis, unless mentioned otherwise, we will be solely concerned with strong solutions of SDEs (although some results may hold in the weak sense too). Moreover, strong solutions of SDEs will be referred to as It\^{o} processes.
\end{remark}

Pathwise and weak uniqueness of SDE solutions are defined as follows (see, \citealp[][Chapter 5.3]{Karatzas1991} or \citealp[][Page 247]{Chung1990}).
\begin{definition}[Pathwise uniqueness]
	\label{def:pathwise-unique}
	The pathwise uniqueness holds for the SDE in Equation~\eqref{equ:SDE} if for all solutions $\overline{X}$ and $\widetilde{X}$ that share the same probability space, Wiener process, and initial condition, we have
	\begin{equation}
		\PP\big(\big\lbrace \omega \colon  \absbig{\overline{X}(t) - \widetilde{X}(t)} = 0,~\mathrm{for~all~} t\in\T \big\rbrace\big) = 1.
		\label{equ:indistinguishable}
	\end{equation}
\end{definition}
Notice that the ``for all $t\in\T$'' condition in Equation~\eqref{equ:indistinguishable} can be moved outside of the probability because $\big\lbrace\omega\colon\absbig{\overline{X}(t) - \widetilde{X}(t)} = 0,~\mathrm{for~all~} t\in\T \big\rbrace$ includes $\big\lbrace\omega\colon\absbig{\overline{X}(t) - \widetilde{X}(t)} = 0\big\rbrace$ for all $t\in\T$, and the converse is true as well due to the continuity of the solutions.
\begin{definition}[Weak uniqueness]
	\label{def:weakly-unique}
	The weak uniqueness holds for the SDE in Equation~\eqref{equ:SDE}, if all solutions are identical in law. 
\end{definition}
%
Furthermore, a classical result by~\citet{Yamada1971} shows that the pathwise uniqueness implies the weak uniqueness. 


\subsection{Markov property of SDE solutions}
\label{sec:markov-property}
One of the main purposes of using SDEs is to construct continuous-time Markov processes. Hence, it is necessary to examine if solutions of SDEs admit the Markov property defined as follows. 
\begin{definition}[Markov process]
	\label{def:markov-process}
	Let $\FF_t$ be a given filtation on $\Omega$, and $X(t)$ be an $\FF_t$-adapted process. Then $X$ is said to be a Markov process (with respect to $\FF_t$) if 
	\begin{equation}
		\expec{\varphi(X(t + s)) \cond \FF_t} = \expec{\varphi(X(t+s)) \cond X(t)}
	\end{equation}
	for every $t\in\T, s\in\R_{\geq 0}$, and bounded Borel measurable function $\varphi$.
\end{definition}
%
%
It can be shown that It\^{o} processes are indeed Markov processes. Proofs can be found, for example, in~\citet[][Theorem 7.1.2]{Oksendal2003}, \citet[][Theorem 10.6.2]{Kuo2006Book},~\citet[][Theorem 18.13]{ReneBrownianBook2012}, \citet[][Theorem 8.6]{LeGall2016}, and~\citet[][Lemma 10.10]{Chung1990}. 

\begin{remark}
    Thanks to the martingale-problem method~\citep{Stroock1969}, the Markov property for SDEs can be proved in more general context than strong solutions of SDEs, if the associated martingale problem is well-posed. For details, see, for example, \citet[][Theorem 21.1]{Williams2000Vol2, Ethier1986}.
\end{remark}

The Markov property is useful in the sense that it allows for predicting the future given some past information (i.e., $\FF_t$) by only using the present (i.e., $X(t)$). This feature makes many applications -- such as Bayesian filtering and smoothing~\citep{Sarkka2013} -- computationally efficient. To see this, let $p_{X(t_1),\ldots, X(t_k), X(t_{k+1})}(x_1, \ldots,\allowbreak x_k, x_{k+1})$ be the finite-dimensional probability density function of $\big\lbrace X(t_1), \ldots,\allowbreak  X(t_k), X(t_{k+1})\big\rbrace$ for any integer $k\geq 1$ and $t_1\leq \cdots\leq t_k\leq t_{k+1}\in\T$. The Markov property implies that
\begin{equation}
	p_{X(t_{k+1}) \cond X(t_k), \ldots, X(t_1)}(x_{k+1} \cond x_k, \ldots, x_1) = p_{X(t_{k+1}) \cond X(t_k)}(x_{k+1} \cond x_k)
\end{equation}
and
\begin{equation}
	p_{X(t_k) \cond X(t_i)}(x_k \cond x_i) = \int p_{X(t_k) \cond X(t_j)}(x_k \cond x_j) \, p_{X(t_j) \cond X(t_i)}(x_j \cond x_i) \diff x_j
	\label{equ:chapman-kol}
\end{equation}
hold for every $t_i \leq t_j \leq t_k\in\T$. The conditional density $p_{X(t_{k+1}) \cond X(t_k)}(x_{k+1} \cond x_k)$ and Equation~\eqref{equ:chapman-kol} are known as the transition probability density function and the Chapman--Kolmogorov equation, respectively. In particular, the Chapman--Kolmogorov equation means that the joint probability density function of a Markov process at times $t_1,\ldots,t_k$ factorises with respect to its transition densities. Therefore, one can compute Markov processes marginal distributions sequentially with linear complexity in time. This is particularly useful in the context of Bayesian filtering and smoothing which will be the subject of Section~\ref{sec:cd-smoothing}. 

\subsection{It\^{o}'s formula}
\label{sec:ito-formula}
Suppose that $X\colon\T\to\R$ is a deterministic smooth function, and that $\phi\in\mathcal{C}^1(\R;\R)$ is another smooth function. Then by Newton--Leibniz formula/chain rule, we have
\begin{equation}
	\phi(X(t)) = \phi(X(t_0)) + \int^t_{t_0}\tash{\phi}{X}\tash{X}{t}(s) \diff s.\nonumber
\end{equation}
Unfortunately the rule above does not generally hold when $X$ is a stochastic process. As an example, if $X$ is a Wiener process then the derivative $\partial X \, / \, \partial t$ does not exist in the usual limit definition~\citep[][Chapter 14]{ReneBrownianBook2012}.

The differentiation rule for continuous semimartingales is given by the so-called It\^{o}'s formula~\citep[see, e.g.,][Theorem 5.10]{LeGall2016}. In the special case when $X$ is an It\^{o} process, It\^{o}'s formula takes the following form.

\begin{theorem}[It\^{o}'s formula]
	\label{thm:ito-formula}
	Let $\phi\colon \R^d \times \T \to \R$ be a function that is twice-differentiable in the first argument and differentiable in the second argument. Suppose that $X\colon \T\to\R^d$ is an It\^{o} process solving the SDE in Equation~\eqref{equ:SDE}, then
	\begin{equation}
		\begin{split}
			\phi(X(t), t) &= \phi(X(t_0), t_0)  + \int^t_{t_0} \tash{\phi}{t}(X(s), s) \diff s \\
			&\quad+ \int^t_{t_0}  (\nabla_X \phi)^\trans \, a(X(s), s) + \frac{1}{2} \, \trace*{\Gamma(X(s), s) \, \hessian_X\phi  } \diff s \\
			&\quad+ \int^t_{t_0} (\nabla_X \phi)^\trans \, b(X(s), s) \diff W(s), 
		\end{split}
	\end{equation}
	where $\Gamma(X(s), s) \coloneqq b(X(s), s) \, b(X(s), s)^\trans$, and $\nabla$ and $\hessian$ denote the gradient and Hessian operators, respectively. Moreover, $t \mapsto \phi(X(t), t)$ is also an It\^{o} process.
\end{theorem}

\section{Continuous-discrete filtering and smoothing}
\label{sec:cd-smoothing}
In this section, we review Bayesian filtering and smoothing algorithms for continuous-discrete state-space models~\citep{Jazwinski1970, Sarkka2013, Sarkka2019}. 

\subsection{Continuous-discrete state-space models}
\label{sec:cd-ss-model}
Consider a system 
\begin{equation}
	\begin{split}
		\diff X(t) &= a(X(t), t) \diff t + b(X(t), t) \diff W(t), \quad X(t_0) = X_0, \\
		Y_k &= h(X_k) + \xi_k, \quad \xi_k \sim \mathrm{N}(0, \Xi_k),
	\end{split}
	\label{equ:cd-model}
\end{equation}
where $X\colon\T\to\R^d$, $X_k \coloneqq X(t_k)$, $Y_k\in\R^{d_y}$, $\Xi_k\in\R^{d_y \times d_y}$, and $h\colon \R^d \to \R^{d_y}$. Models represented by the combination of an SDE and a discrete-time measurement model as per Equation~\eqref{equ:cd-model} are called \textit{continuous-discrete state-space models}, or simply continuous-discrete models. These are ubiquitous in physics and engineering (see, e.g., Example~\ref{example:tme-ct-tracking} for manoeuvring target tracking). We call $X_k$ and $Y_k$ the \textit{state} and \textit{measurement}, respectively, of $X(t_k)$ at $t_k$.

Let $Y_{1:T} = \lbrace Y_k \colon k=1,2,\ldots,T \rbrace$ be a collection of measurement variables and $y_{1:T} = \lbrace y_k \colon k=1,2,\ldots,T\rbrace$ be the corresponding data at times $t_1 \leq t_2 \leq \cdots \leq t_T \in \T$. The continuous-discrete filtering and smoothing problem for model~\eqref{equ:cd-model} aims at solving the filtering posterior marginal densities
\begin{equation}
	p_{X_k \cond Y_{1:k}}(x_k \cond y_{1:k})
	\label{equ:cd-filtering}
\end{equation}
and the smoothing posterior marginal densities
\begin{equation}
	p_{X_k \cond Y_{1:T}}(x_k \cond y_{1:T}),
	\label{equ:cd_smoothing}
\end{equation}
for $k=1,2,\ldots, T$~\citep{Sarkka2019}. Although in principle the filtering and smoothing problems aim at more general posterior densities (i.e., $p_{X(t) \cond Y_{1:T}}(x, t \cond y_{1:T})$ for all $t\in\T$), for the sake of simplicity of exposition, we restrict ourselves to estimating the marginal filtering and smoothing distribution at the data points $\lbrace t_k \colon k=1,2,\ldots, T \rbrace$ only. 

Since solutions of SDEs are Markov processes, we can use the Markov property (see, Section~\ref{sec:markov-property}) to sequentially solve the filtering and smoothing posterior densities for $k=1,2,\ldots,T$~\citep{Sarkka2013}. To see this, suppose that the filtering density $p_{X_{k-1} \cond Y_{1:k-1}}(x_{k-1} \cond y_{1:k-1})$ at $t_{k-1}$ is known\footnote{We define $p_{X_{0} \cond Y_{1:0}}(x_{0} \cond y_{1:0}) \coloneqq p_{X_0}(x_0)$ at $t_0$.}. Then by leveraging Bayes' rule, the filtering density at $t_k$ reads
\begin{equation}
	p_{X_k \cond Y_{1:k}}(x_k \cond y_{1:k}) = \frac{p_{Y_k \cond X_k}(y_k \cond x_k) \, p_{X_k \cond Y_{1:k-1}}(x_k \cond y_{1:k-1})}{\int p_{Y_k \cond X_k}(y_k \cond x_k) \, p_{X_k \cond Y_{1:k-1}}(x_k \cond y_{1:k-1}) \diff x_k},
	\label{equ:cd-filtering-bayes}
\end{equation}
where the predictive density
\begin{equation}
	\begin{split}
		&p_{X_k \cond Y_{1:k-1}}(x_k \cond y_{1:k-1}) \\
		&= \int p_{X_k \cond X_{k-1}}(x_k \cond x_{k-1}) \, p_{X_{k-1} \cond Y_{1:k-1}}(x_{k-1} \cond y_{1:k-1}) \diff x_{k-1}
	\end{split} \label{equ:cd-filtering-predictive}
\end{equation}
needs to be computed by propagating $p_{X_{k-1} \cond Y_{1:k-1}}(x_{k-1} \cond y_{1:k-1})$ through the SDE.
One can then obtain the filtering densities sequentially for $k=1,2,\ldots, T$ starting from a known/given initial condition. 

The smoothing densities are solved backward for $k=T,\ldots, 1$ by using the filtering results. Suppose that the smoothing density $p_{X_{k+1} \cond Y_{1:T}} \allowbreak (x_{k+1} \cond y_{1:T})$ at $t_{k+1}$ is known, then again by Bayes' rule~\citep{Kitagawa1987, Sarkka2013}, the smoothing density at $t_{k}$ is
\begin{equation}
	\begin{split}
		&p_{X_k \cond Y_{1:T}}(x_k \cond y_{1:T}) \\
		&= p_{X_k \cond Y_{1:k}}(x_k \cond y_{1:k})  \int \frac{p_{X_{k+1} \cond X_k}(x_{k+1} \cond x_k) \, p_{X_{k+1} \cond Y_{1:T}}(x_{k+1} \cond y_{1:T})}{p_{X_{k+1} \cond Y_{1:k}}(x_{k+1} \cond y_{1:k})} \diff x_{k+1}.
		\label{equ:cd-smoothing-bayes}
	\end{split}
\end{equation}

Unfortunately, for non-linear state-space models, Equations~\eqref{equ:cd-filtering-bayes}, \eqref{equ:cd-filtering-predictive}, and~\eqref{equ:cd-smoothing-bayes} are rarely solvable in closed-form. In practice, one often needs to use approximation schemes, such as Taylor expansion, numerical integration, or particle approximations~\citep{Sarkka2013}. However, if the SDE and measurement model happen to be linear (and also starting from a Gaussian initial condition), then the filtering and smoothing densities are exactly Gaussian and their means and covariances can be computed in closed-form sequentially. This is known as the (continuous-discrete) Kalman filtering and Rauch--Tung--Striebel smoothing~\citep{Sarkka2019}, the details of which are given in the next section. 

\subsection{Rauch--Tung--Striebel smoothing}
\label{sec:rts}
Consider a linear continuous-discrete model
\begin{equation}
	\begin{split}
		\diff X(t) &= A(t) \, X(t) \diff t + B(t) \diff W(t), \quad X(t_0) = X_0, \\
		y_k &= H_k \, X_k + \xi_k, \quad \xi_k \sim \mathrm{N}(0, \Xi_k),
		\label{equ:cd-linear}
	\end{split}
\end{equation}
where $X_0 \sim \mathrm{N}(m_0, P_0)$ is a Gaussian random variable of mean $m_0$ and covariance $P_0$. Here the coefficients $A\colon \T \to \R^{d\times d}$, $B\colon \T \to \R^{d\times w}$, and $H_k\in \R^{d_y \times d}$ are deterministic matrix-valued functions and a constant, respectively. In this case, the filtering and smoothing densities in Equations~\eqref{equ:cd-filtering-bayes} and~\eqref{equ:cd-smoothing-bayes} can be solved exactly by using Kalman filters and Rauch--Tung--Striebel (RTS) smoothers as follows~\citep[cf.][]{Sarkka2019}.

\begin{algorithm}[Continuous-discrete Kalman filter and RTS smoother]
	\label{alg:kfs}
	Let $p_{X_k \cond Y_{1:k}}(x_k \cond y_{1:k}) = \mathrm{N}(x_k\cond m^f_k, P^f_k)$ and $p_{X_k \cond Y_{1:T}}(x_k \cond y_{1:T}) = \mathrm{N}(x_k \cond m^s_k, P^s_k)$ be the Gaussian parametrisations of the filtering and smoothing posterior densities, respectively, at $t_k$. Also let $m_0^f \coloneqq m_0$ and $P_0^f \coloneqq P_0$ at $t_0$. At each step for $k=1,2,\ldots,T$, the Kalman filter first obtains the predictive density $\mathrm{N}(x_k \cond m^-_k, P^-_k)$  by solving the system of ordinary differential equations (ODEs)
	\begin{equation}
		\begin{split}
			\frac{\diff m(t)}{\diff t} &= A(t) \, m(t), \\
			\frac{\diff P(t)}{\diff t} &= A(t) \, P(t) + P(t) \, A(t)^\trans + B(t) \, B(t)^\trans,
			\label{equ:linear-sde-moment-ode}
		\end{split}
	\end{equation}
	at $t_k$, starting from the initial values $m^f_{k-1}$ and $P^f_{k-1}$ at time $t_{k-1}$. Then, it updates the predictive density to get the filtering posterior mean $m^f_k$ and covariance $P^f_k$ at time $t_k$ by computing
	\begin{equation}
		\begin{split}
			K_k &= P^-_k \, H_k^\trans \, (H_k \, P^-_k \, H_k^\trans + \Xi_k)^{-1},\\
			m^f_k &= m^-_k + K_k \, (y_k - H_k \, m^-_k), \\
			P^f_k &= P^-_k - K_k \, (H_k \, P^-_k \, H_k^\trans + \Xi_k) \, K_k^\trans.
		\end{split}
		\label{equ:kfs-update}
	\end{equation}
	Let $m^s_T \coloneqq m^f_T$ and $P^s_T \coloneqq P^f_T$. At each step for $k=T-1, \ldots, 1$, the RTS smoother computes $m^s_k$ and $P^s_k$ at $t_k$ by solving the system of ODEs
	\begin{equation}
		\begin{split}
			\frac{\diff m(t)}{\diff t} 
			    &= A(t) \, m(t) + B(t) \, B(t)^\trans \, \big(P^f(t)\big)^{-1} \, \big(m(t) - m^f(t)\big), \\
			\frac{\diff P(t)}{\diff t} 
			    &= \Big[A(t) + B(t) \, B(t)^\trans \, \big(P^f(t)\big)^{-1} \Big] \, P(t) \\
			    &\quad+ P(t) \, \Big[A(t) + B(t) \, B(t)^\trans \, \big(P^f(t)\big)^{-1} \Big]^\trans - B(t) \, B(t)^\trans,
		\end{split}
	\end{equation}
	starting from the initial values $m^s_{k+1}$ and $P^s_{k+1}$ at time $t_{k+1}$, where $m^f(t)$ and $P^f(t)$ stand for the filtering mean and covariance at time $t$, respectively.
\end{algorithm}

Furthermore, if the SDE coefficients in Equation~\eqref{equ:cd-linear} do not depend on time (i.e., $A$ and $B$ are constant matrices), then the continuous-discrete filtering and smoothing problem can be reformulated in an equivalent discrete-discrete problem of the form
\begin{equation}
	\begin{split}
		X_k &= F_{k-1} \, X_{k-1} + q_{k-1}, \\
		Y_k &= H_k \, X_k + \xi_k,
	\end{split}
	\label{equ:disc-linear}
\end{equation}
where $q_{k-1} \sim \mathrm{N}(0, Q_{k-1})$. The coefficients $F_{k-1} \in \R^{d \times d}$ and $Q_{k-1} \in \R^{d \times d}$ are in turn determined by
\begin{equation}
	\begin{split}
		F_{k-1} &= e^{(t_k - t_{k-1}) \, A}, \\
		Q_{k-1} &= \int^{t_k}_{t_{k-1}} e^{(t_{k} - s) \, A} \,B \, B^\trans \, \big(e^{(t_{k} - s) \, A}\big)^\trans \diff s.
	\end{split}
	\label{equ:disc-coeff-exp}
\end{equation}
Provided one can numerically compute Equations~\eqref{equ:disc-coeff-exp}~\citep[see, e.g.,][for how to do so in practice]{Axelsson2015, Sarkka2019}, one can then apply standard Kalman filters and RTS smoothers~\citep[][Theorems 4.2 and 8.2]{Sarkka2013} to the discretised state-space model.

\subsection{Gaussian approximate smoothing}
\label{sec:gaussian-filter-smoother}
In this section, we review the Gaussian approximated density filtering and smoothing for non-linear continuous-discrete state-space models~\citep{Kazufumi2000, SimoGFS2013}. The idea of Gaussian filtering and smoothing is to approximate the filtering and smoothing densities by
\begin{equation}
	\begin{split}
		p_{X_k \cond Y_{1:k}}(x_k \cond y_{1:k}) &\approx \mathrm{N}\big(x_k \cond m^f_k, P^f_k\big), \\
		p_{X_k \cond Y_{1:T}}(x_k \cond y_{1:T}) &\approx \mathrm{N}\big(x_k \cond m^s_k, P^s_k\big).
	\end{split}
	\label{equ:gfs-posteriors}
\end{equation}
Then, by applying Gaussian identities, the general Bayesian filtering and smoothing formulations in Equations~\eqref{equ:cd-filtering-bayes} and~\eqref{equ:cd-smoothing-bayes} admit closed-form approximations. We therefore have the following algorithm~\citep[cf.][Chapter 10]{Sarkka2019}.
\begin{algorithm}[Continuous-discrete Gaussian filter and smoother]
	\label{alg:gfs}
	Let $p_{X_k \cond Y_{1:k}}(x_k \cond y_{1:k}) \approx \mathrm{N}\big(x_k \cond m^f_k, P^f_k\big)$ and $
	p_{X_k \cond Y_{1:T}}(x_k \cond y_{1:T}) \approx \mathrm{N}\big(x_k \cond m^s_k, P^s_k\big)$ be approximate filtering and smoothing densities. Also consider a Gaussian approximation to the initial density $p_{X_0}(x_0) \approx \mathrm{N}(x_0\cond m_0, P_0)$. The Gaussian filter obtains $\big\lbrace m_k^f, P^f_k\colon k=1,2,\ldots, T \big\rbrace$ by computing the following prediction and update steps sequentially for $k=1,2,\ldots, T$.
	\begin{enumerate}[I.]
		\item Prediction:
		\begin{equation}
			\begin{split}
				m^-_k &= \int x_k \, p_{X_k \cond Y_{1:k-1}}(x_k \cond y_{1:k-1}) \diff x_k, \\
				P^-_k &= \int (x_k - m_k^-) \, (x_k - m_k^-)^\trans \, p_{X_k \cond Y_{1:k-1}}(x_k \cond y_{1:k-1}) \diff x_k.
			\end{split}
		\end{equation}
		\item Update:
		\begin{equation}
			\begin{split}
				S_k &= \expecBig{\big(h(X_k) - \expec{h(X_k)}\big) \, \big(h(X_k) - \expec{h(X_k)}\big)^\trans} + \Xi_k,\\
				K_k &= \expecBig{\big(X_k - m_k^-\big) \, \big(h(X_k) - \expec{h(X_k)} \big)^\trans} \, S_k^{-1}, \\
				m^f_k &= m^-_k + K_k \, (y_k - \expec{h(X_k)}), \\
				P^f_k &= P^-_k - K_k \, S_k \, K_k^\trans.
			\end{split}
		\end{equation}
		Note that the expectations above are taken with respect to the predictive density $p_{X_k \cond Y_{1:k-1}}(x_k \cond y_{1:k-1})$. In addition, if the measurement model is linear, then the update step above reduces to Equation~\eqref{equ:kfs-update}.
	\end{enumerate}
	Let $m^s_T \coloneqq m^f_T$ and $P^s_T \coloneqq P^f_T$. The Gaussian smoother obtains $\big\lbrace m_k^s, P^s_k\colon k=1,2,\ldots, T-1 \big\rbrace$ by sequentially computing
	\begin{equation}
		\begin{split}
			D_{k+1} &= \cov*{X_k, X_{k+1}^\trans \cond y_{1:k}}, \\
			G_k &= D_{k+1} \, \big( P^-_{k+1} \big)^{-1}, \\
			m^s_k &= m_k^f + G_k \, (m^s_{k+1} - m^-_{k+1}), \\
			P^s_k &= P_k^f + G_k \, (P^s_{k+1} - P^-_{k+1}) \, G_k^\trans,
		\end{split}
	\end{equation}
	for $k=T-1, T-2,\ldots, 1$.
\end{algorithm}

In order to compute the integrals/expectations in Algorithm~\ref{alg:gfs}, it is often necessary to approximate the transition density by
\begin{equation}
	p_{X_{k} \cond X_{k-1}}(x_k \cond x_{k-1}) \approx \mathrm{N}\big(x_k \cond \expec{X_k \cond X_{k-1}}, \cov{X_k \cond X_{k-1}}\big).
	\label{equ:gfs-transition}
\end{equation}
There are various approaches to approximate the mean and covariance in the transition density above. One popular approach is linearising the SDE (or its discretisation) by using, for example, Taylor expansions. This leads to (continuous-discrete) extended Kalman filters and smoothers~\citep{Jazwinski1970}. Another commonly used approach is to solve the ODEs (see, e.g., Equation~\eqref{equ:mean-ode}) that characterise the mean and covariance functions of the SDE~\citep{Sancho1970, Jazwinski1970, Maybeck1982, SimoGFS2013}. However this ODE approach requires to compute expectations with respect to the probability measure of SDEs, which in practice requires further approximation schemes (such as Monte Carlo). 

We can also approximate the SDE by a Gaussian increment-based discretisation defined as
\begin{equation}
	X_k \approx f_{k-1}(X_{k-1}) + q_{k-1}(X_{k-1}), \label{equ:approx-transition-model}
\end{equation}
where $q_{k-1}(X_{k-1}) \sim \mathrm{N}(0, Q_{k-1}(X_{k-1}))$. In particular, $\expec{X_k \cond X_{k-1}} \approx f_{k-1}(X_{k-1})$ and $\cov{X_k \cond X_{k-1}} \approx Q_{k-1}(X_{k-1})$. The choice of the functions $f_{k-1} \colon \R^{d} \to \R^d$ and $Q_{k-1} \colon \R^{d} \to \R^{d \times d}$ depends on the discretisation method  used for the approximation. 
\begin{example}
For instance, the Euler--Maruyama scheme gives
\begin{equation}
	\begin{split}
		f_{k-1}(X_{k-1}) &= X_{k-1} + a(X_{k-1}, t_{k-1}) \, (t_k - t_{k-1}), \\
		Q_{k-1}(X_{k-1}) &= b(X_{k-1}, t_{k-1}) \, b(X_{k-1}, t_{k-1})^\trans \, (t_k - t_{k-1}).
	\end{split}
	\label{equ:gfs-euler-maruyama}
\end{equation}
Furthermore, in Chapter~\ref{chap:tme} we illustrate a Taylor moment expansion-based approach generalising the Euler--Maruyama scheme for approximating the transition coefficients in Equation~\eqref{equ:gfs-transition}. 
\end{example}
Recall that the expectations in Algorithm~\ref{alg:gfs} are usually hard to compute exactly for non-linear models. However, we can use quadrature methods, for example, Gauss--Hermite quadrature~\citep{Davis1984, Arasaratnam2007}, unscented transform~\citep{Julier2004}, spherical cubature~\citep{Cubature2009, Simo2012}, or sparse-grid quadratures~\citep{Jia2012, Radhakrishnan2016} to compute them numerically. 

\subsection{Non-Gaussian approximate smoothing}
\label{sec:other-filters-smoothers}
Despite the simplicity and efficiency of Gaussian approximated filtering and smoothing, these might lead to poor approximations for densities that are, for example, multi-modal or skewed~\citep{Sarkka2013}. Moreover,~\citet{Zhao2020SSDGP} show that, for many SS-DGPs constructions, the Kalman gain (i.e., $K_k$ in Algorithm~\ref{alg:gfs}) of Gaussian approximated filters and smoothers converge to zero as $t\to\infty$. This can be problematic as a zero Kalman gain means that no further information from data is used for updating the posterior distributions. This issue is detailed in Section~\ref{sec:identi-problem}. Hence, the aim of this section is to briefly review some other non-linear filters and smoothers that could be useful for solving the continuous-discrete model in Equation~\eqref{equ:cd-model} without relying on Gaussian approximations.  

One way to compute the general filtering and smoothing densities is by using sequential Monte Carlo (SMC) methods~\citep{Chopin2020}. This class of methods considers Monte Carlo approximations of the integrals in Equations~\eqref{equ:cd-filtering-bayes} and~\eqref{equ:cd-smoothing-bayes} instead of Gaussian quadrature ones. They sequentially propose new Monte Carlo samples that they then weight via a potential function, and use a resampling step in order to keep weight distribution non-degenerate~\citep{Docet2000, Simon2004, Christophe2010}. These result in two generic classes of methods called particle filters and particle smoothers retaining linear complexity at the cost of losing the closed-form interpretation. These methods can be customised to the problem at hand so as to provide better approximations of the distributions~\citep{Chopin2020}. In particular, in the context of SS-DGPs, \citet{Zhao2020SSDGP} show that they result in a better approximation of the posterior density for regression problems such as the rectangular signal in Figure~\ref{fig:gp-fail}. However, parameter learning in particle filters can be problematic, as the resampling procedure, in general, makes their loss functions non-differentiable. This can be addressed by using smooth resampling methods, such as the one in~\citet{Corenflos2021OT}.

Another way to compute the filtering and smoothing densities is to think of them as solutions of ODEs/partial differential equations (PDEs). These connections are well-known for continuous-continuous state-space models (i.e., where instead of the discrete measurements in Equation~\eqref{equ:cd-model} we have a continuous measurement modelled as an SDE depending on the state), such as the Kalman--Bucy filter~\citep{KalmanBucy1961} for linear models. More generally, for non-linear continuous state-space models, the filtering density~\citep{Kushner1964, Zakai1969, Bain2009, Sarkka2013} is governed by the Kushner--Stratonovich equation or Zakai's equation\footnote{Note that Zakai's equation gives \emph{unnormalised} filtering densities.}. For the PDEs that characterise the continuous smoothing solutions, see, for example,~\citet[][Algorithm 10.30]{Sarkka2019} or~\citet{Anderson1972}.

Analogously to the continuous filtering and smoothing, it is also possible to obtain continuous-discrete posterior densities by solving certain PDEs or ODEs. For example, \citet{Jazwinski1970, Beard1999, Challa2000} show that one can combine the Fokker--Planck--Kolmogorov equation and Bayes' rule in order to compute the filtering solution. More specifically, Fokker--Planck--Kolmogorov equation is used to predict the state in Equation \eqref{equ:cd-filtering-predictive}, while Bayes' rule is then used to update the predicted state into the filtered state as per the filtering formulation in Equation~\eqref{equ:cd-filtering-bayes}. In a different flavour, \citet{Brigo1998, Koyama2018} consider the projection filter and smoother, which consist in projecting the filtering and smoothing solutions (of certain families of probability densities) on the space of their density parameters (e.g., the natural parameters of the exponential family). This transforms the problem in a system of ODEs in their density parameters that one then can solve instead of solving the original problem.

\citet{Archambeau2007, Archambeau2008, XuechenLi2020} show that one can also approximate the filtering/smoothing solution by another SDE. The idea is to use a parametrised SDE~\citep[e.g., a linear SDE is used in][]{Archambeau2007, Archambeau2008} to approximate the filtering/smoothing solution and learn the SDE parameters by minimising the Kullback--Leibler (KL) divergence from the true filtering/smoothing distribution. Once this approximate SDE is learnt, the statistical properties (e.g., mean or covariance) of the filtering/smoothing solution can be computed in closed-form from the approximate linear SDE or by simulating trajectories from the approximate SDE (if the SDE is non-linear). Recall that solutions of SDEs are Markov processes. This SDE-based variational filtering/smoothing method is indeed reasonable in the sense that the optimal variational distribution (among a family of parametric variational distributions) for minimising the KL divergence admits the Markov property as shown in~\citet[][Lemma 1]{Courts2021}.

For more comprehensive reviews of non-linear filtering and smoothing methods, we refer the reader to, for example, \citet{Jazwinski1970, Maybeck1982, Sarkka2013, Bain2009, Law2015, Evensen2009, Doucet2001, Sarkka2019}.

\section{Some theorems}
\label{sec:useful-theorem}
For the sake of self-containedness, in this section we list several intermediate results that will be used in the course of the thesis.

\begin{theorem}[Cauchy product]
	\label{thm:cauchy-product}
	Let $\sum^\infty_{i=0} \alpha_i \, x^i$ and $\sum^\infty_{i=0} \beta_i \, x^i$ be two power series of $x$ with convergence radius $D_{\alpha} > 0$ and $D_{\beta} > 0$ . Then their product is a power series
	\begin{equation}
		\Big(\sum^\infty_{i=0} \alpha_i \,x^i \Big)\, \Big(\sum^\infty_{i=0} \beta_i \,x^i \Big) = \sum^\infty_{k=0} \Big(\sum^k_{j=0} \alpha_j\, \beta_{k-j}\Big)\, x^k
	\end{equation}
	on an open disk of radius $D \geq \min(D_{\alpha}, D_{\beta})$~\citep[see, e.g.,][Theorem 2.37]{Canuto2015}.
\end{theorem}
We use the Cauchy product in Theorem~\ref{thm:tme-cov-pd} to truncate the product of two finite power series. 

\begin{theorem}[Weyl's inequality]
	\label{thm:weyl-inequality}
	Let $A$ and $B$ be Hermitian matrices of size $n\times n$. Also let $\lambda_{1} \geq \cdots \geq \lambda_{n}$ denote the ordered eigenvalues of any $n\times n$ Hermitian matrix. Then
	\begin{equation}
		\lambda_i(A) + \lambda_{n}(B) \leq \lambda_i(A+B) \leq \lambda_i(A) + \lambda_{1}(B),
	\end{equation}
	for $i=1,\ldots, n$.
\end{theorem}
Weyl's inequality was originally posed by~\citet{Weyl1912}, and it can also be found, for example, in~\citet[][Theorem 8.4.11]{Bernstein2009},~\citet{HornBook1991}, or~\citet[][Section 5]{Helmke1995}. Weyl's inequality is used in Theorem~\ref{thm:tme-cov-pd} to form a lower bound on the minimum eigenvalue of a covariance approximation.

\begin{theorem}[Langenhop (1960)]
	\label{thm:langenhop}
	Let $u\colon \T \to \R_{\geq0}$ and $f\colon \T \to \R_{\geq0}$ be continuous functions, and let $v\colon \R_{\geq0} \to \R_{\geq 0}$ be a continuous non-decreasing function with $v>0$ on $\R_{>0}$. Now consider the invertible function
	\begin{equation}
	G(r) = \int^r_{r_0} \frac{1}{v(\tau)} \diff \tau, \quad r>0, \quad r_0>0,
	\end{equation}
	and its inverse function $G^{-1}$ defined on domain $E$. Suppose that there is a $t_2\in\T$ such that $G(u(s))-\int^t_s f(\tau) \diff \tau \in E$ for all $s,t\in\T$ and $s\leq t \leq t_2$. If the following inequality is verified,
	\begin{equation}
	u(t) \geq u(s) - \int^t_s f(\tau) \, v(u(\tau)) \diff \tau, \quad s,t\in\T, \quad s\leq t,
	\end{equation}
	then 
	\begin{equation}
	u(t) \geq G^{-1} \Bigg( G(u(s)) - \int^t_s f(\tau) \diff \tau \Bigg), \quad s, t, t_2 \in \T, \quad s\leq t \leq t_2.
	\end{equation}
\end{theorem}
\begin{remark}
    Note that Theorem~\ref{thm:langenhop} is independent of the choice of $r_0 > 0$.
\end{remark}
Langenhop's inequality was originally derived in \citet{Langenhop1960}. A more modern presentation can be found, for example, in ~\citet[][Theorem 2.3.2]{Pachpatte1997}. This theorem is used in Remark~\ref{remark:langenhop-var-f} to obtain a positive lower bound on the variance of an SDE solution.

\begin{theorem}[Peano--Baker series]
	\label{thm:peano-baker}
	Consider linear ODE of the form
	\begin{equation}
		\frac{\diff x(t)}{\diff t} = A(t) \, x(t) + z(t), \quad x(t_0) = x_0 \in \R^d,
	\end{equation}
    where the coefficients $A\colon \T \to \R^{d \times d}$ and $z\colon \T \to \R^{d}$ are locally bounded measurable functions. Then for every $t_0, t\in\T$, the ODE above has a unique solution of the form
	\begin{equation}
		x(t) = \cu{\Lambda}(t,t_0) \, \Bigg( x_0 + \int^t_{t_0} \cu{\Lambda}(t,s) \, z(s) \diff s \Bigg).
		\label{equ:peano-ode-solution}
	\end{equation}
	If moreover $A$ and $z$ are continuous functions, then $\cu{\Lambda}$ can be represented by its Peano--Baker series
	\begin{equation}
		\cu{\Lambda}(t,t_0) = I + \int^t_{t_0}A(s) \diff s + \int^t_{t_0} A(s_1) \int^{s_1}_{t_0} A(s_2) \diff s_2 \diff s_1 + \cdots,
	\end{equation}
	for all $t \in \T$.
\end{theorem}

While the continuity of $A$ and $z$ is not a necessary condition for the existence of $\cu{\Lambda}$ (cf. Theorem~\ref{thm:linear-sde-solution}), the fact that its Peano--Baker series approximation is compactly convergent relies on the continuity of $A$ and $z$~\citep{Baake2011, Brogan2011}. Other constructions of $\cu{\Lambda}$ include, for example, Magnus expansion~\citep{Moan2008MagnusConv} but they have somewhat stricter hypotheses.

\begin{theorem}[Solution of linear SDEs]
	\label{thm:linear-sde-solution}
	Let $A\colon \T \to \R^{d\times d}$ and $B\colon \T \to \R^{d\times w}$ be locally bounded measurable functions, and let $W \colon \T \to \R^w$ be a Wiener process. Then the solution of linear SDE of the form
	\begin{equation}
		\diff U(t) = A(t) \, U(t) \diff t + B(t) \diff W(t)
		\label{equ:thm-sol-lin-sde}
	\end{equation}
	is given by
	\begin{equation}
		U(t) = \Lambda(t)\,U(t_0) + \Lambda(t)\int^t_{t_0} (\Lambda(s))^{-1} \, B(s) \, \diff W(s),
		\label{equ:solution-linear-sdes}
	\end{equation}
	where $\Lambda$ is the unique solution of the matrix ODE
	\begin{equation}\label{equ:solution-linear-sdes-initial-condition}
		\frac{\diff x(t)}{\diff t} = A(t) \, x(t), \quad x(t_0) = I_d \in \R^{d \times d}, \quad t \in \T.
	\end{equation}
\end{theorem}
\begin{remark}
    The $\Lambda$ appearing in Theorem~\ref{thm:linear-sde-solution} is absolutely continuous on $\T$, and is such that $\Lambda(t)$ is a non-singular matrix for all $t\in\T$. Moreover, we have $\cu{\Lambda}(t,s) = \Lambda(t) \, (\Lambda(s))^{-1}$, for all $s, t\in \T$, where $\cu{\Lambda}$ is defined in Theorem~\ref{thm:peano-baker}.
\end{remark}
Theorem~\ref{thm:linear-sde-solution} can be found in~\citet[][Section 5.6]{Karatzas1991}. We used it in Theorem~\ref{thm:ss-dgp-cov} in order to prove the strong existence of solutions to the SDE characterisation of SS-DGP as well as to give an explicit expression for the covariance functions of SS-DGP solutions. Noting that the conditions in Definitions~\ref{def:strong-solution} and~\ref{def:pathwise-unique} are verified, the process defined in Equation~\eqref{equ:solution-linear-sdes} is a strong solution, and the pathwise uniqueness holds for the SDE in Equation~\eqref{equ:thm-sol-lin-sde}~\citep[see,][Section 5.6]{Karatzas1991}.

\chapter{Taylor moment expansion filtering and smoothing}
\label{chap:tme}
This chapter is concerned with Publication~\cp{paperTME}. More specifically, this chapter presents the Taylor moment expansion (TME) scheme for approximating the statistical properties of SDE solutions, such as their mean and covariance. Based on this, we thereupon present TME-based Gaussian filters and smoothers and analyse their stability.

The chapter starts with a general discussion on the motivation and background of the TME method. In Section~\ref{sec:generator}, we briefly review diffusion processes and related infinitesimal generators which are the key ingredients of TME. Then, in Sections~\ref{sec:tme} we formally introduce TME, and in Section~\ref{sec:tme-pd} we analyse the positive definiteness of their covariance approximants. Section~\ref{sec:tme-examples} features several examples that illustrate how to use TME in practice. Finally, in Section~\ref{sec:tme-filter-smoother}, we present Gaussian approximated density filters and smoothers that leverage the TME method for approximating the predictive means and covariances of the system.

\section{Motivation}
\label{sec:tme-motivation}
Let $X(t)$ be an It\^{o} process that satisfies the SDE given by Equation~\eqref{equ:SDE}. In stochastic filtering and smoothing~\citep{Jazwinski1970, Bain2009, Sarkka2013}, it is often of interest to compute the conditional expectation of a given target function $\phi$ for any two time points $t\geq s \in\T$.  For instance, as shown in Algorithm~\ref{alg:gfs}, Gaussian approximated density filters and smoothers require to be able to compute the predictive mean and covariance of SDE solutions. These conditional expectations take the form
\begin{equation}
	\expec{\phi(X(t)) \cond X(s)},
	\label{equ:tme-motivation-moment}
\end{equation}
where different target functions result in different statistical quantities, such as mean, covariance or higher-order moments.

There exist several approaches to computing the expectation in Equation~\eqref{equ:tme-motivation-moment} numerically. One approach is based on forming an ODE that governs the conditional expectation in Equation~\eqref{equ:tme-motivation-moment}~\citep[see, e.g., ][]{DongbinXiu2010, Khasminskii2012, Sarkka2019}. For example, let $\phi(x) = x$, then by It\^{o}'s formula we can obtain an ODE
\begin{equation}
	\frac{\diff \expec{X(t) \cond X(s)}}{\diff t} = \expec{a(X(t), t)  \cond X(s) }
	\label{equ:mean-ode}
\end{equation}
starting from $s\in \T$. However, it is usually hard to solve the ODE in Equation~\eqref{equ:mean-ode} analytically. This is due to the fact that computing its driving term $\expec{a(X(t), t)  \cond X(s)}$ requires computing an expectation with respect to the SDE distribution, which is in general intractable analytically. One solution to this problem is to approximate the expectation using quadrature integration methods~\citep{Sarkka2007, Sarkka2010CD, Kulikova2014}, but the approximation error can accumulate in time, resulting in unstable estimation. Another solution is to iteratively form ODEs that characterise their parent driving terms. Explicitly, one can choose $\phi=a$ in the above, so as to characterise $\expec{a(X(t), t) \cond X(s)}$ by another ODE driven by some function $a'$, then choose $\phi=a'$, and so on. However, this leads to a so-called closure problem as explained in~\citet[][Section 4.4.2]{DongbinXiu2010}.

It is also common to approximate Equation~\eqref{equ:tme-motivation-moment} using numerical discretisation methods, such as Euler--Maruyama, Milstein's method, or higher-order It\^{o}--Taylor expansions~\citep{Kloeden1992}. The upside of these methods is that if the function $\phi$ happens to be a polynomial function, then these methods can give analytical approximations of Equation~\eqref{equ:tme-motivation-moment}. As an example, let $\phi(x) = (x - \expec{X(t)})\, (x - \expec{X(t)})^\trans.$ Then the Euler--Maruyama method gives the approximation
\begin{equation}
	\expec{\phi(X(t)) \cond X(s)} = \cov{X(t) \cond X(s)} \approx (t-s) \, b(X(s), s) \, b(X(s), s)^\trans \nonumber.
\end{equation}
However, for more general non-linear $\phi$ these approaches usually fail to give analytical approximations (see, e.g., Example~\ref{example-tme-benes}), and one often needs to use Monte Carlo methods to approximate the expectation. 

In the remainder of this section, we present the so-called Taylor moment expansion~\citep{Zmirou1986, Zmirou1989, Kessler1997, ZhaoTME2020} approach for computing expectations of the form given in Equation~\eqref{equ:tme-motivation-moment}. This method relies on approximating the expectation in Equation~\eqref{equ:tme-motivation-moment} in terms of a Taylor expansion up to a given order $M$ that depends on the regularity of the coefficients of the SDE verified by $X$. The terms in this expansion are expressed as iterative applications of the infinitesimal generator of the SDE at hand on the target function $\phi$ (see, Section~\ref{sec:generator} for a formal definition). When the coefficients are infinitely smooth, this method offers asymptotically exact representations. We start by giving an overview of diffusion processes and the infinitesimal generator which is an essential part of the TME method.

\section{Infinitesimal generator}
\label{sec:generator}
Diffusion processes~\citep{Dynkin1965, Ikeda1992, Ito2004} are an important subclass of continuous-time Markov processes whose transition probability densities verify certain (infinitesimal) regularities~\citep[see, e.g.,][Definition 10.8.3]{Kuo2006Book}. These processes are entirely characterised by their infinitesimal generators which are defined as follows. Let $X\colon \T \to \R^d$ be a diffusion process starting from any $x\in\R^d$ at $t_0$ and let $\phi$ be a suitable function. The operator $\A$ defined by
\begin{equation}
	\A \phi(x) = \lim_{t\downarrow t_0} \frac{\expec{\phi (X(t)) \cond x} - \phi (x)}{t-t_0}
	\label{equ:generator-general}
\end{equation}
is called the \emph{infinitesimal generator} of the diffusion $X$. Heuristically, the infinitesimal generator represents the expected rate of change of $\phi(X(t))$ around $x$. 

There are many approaches to construct diffusion processes (with desired drift and diffusion coefficients), such as the semigroup approach, the PDE approach (i.e., Kolmogorov backward equation), and the (It\^{o}'s) SDE approach~\citep{Kuo2006Book, ReneBrownianBook2012}. In particular, if one considers diffusion processes that are solutions of SDE, then their infinitesimal generators can be expressed in terms of their SDE coefficients. 

\begin{theorem}[Infinitesimal generator in It\^{o}'s SDE representation]
    \label{thm:generator-in-sde}
	Let $X \colon \T \to \R^d$ be a diffusion process that is the solution of the following time-homogeneous SDE
	\begin{equation}
		\diff X(t) = a(X(t)) \diff t + b(X(t)) \diff W(t), 
		\label{equ:t-homo-sde}
	\end{equation}
	where $a \colon \R^d \to \R^d$, $b \colon \R^d \to\R^{d\times w}$, and $W \colon \T \to \R^w$ is a Wiener process. Also let us define $\Gamma(x) \coloneqq b(x) \, b(x)^\trans$. Then, the infinitesimal generator $\A$ defined in Equation~\eqref{equ:generator-general} is given by
	\begin{equation}
		\begin{split}
			\A \phi(x) &= \sum^d_{i=1} a_{i}(x) \, \tash{\phi}{x_i}(x) + \frac{1}{2} \sum_{i,j=1}^d \Gamma_{ij}(x) \, \frac{\partial^2 \phi}{\partial x_i \, \partial x_j}(x)\\
			&\coloneqq (\nabla_x \phi(x))^\trans \, a(x) + \frac{1}{2} \trace*{\Gamma(x) \, \hessian_x \phi(x)}
			\label{equ:generator-ito}
		\end{split}
	\end{equation}
	for any suitable $\phi\in \mathcal{C}^2(\R^d;\R)$.
	The drift and diffusion coefficients of the diffusion $X$ are then given by $a$ and $\Gamma$, respectively.
\end{theorem}
\begin{proof}
	The proof can be found, for example, in \citet[][Theorem~7.3.3]{Oksendal2003} or \citet[][Theorem 10.9.11]{Kuo2006Book}.
\end{proof}

Theorem~\ref{thm:generator-in-sde} can be extended to the case of time-dependent $x, t \mapsto \phi(x, t)$ and SDE coefficients $x, t \mapsto a(x, t)$, $x, t \mapsto b(x, t)$. For details of this, see, for example,  \citet[][Definition 5.3]{Sarkka2019}.

\section{Taylor moment expansion (TME)}
\label{sec:tme}
Recall that the aim of this section is to compute expectations of the form
\begin{equation}
	\expec{\phi(X(t)) \cond X(s)}
	\label{equ:tme-of-interests}
\end{equation}
for $t\geq s \in \T$ and any given target function $\phi$. 

The idea of TME~\citep{Zmirou1989} is to approximate Equation~\eqref{equ:tme-of-interests} by means of a Taylor expansion
\begin{equation}
	\expec{\phi(X(t)) \cond X(s)} \approx \sum^M_{r=0} \frac{1}{r!} \, \frac{\diff^r \expec{\phi(X(t)) \cond X(s)}}{\diff t^r}(s) \, \Delta t^r
	\label{equ:tme-}
\end{equation}
centred at time $s,$ where $\Delta t \coloneqq t - s$, and $M$ is the expansion order. The right hand side of Equation~\eqref{equ:tme-} involves computing derivatives (when they exist) of the conditional expectation in Equation~\eqref{equ:tme-of-interests} when seen as a function of $t$. It turns out that these derivatives can be explicitly expressed as iterations of the infinitesimal generator in Equation~\eqref{equ:generator-ito}. This is formally stated in the following theorem.
\begin{theorem}[Taylor moment expansion]
	\label{thm:tme}
	Let $M\geq 0$ be an integer and $X\colon \T\to \R^d$ be the solution of the SDE given in Equation~\eqref{equ:t-homo-sde}, where the SDE coefficients $a\colon \R^d\to\R^d$ and $b\colon \R^d\to\R^{d\times w}$ are $M$ times differentiable. Suppose that the target function $\phi \in\mathcal{C}^{2\,(M+1)}(\R^d; \R)$, then we have
	\begin{equation}
		\begin{split}
			\expec{\phi(X(t)) \cond X(s)} &= \sum^M_{r=0} \frac{1}{r!} \, \A^r \phi(X(s)) \, \Delta t^r + R_{M,\phi}(X(s),\Delta t)
			\label{equ:tme-expansion-full}
		\end{split}
	\end{equation}
	for every $t \geq s \in \T$, where $\Delta t \coloneqq t - s$, and
	\begin{equation}
		R_{M,\phi}(X(s),\Delta t) = \int^t_s \int^{\tau_1}_s\cdots \int^{\tau_M}_s \expecbig{\A^{M+1} \phi(X(\tau)) \cond X(s)}\diff \tau
		\label{equ:tme-expansion-remainder}
	\end{equation} 
	is the remainder.
\end{theorem}

\begin{proof}
	We prove that
	\begin{equation}
		\frac{\diff^r \expec{\phi(X(t)) \cond X(s)}}{\diff t^r}(t) = \expec{\A^r \phi(X(t)) \cond X(s)}
		\label{equ:tme-prop-claim}
	\end{equation}
	for every $r = 0,1,\ldots, M$ by induction. When $r=0$, the result trivially holds. When $r=1$ by It\^{o}'s formula (see, Theorem~\ref{thm:ito-formula}) we obtain 
	\begin{equation}
		\begin{split}
			\phi(X(t)) &= \phi(X(s)) + \int^t_s (\nabla_X \phi)^\trans\, a(X(\tau)) + \frac{1}{2}\trace*{\Gamma(X(\tau)) \, \hessian_X\phi} \diff \tau \\
			&\quad+ \int^t_s b(X(\tau)) \diff W(\tau).
			\label{equ:tme-derivation-1}
		\end{split}
	\end{equation}
	Taking the expectation on both sides of the equation above yields
	\begin{equation}
		\expec{\phi(X(t)) \cond X(s)} = \phi(X(s)) + \int^t_s \expec{\A \phi(X(\tau)) \cond X(s)} \diff \tau.
		\label{equ:tme-derivation-1-2}
	\end{equation}
	The fundamental theorem of calculus ensures that $\expec{\phi(X(t)) \cond X(s)}$ is differentiable with respect to $t$ because the integrand in the integral above is continuous. Therefore, we can interchangeably use its differential form
	\begin{equation}
		\frac{\diff \expec{\phi(X(t)) \cond X(s)}}{\diff t}(t) = \expec{\A \phi(X(t)) \cond X(s)}
		\label{equ:tme-derivation-2}
	\end{equation}
	when $\expec{\phi(X(t)) \cond X(s)}$ is seen as function of $t$, so that the claim in Equation~\eqref{equ:tme-prop-claim} holds for $r=1$. Suppose now that Equation~\eqref{equ:tme-prop-claim} holds for an $r> 1$, then by applying It\^{o}'s formula again on $\A^r \phi(X(t))$ we obtain
	\begin{equation}
		\A^r \phi(X(t)) = \A^r \phi(X(s)) + \int^t_s \A^{r+1} \phi(X(\tau)) \diff \tau.
		\label{equ:tme-derivation-3}
	\end{equation}
	Noting the fact that 
	\begin{equation*}
		\frac{\diff^r \expec{\phi(X(t)) \cond X(s)}}{\diff t^r}(s) = \A^r \phi(X(s)),
	\end{equation*} 
	we can take expectations on both sides of Equation~\eqref{equ:tme-derivation-3}, and substitute the resulting expression into Equation~\eqref{equ:tme-prop-claim}, we thus obtain
	\begin{equation}
		\begin{split}
			\frac{\diff^r \expec{\phi(X(t)) \cond X(s)}}{\diff t^r}(t) &= \A^r \phi(X(s)) +  \int^t_s \expecbig{\A^{r+1} \phi(X(\tau)) \cond X(s)} \diff \tau \\
			&= \frac{\diff^r \expec{\phi(X(t)) \cond X(s)}}{\diff t^r}(s) +  \int^t_s \expecbig{\A^{r+1} \phi(X(\tau)) \cond X(s)} \diff \tau\nonumber
		\end{split}
	\end{equation}
	which is the integral form of the ordinary differential equation
	\begin{equation}
		\frac{\diff^{r+1} \expec{\phi(X(t)) \cond X(s)}}{\diff t^{r+1}}(t) = \expecbig{\A^{r+1} \phi(X(t)) \cond X(s)}.\nonumber
	\end{equation}
	Hence, Equation~\eqref{equ:tme-prop-claim} is proven. 
	
	Finally, by Taylor's theorem, we arrive at Equation~\eqref{equ:tme-expansion-full}. The remainder in Equation~\eqref{equ:tme-expansion-remainder} is obtained by taking expectations on both sides of Equation~\eqref{equ:tme-derivation-3} and substituting back into Equation~\eqref{equ:tme-derivation-1-2} multiple times for $r=1,\ldots, M$. The proof details can be found in \citet[][Lemma 4]{Zmirou1986} or \citet[][Lemma 1]{Zmirou1989}, for example.
\end{proof}

Note that even though the expansion is taken up an order $M \geq 0$, the TME method gives an exact representation of $\expec{\phi(X(t)) \cond X(s)}$ for any suitable function $\phi$. However, computing the remainder is infeasible in practice, and we usually approximate the representation by discarding the remainder\footnote{If we discard the remainder, then the TME \emph{approximation} only needs $a$ and $b$ to be $M-1$ times differentiable and $\phi$ to be $2\,M$ times differentiable.}. This leads to a polynomial approximation with respect to $\Delta t$. However, please note that the order $M$ cannot be chosen entirely arbitrarily because it depends on the smoothness of the SDE coefficients and function $\phi$.

In Gaussian filtering and smoothing we are particularly interested in estimating the conditional means and covariances of the process $X$. In order to do so, we introduce the following target functions
\begin{equation}
	\begin{split}
		\phi^{\mathrm{I}}(x) &= x, \\
		\phi^{\mathrm{II}}(x) &= x\,x^\trans,
	\end{split}
\end{equation}
corresponding to the first and second moments, respectively. Their TME representations are then given in Lemma~\ref{lemma:tme-1-2-moments}.
\begin{remark}
	\label{remark:multidim-generator}
	While generator $\A$ in Equation~\eqref{equ:generator-ito} is defined for scalar-valued target functions only, this definition can be extend to vector/matrix-valued target functions by introducing an elementwise operator $\Am$. Namely, let $\phi\colon \R^d \to \R^{m\times n}$, then $\Am$ is defined via
	\begin{equation}
		\Am \phi(x) = 
		\begin{bmatrix}
			\A \phi_{11} & \cdots & \A \phi_{1n} \\
			\vdots & \ddots & \vdots \\
			\A \phi_{m1} & \cdots & \A \phi_{mn}
		\end{bmatrix}(x),
	\end{equation}
	where $\phi_{ij}$ stands for the $i,j$-th element of $\phi$.
\end{remark}
\begin{lemma}[TME for first and second moments]
	\label{lemma:tme-1-2-moments}
	The first and second conditional moments of $X$ are given by
	\begin{equation}
		\expec{X(t) \cond X(s)} = \sum^M_{r=0} \frac{1}{r!} \, \Am \phi^{\mathrm{I}}(X(s)) \Delta t^r + R_{M,\phi^\mathrm{I}}(X(s), \Delta t)
	\end{equation}
	and
	\begin{equation}
		\expecbig{X(t)\,X(t)^\trans \cond X(s)} = \sum^M_{r=0} \frac{1}{r!} \, \Am \phi^{\mathrm{II}}(X(s)) \Delta t^r + R_{M,\phi^{\mathrm{II}}}(X(s), \Delta t)
	\end{equation}
	for all $s < t \in \T$, respectively. 
\end{lemma}
Notice that if we choose $M=1$ in the lemma above, then the resulting TME approximation $\sum^1_{r=0} \frac{1}{r!} \, \Am \phi^{\mathrm{I}}(X(s)) \Delta t^r$ is exactly the same as the Euler--Maruyama approximation for the first moment. Moreover, the TME covariance approximation (formulated in the next section) will also coincide with the Euler--Maruyama approximation for the covariance when $M=1$.

\section{Covariance approximation by TME}
\label{sec:tme-pd}
This section shows how to use the TME method to approximate conditional covariances of the form in Equation~\eqref{equ:tme-of-interests}. Based on the first and second moment representations in Lemma~\ref{lemma:tme-1-2-moments}, it seems that we can approximate the covariance $\cov{X(t) \cond X(s)}$ by
\begin{align}
	\cov{X(t) \cond X(s)} &= \expecbig{X(t)\,X(t)^\trans \cond X(s)} - \expec{X(t) \cond X(s)} \expec{X(t) \cond X(s)}^\trans \nonumber\\
	&\approx \sum^M_{r=0} \frac{1}{r!} \, \Am^r \phi^{\mathrm{II}}(X(s)) \Delta t^r \label{equ:tme-crude-cov}\\
	&\quad- \Bigg(\sum^M_{r=0} \frac{1}{r!} \, \Am^r \phi^{\mathrm{I}}(X(s)) \Delta t^r \Bigg) \, \Bigg(\sum^M_{r=0} \frac{1}{r!} \, \Am^r \phi^{\mathrm{I}}(X(s)) \Delta t^r\Bigg)^\trans \nonumber
\end{align}
up to an order $M$. However, this approximation has two problems. First, the polynomial degree in this approximation is inconsistent with the approximations of the first and second moments. This is because the power of $\Delta t$ in Equation~\eqref{equ:tme-crude-cov} is now up to order $2\,M$ instead of $M$. Hence, we need to truncate the polynomial terms $\Delta t^r$ for $r>M$ in Equation~\eqref{equ:tme-crude-cov} for the sake of consistency. 

The second problem is that the positive definiteness of the covariance approximation is not guaranteed as we discard the remainders~\citep{Iacus2008, ZhaoTME2020}. To see this, let us consider a simple one-dimensional example as follows. Let $X\colon\T\to\R$ be an It\^{o} process that solves the  SDE~\eqref{equ:t-homo-sde}. Suppose that its dispersion term is non-zero and let us also choose $M=2$. Then the variance approximation in Equation~\eqref{equ:tme-crude-cov}, after truncating the polynomial terms $\Delta t^r$ for $r>M$, becomes $\Gamma(X(s)) \, \Delta t + \Gamma(X(s))\,\frac{\diff a}{\diff X}(X(s)) \Delta t^2$. This approximation is not positive in general because it is positive if and only if $\frac{\diff a}{\diff X}(X(s)) > -1 \, / \, \Delta t$. Moreover, if one requires the positivity hold uniformly for all $\Delta t\in\R_{>0}$ and all $X(s)\in\R$, then the function $\frac{\diff a}{\diff X}$ must be positive on its domain. 

Therefore, in the following theorem we derive the TME approximation for the covariance $\cov{X(t) \cond X(s)}$ by truncating the unnecessary polynomial terms of $\Delta t$ in Equation~\eqref{equ:tme-crude-cov}, and we thereupon provide a sufficient criterion to ensure the positive definiteness of such approximation. 
\begin{theorem}[TME covariance approximation]
	\label{thm:tme-cov-pd}
	Let $X\colon \T \to\R^d$ be the solution of the SDE that verifies Theorem~\ref{thm:tme}. Let integer $M\geq 1$. The $M$-order TME approximation for $\cov{X(t) \cond X(s)}$ is
	\begin{equation}
		\Sigma_M(\Delta t) =  \sum^M_{r=1} \frac{1}{r!} \, \Theta_{ r} \, \Delta t^r,
		\label{equ:tme-cov-sigma}
	\end{equation}
	where 
	\begin{equation}
		\begin{split}
			\Theta_{r} &\coloneqq  \Theta_{X(s), r} \\
			&= \Am^r \phi^{\mathrm{II}}(X(s)) - \sum^r_{k=0} \binom{r}{k}\, \Am^k \phi^{\mathrm{I}}(X(s)) \, \Big(\Am^{r-k} \phi^{\mathrm{I}}(X(s))\Big)^\trans,
		\end{split}
		\label{equ:tme-theta}
	\end{equation}
	and $\binom{r}{k}$ denotes binomial coefficient. The approximation $\Sigma_M(\Delta t)$ is positive definite if the associated polynomial 
	\begin{equation}
		\chi(\Delta t) = \sum^M_{r=1} \frac{1}{r!} \, \mineig(\Theta_{r}) \Delta t^r >0.
		\label{equ:tme-cov-polynomial}
	\end{equation}
\end{theorem}
\begin{proof}
	Let us denote by $\phi^{\mathrm{II}}_{ij}$ the $i,j$-th element of $\phi^{\mathrm{II}}$, and let us also denote by $\phi^{\mathrm{I}}_{i}$ the $i$-th element of $\phi^{\mathrm{I}}$. Then the $i,j$-th element of the covariance approximation in Equation~\eqref{equ:tme-crude-cov} is
	\begin{equation}
		\begin{split}
			&\sum^M_{r=1} \frac{1}{r!} \, \A^r \phi^{\mathrm{II}}_{ij}(X(s)) \Delta t^r\\
			&\quad- \Bigg(\sum^M_{r=1} \frac{1}{r!} \, \A^r \phi^{\mathrm{I}}_i(X(s)) \Delta t^r\Bigg) \, \Bigg(\sum^M_{r=1} \frac{1}{r!} \, \A^r \phi^{\mathrm{I}}_j(X(s)) \Delta t^r\Bigg).
			\label{equ:tme-sigma-untrancated}
		\end{split}
	\end{equation}
	Let $[\Sigma_M]_{ij}$ be the truncation of Equation~\eqref{equ:tme-sigma-untrancated} up to order $M$ (i.e., eliminating terms with $\Delta t^r$ for all $r>M$).
	Then, by Cauchy product (see, Theorem~\ref{thm:cauchy-product}) we have
	\begin{equation}
		\begin{split}
			[\Sigma_M]_{ij} &= \sum^M_{r=1} \left[ \frac{1}{r!} \, \A^r \phi^{\mathrm{II}}_{ij}(X(s)) - \Bigg(\sum^r_{k=0} \frac{\A^k \phi^{\mathrm{I}}_i(X(s)) \, \A^{r-k} \phi^{\mathrm{I}}_j(X(s))}{k!\,(r-k)!}\Bigg)\right]  \Delta t^r \\
			&= \sum^M_{r=1} \frac{1}{r!} \, \Bigg[ \A^r \phi^{\mathrm{II}}_{ij}(X(s)) - \sum^r_{k=0}\binom{r}{k}\, \A^k \phi^{\mathrm{I}}_i(X(s)) \, \A^{r-k} \phi^{\mathrm{I}}_j(X(s)) \Bigg]\, \Delta t^r. \nonumber
		\end{split}
	\end{equation}
	Hence, by rearranging $[\Sigma_M]_{ij}$ into a matrix for $i,j=1,\ldots,d$ we obtain Equation~\eqref{equ:tme-cov-sigma}. Since $\Sigma_M(\Delta t)$ is symmetric by definition, its eigenvalues are real. Then by using Weyl's inequality (see, Theorem~\ref{thm:weyl-inequality}) we obtain
	\begin{equation}
		\mineig(\Sigma_M(\Delta t)) \geq \sum^M_{r=1} \frac{1}{r!} \, \mineig(\Theta_{r}) \, \Delta t^r.
	\end{equation}
	Hence, $\Sigma_M(\Delta t)$ is positive definite if Equation~\eqref{equ:tme-cov-polynomial} holds. Note that $\Theta_0 = 0$.
\end{proof}

Theorem~\ref{thm:tme-cov-pd} shows that $\Sigma_M(\Delta t)$ is a polynomial of $\Delta t$ with coefficients determined by Hermitian matrices $\lbrace \Theta_r\colon r=1,\ldots,M \rbrace$. These matrices depend on the starting condition $X(s)$. In order to guarantee the positive definiteness of $\Sigma_M(\Delta t)$, we use Weyl's inequality in order to find a lower bound on its minimum eigenvalue, resulting in another polynomial $\chi(\Delta t)$ of $\Delta t$. This reduces the problem of analysing the positive definiteness of $\Sigma_M(\Delta t)$ into the problem of analysing the positivity of polynomial $\chi(\Delta t)$.

To ensure the positivity of polynomial $\chi(\Delta t)$, one can trivially restrict all the coefficients $\lbrace \mineig(\Theta_{r})\colon r=1,\ldots,M \rbrace$ to be positive, but this in turn significantly limit the SDE models that the TME approximation applies. Another solution is to let $\chi(\Delta t)$ have no real roots on $\R_{>0}$ and $\chi(\Delta t)>0$ for some $\Delta t\in\R_{>0}$. For instance, one can bound/count the number of real roots of polynomial on any intervals by using Budan's theorem or Sturm's theorem~\citep{Basu2006}.

The positive definiteness of $\Sigma_M(\Delta t)$ is entirely determined by the order $M,$ the time interval $\Delta t,$ the starting point $X(s),$ and the SDE coefficients. If $\Delta t$ is somehow tunable, one can then let $\Delta t$ be small enough to guarantee the positive definiteness numerically. This is true because the term $\Theta_1=\Gamma(X(s))$ which is positive semi-definite by definition, dominates $\Sigma_M(\Delta t)$ in the limit $\Delta t \to 0.$ This numerical approach is especially useful in Gaussian filtering and smoothing, as it is common to perform multiple integration steps with small $\Delta t$ in the prediction steps~(see, Algorithm~\ref{alg:gfs}).

However, it might not always be possible to tune $\Delta t$. For example, if we have limited computational resources, using multiple integration steps with smaller $\Delta t$ in Gaussian filtering and smoothing may not be realistic. Hence, it is also important to show the positive definiteness conditions of $\Sigma_M(\Delta t)$ that are independent of the choice of $\Delta t$. A few results on these conditions are given in the following corollary.

\begin{corollary}
	\label{corol-tme-pd-all-dt}
	The following results hold for all $\Delta t \in\R_{>0}$.
	\begin{enumerate}[I.]
		\item $\Sigma_1(\Delta t)$ is positive definite, if $\Gamma(X(s))$ is positive definite. Notice that $\Gamma(X(s))$ is always positive semi-definite by definition.
		\item $\Sigma_2(\Delta t)$ is positive definite, if $\Theta_2$ and $\Gamma(X(s))$ are positive semi-definite, and one of the two is positive definite.
		\item $\Sigma_3(\Delta t)$ is positive definite, if $\Theta_{3}$ is positive semi-definite and $\mineig(\Theta_{2}) > \frac{-2\,\sqrt{6}}{3}\,\sqrt{\mineig(\Theta_{1})\, \mineig(\Theta_{3})}$.
	\end{enumerate}
\end{corollary}
\begin{proof}
	This corollary follows from Theorem~\ref{thm:tme-cov-pd} and the root conditions of quadratic and cubic polynomials (i.e., by letting $\chi(\Delta t)$ have no real roots on $\R_{>0}$). See,~\citet[][Proposition 5]{ZhaoTME2020} for details.
\end{proof}

\begin{remark}
	For $r=0,1,2$ we can immediately derive $\Theta_0 = 0$, $\Theta_1 = \Gamma(X(s))$, and $\Theta_2 = \Gamma(X(s)) \, \jacob_X a(X(s)) + (\Gamma(X(s)) \, \jacob_X a(X(s)))^\trans$. For results in higher orders (and in one state dimension), see, \citet[][Example 9]{ZhaoTME2020}.
\end{remark}

The approximation $\Sigma_M$ has an important property that it does not \textit{explicitly} depend on $X(s)$. More precisely, the expression of $\Sigma_M$ only have $X(s)$ appearing inside the SDE coefficients and their derivatives. With a slight abuse of terminology, we say that $\Sigma_M$ is $X(s)$-\emph{homogeneous}. This property is meaningful in the sense that it is possible to ensure the positive definiteness of $\Sigma_M$ independent of $X(s)$ by manipulating the SDE coefficients.

\begin{lemma}[$X(s)$-homogeneity]
	\label{lemma:tme-cov-homo}
	Let $b(X(t)) = b\in\R^{d\times w}$ be a constant, hence $\Gamma = b \, b^\trans$. Denote by $\Theta^{uv}_r$ the $u,v$-th element of $\Theta_r$ and $\Gamma_{ij}$ the $i,j$-th element of $\Gamma$. Also denote by $\alpha^u_r \coloneqq \A^r \phi^{\mathrm{I}}_u(X(s))$. Then 
	\begin{equation}
		\begin{split}
			\Theta_{r}^{uv} &= \sum^d_{i,j=1} \sum^{r-1}_{k=0} \binom{r-1}{k} \tash{\alpha^u_k}{X_i(s)} \, \tash{\alpha^v_{r-k-1}}{X_j(s)} \, \Gamma_{ij} + \A \Theta_{r-1}^{uv} \\
			&= \sum^{r-1}_{k=0} \Am^k \sum^{r-k-1}_{l=0} \binom{r-k-1}{l} \trace*{\nabla_X \alpha^v_{r-k-l-1} \, \big(\nabla_X \alpha^u_k\big)^\trans \, \Gamma}
		\end{split}
		\label{equ:tme-homo-theta}
	\end{equation}
	and $\Theta^{uv}_0=0$, for all $r\geq 1$ and $u,v\leq d$. Notice that $\Am^r \phi^{\mathrm{I}}(X(s))$ is $X(s)$-homogeneous for $r\geq 0$.
\end{lemma}
\begin{proof}
	Define $\beta_r^{uv} \coloneqq \A^r \phi^{\mathrm{II}}_{uv}(X(s))$. Since $\Theta^{uv}_r = \beta^{uv}_r - \sum^r_{k=0}\binom{r}{k} \alpha^u_k \, \alpha^v_{r-k}$ by Equation~\eqref{equ:tme-theta}, the task is to find an $X(s)$-homogeneous expression for $\Theta^{uv}_r$. If we do a few initial trials for $r=0,1,\ldots$, we will find a pattern
	\begin{equation}
		\begin{split}
			\beta^{uv}_0 &= \alpha^u_0 \, \alpha^v_0, \\
			\beta^{uv}_1 &= \alpha^u_0 \, \alpha^v_1 + \alpha^v_0 \, \alpha^u_1 + \Gamma_{uv}, \\
			&\vdots
		\end{split}
	\end{equation}
	Hence, we want to prove that
	\begin{equation}
		\beta_r^{uv} = \sum^r_{k=0}\binom{r}{k} \alpha^u_k \, \alpha^v_{r-k} + \Theta_r^{uv},
		\label{equ:tme-x-homo-induction1}
	\end{equation}
	where 
	\begin{equation}
		\Theta_{r}^{uv} = \sum^d_{i,j=1} \sum^{r-1}_{k=0} \binom{r-1}{k} \tash{\alpha^u_k}{X_i(s)} \, \tash{\alpha^v_{r-k-1}}{X_j(s)} \, \Gamma_{ij} + \A \Theta_{r-1}^{uv}.
		\label{equ:tme-x-homo-induction2}
	\end{equation}
	Equation~\eqref{equ:tme-x-homo-induction1} holds for $r=0$ and $1$. Now let us suppose that they hold for an $r> 1$. Then, by the definition of $\beta^{uv}_r$ we have
	\begin{equation}
		\begin{split}
			\beta^{uv}_{r+1} = \A^{r+1} \phi^{\mathrm{II}}_{uv}(X(s)) &= \A \beta^{uv}_{r} \\
			&= \sum^d_{i=1} \tash{\beta^{uv}_r}{X_i(s)} \, a_i(X(s)) + \frac{1}{2} \sum^d_{i,j=1}\frac{\partial^2 \beta^{uv}_r}{\partial X_i(s) \, \partial X_j(s)} \, \Gamma_{uv}.
		\end{split}
		\label{equ:tme-x-homo-beta-r+1}
	\end{equation}
	Now, by substituting Equation~\eqref{equ:tme-x-homo-induction1} in Equation~\eqref{equ:tme-x-homo-beta-r+1},  Equation~\eqref{equ:tme-x-homo-beta-r+1} becomes
	\begin{align*}
		&\beta^{uv}_{r+1} = \\
		&\sum_{i=1}^d\sum^r_{k=0}\binom{r}{k}\Bigg(\tash{\alpha^u_k}{X_i(s)} \, \alpha^v_{r-k} + \alpha^u_k \, \tash{\alpha^v_{r-k}}{X_i(s)} \Bigg) \, a_i(X(s)) \\
		&\quad+ \frac{1}{2}\sum_{i, j=1}^d\Bigg(\sum^r_{k=0}\binom{r}{k}\Big(\tashh{\alpha_k^u}{X_i(s)}{X_j(s)} \, \alpha^v_{r-k} + \tash{\alpha_k^u}{X_i(s)} \, \tash{\alpha^v_{r-k}}{X_j(s)} + \tash{\alpha_k^u}{X_j(s)} \, \tash{\alpha^v_{r-k}}{X_i(s)} \\
		&\qquad\qquad\qquad\qquad\qquad  + \alpha^u_k \, \tashh{\alpha^v_{r-k}}{X_i(s)}{X_j(s)} \Big)\Bigg)\Gamma_{ij} + \A \Theta^{uv}_r.
	\end{align*}
	By the definition of generator $\A$ we have that $\sum^d_{i=1} \tash{\alpha^u_k}{X_i(s)} \, \alpha^v_{r-k} \, a_i(X(s)) + \frac{1}{2} \sum^d_{i,j=1} \tashh{\alpha_k^u}{X_i(s)}{X_j(s)} \, \alpha^v_{r-k} = \A \alpha^u_k \, \alpha^v_{r-k} = \allowbreak \alpha^u_{k+1} \, \alpha^v_{r-k}$. Defining $\binom{r}{-1} = 0,$ we arrive at
	\begin{align*}
		\beta^{uv}_{r+1} &= \sum^r_{k=0} \binom{r}{k} \big(\alpha^u_{k+1} \, \alpha^v_{r-k} + \alpha^u_{k} \, \alpha^v_{r-k+1}\big) + \sum_{i, j=1}^d \sum^r_{k=0} \binom{r}{k}\tash{\alpha_k^u}{X_i(s)} \, \tash{\alpha^v_{r-k}}{X_j(s)} \,\Gamma_{ij} + \A \Theta^{uv}_r\\
		&= \sum^{r+1}_{k=0} \binom{r}{k-1}\alpha^u_k \, \alpha^v_{r-k+1} + \sum^{r+1}_{k=0}\binom{r}{k}\alpha^u_{k} \, \alpha^v_{r-k+1} + \Theta^{uv}_{r+1} \\
		&= \sum^{r+1}_{k=0} \binom{r+1}{k} \alpha^u_{k} \, \alpha^v_{r-k+1} + \Theta^{uv}_{r+1}
	\end{align*}
	which is exactly Equation~\eqref{equ:tme-x-homo-induction1} at $r+1$. Thus, Equation~\eqref{equ:tme-x-homo-induction1} is proven by mathematical induction. Finally, 
	\begin{equation}
		\begin{split}
			\Theta^{uv}_{r} &= \beta^{uv}_r - \sum^r_{k=0}\binom{r}{k} \alpha^u_k \, \alpha^v_{r-k} \\
			&= \sum^d_{i,j=1} \sum^{r-1}_{k=0} \binom{r-1}{k} \tash{\alpha^u_k}{X_i(s)} \, \tash{\alpha^v_{r-k-1}}{X_j(s)} \, \Gamma_{ij} + \A \Theta_{r-1}^{uv} \\
			&= \sum^{r-1}_{k=0} \binom{r-1}{k} \trace*{\nabla_X \alpha^v_{r-k-1} \, \big(\nabla_X \alpha^u_k\big)^\trans \, \Gamma} + \A \Theta_{r-1}^{uv}.
			\nonumber
		\end{split}
	\end{equation}
	Starting from $\Theta_0 = 0$, one can arrive at the last line in Equation~\eqref{equ:tme-homo-theta} by iterating $\Theta^{uv}_r$ for $r\geq 1$.
\end{proof}

The homogeneity property does not hold for the first and second moment approximations in Lemma~\ref{lemma:tme-1-2-moments}. For instance, the TME mean approximation reads $\expec{X(t) \cond X(s)} \approx X(s) + a(X(s))\, \Delta t + \cdots$ which explicitly depends on $X(s)$.

\section{Numerical examples of TME}
\label{sec:tme-examples}
In this section we present a few examples that apply the TME method for approximating expectations of the form in Equation~\eqref{equ:tme-of-interests}. In addition, we compare the results of TME against some commonly-used methods, such as the Euler--Maruyama scheme and the It\^{o}--Taylor strong order 1.5 (It\^{o}-1.5) method~\citep{Kloeden1992}. In Example~\ref{example:tme-softplus}, we present a concrete example showing how to use Theorem~\ref{thm:tme-cov-pd} to analyse the positive definiteness of a TME covariance approximation. 

For simplicity we call TME-$M$ the $M$-order TME approximation.

\begin{figure}[t!]
	\centering
	\includegraphics[width=.9\linewidth]{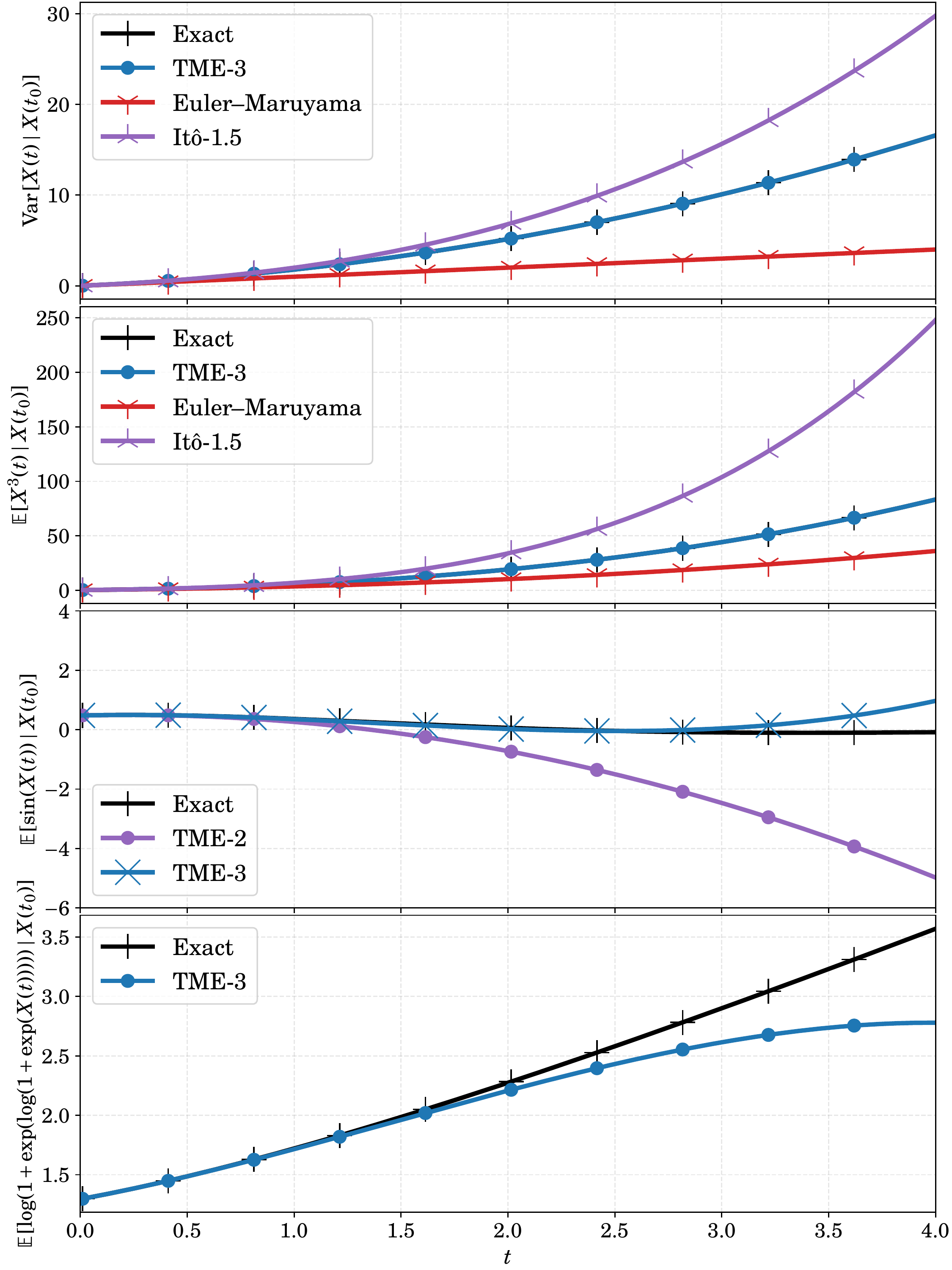}
	\caption{Expectation approximations in Example~\ref{example-tme-benes}. The exact solutions are computed by numerical integration, as the transition density of the Bene\v{s} model is explicitly known~\citep[][Equation 10.58]{Sarkka2019}.}
	\label{fig:tme-benes}
\end{figure}

\begin{example}
	\label{example-tme-benes}
	Consider an It\^{o} process $X\colon \T \to \R$ which solves the Bene\v{s} model
	\begin{equation}
		\diff X(t) = \tanh(X(t)) \diff t + \diff W(t),
	\end{equation}
	starting from $X(t_0) = 0.5$. We are interested in computing its variance $\varr{X(t) \cond X(t_0)}$, third moment $\expec{X(t)^3 \cond X(t_0)}$, and two expectations 
	\begin{equation}
		\begin{split}
			&\expec{\sin(X(t)) \cond X(t_0)}, \\
			&\expec{\log(1+\exp(\log(1+\exp(X(t))))) \cond X(t_0)}.
		\end{split}
		\label{equ:tme-benes-nn}
	\end{equation}
	Notice that one can understand the last expectation above as a way to describe the propogation of $X$ through a neural network consiting of two single-neuron layers with Softplus activation functions.
	
	The TME-2 approximation for the variance $\varr{X(t) \cond X(t_0)}$ is exact. Specifically, $\Sigma_2(\Delta t) = \Delta t + (1-\tanh(X(t_0))^2)\,\Delta t^2$ is equal to $\varr{X(t) \cond X(t_0)}$.
\end{example}

In Figure~\ref{fig:tme-benes}, we plot the results for the expectations in Example~\ref{example-tme-benes}. In addition, we compare the TME method against the Euler--Maruyama and It\^{o}-1.5 methods. From the figure, we see that the TME approach outperforms the Euler--Maruyama and It\^{o}-1.5 methods significantly. Also, the TME approach can approximate the expectations in Equation~\eqref{equ:tme-benes-nn} to a good extent within a small time span.  Note that the Euler--Maruyama and It\^{o}-1.5 schemes do not give closed-form approximations for the expectations in Equation~\eqref{equ:tme-benes-nn}, we thus omit the two methods for these expectations. 

In the next example, we show how to use Theorem~\ref{thm:tme-cov-pd} and Corollary~\ref{corol-tme-pd-all-dt} in practice to analyse the positive definiteness of the TME covariance approximation of a non-linear multidimensional SDE. 

\begin{figure}[t!]
	\centering
	\includegraphics[width=.7\linewidth]{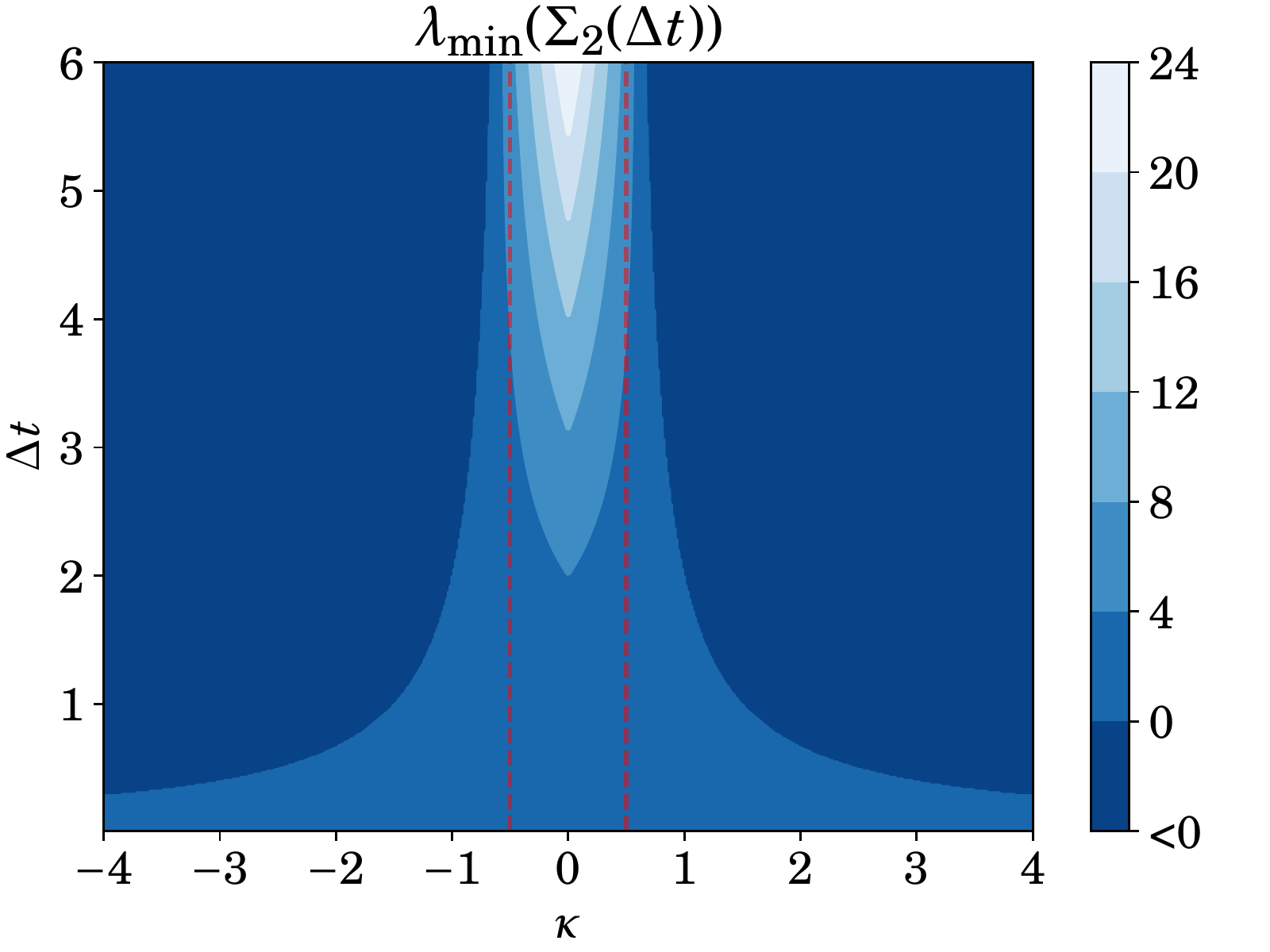}
	\caption{Contour plot of the minimum eigenvalue of $\Sigma_2$ with respect to $\Delta t$ and $\kappa$ in Example~\ref{example:tme-softplus}. Red dashed lines stand for $\kappa=-0.5$ and $0.5$.}
	\label{fig:tme-softplus}
\end{figure}

\begin{example}
	\label{example:tme-softplus}
	Consider a two-dimensional SDE
	\begin{equation}
		\begin{split}
			\diff X^1(t) &= (\log(1 + \exp(X^1(t))) + \kappa \, X^2(t)) \diff t + \diff W_1(t), \\
			\diff X^2(t) &= (\log(1 + \exp(X^2(t))) + \kappa \, X^1(t)) \diff t + \diff W_2(t), 
		\end{split}
	\end{equation}
	where $\kappa \in \R$ is a tunable parameter. We want to ensure the positive definiteness of $\Sigma_2(\Delta t)$ for all $\Delta t \in \R_{>0}$ by tuning $\kappa$. In order to do so, we can first explicitly derive $\Theta_1$ and $\Theta_2$. It turns out that $\Theta_1$ is an identity matrix and
	\begin{equation}
		\Theta_2 = 2\,
		\begin{bmatrix}
			\frac{e^{X^1(t_0)}}{e^{X^1(t_0)}+1} & \kappa \\
			\kappa & \frac{e^{X^2(t_0)}}{e^{X^2(t_0)}+1}
		\end{bmatrix}.
	\end{equation}
	Then, by Corollary~\ref{corol-tme-pd-all-dt} it is sufficient to guarantee the positive definiteness of $\Sigma_2(\Delta t)$ for all $\Delta t\in\R_{>0}$ by ensuring that $\Theta_2$ is positive semi-definite. Thus, one should let 
	\begin{equation}
		\abs{\kappa} \leq \sqrt{\frac{e^{X^1(t_0) + X^2(t_0)}}{(e^{X^1(t_0)}+1) \, (e^{X^2(t_0)}+1)}} \nonumber.
	\end{equation}
	Figure~\ref{fig:tme-softplus} plots the minimum eigenvalues of $\Sigma_2(\Delta t)$ with respect to $\Delta t$ and $\kappa$ when $X^1(t_0)=X^2(t_0)=0$. In this case, $\abs{\kappa}$ should be less than $0.5$ (red dashed lines in the figure) in order to guarantee the positive definiteness of $\Sigma_2(\Delta t)$. The figure shows that $\Sigma_2(\Delta t)$ is indeed positive definite for all $\Delta t\in\R_{>0}$ within the region $\abs{\kappa}\leq 0.5$, and that this sufficient region is very close to the true region (i.e., the region of $\kappa$ that $\lambda_\mathrm{min}(\Sigma_2(\Delta t))>0$ for all $\Delta t$).
\end{example}

The positive definiteness analysis via Theorem~\ref{thm:tme-cov-pd} might be limited to simple SDEs and low orders of expansion. For $M\leq 3$, some explicit results can be stated as shown in Corollary~\ref{corol-tme-pd-all-dt}, but higher-order expansions can result in complicated polynomials to be analysed. 

\section{TME Gaussian filter and smoother}
\label{sec:tme-filter-smoother}

In the pioneering works by~\citet{Zmirou1986, Kessler1997, Sahalia2003}, the TME method was originally introduced for estimating unknown parameters of SDEs. More specifically, they use TME to discretise SDEs in order to perform maximum likelihood estimations. In this section, we show that the TME method could also be applied for solving Gaussian filtering and smoothing problems~\citep{ZhaoTME2020, ZhaoZTMEsmoother2021}.

Consider a (time-homogeneous) continuous-discrete state-space model
\begin{equation}
	\begin{split}
		\diff X(t) &= a(X(t)) \diff t + b(X(t)) \diff W(t), \quad X(t_0) = X_0,\\
		Y_k &= h(X_k) + \xi_k, \quad \xi_k \sim \mathrm{N}(0, \Xi_k),
	\end{split}
	\label{equ:tme-cd-model}
\end{equation}
where the solution $X\colon \T\to\R^d$ is observed through a non-linear function $h\colon \R^d\to\R^{d_y}$ and additive Gaussian noises $\lbrace \xi_k\colon k=1,2,\ldots \rbrace$. Furthermore, we assume that the SDE coefficients satisfy the conditions in Theorem~\ref{thm:tme}, so that we can apply the TME method. 

As shown in Algorithm~\ref{alg:gfs}, a key procedure of Gaussian filtering is to propagate the previous filtering result through the SDE and compute the predictive mean $m^-_k$ and covariance $P^-_k$. As for the Gaussian smoothing steps, one needs to compute the cross-covariance $D^-_{k+1}$ in Algorithm~\ref{alg:gfs}. These quantities can be approximated by using the TME method as follows.

Let us denote by $f^M$ and $Q^M$ the $M$-order TME approximations to the conditional mean and covariance (see, Lemma~\ref{lemma:tme-1-2-moments} and Theorem~\ref{thm:tme-cov-pd}), that are,
\begin{equation}
	\begin{split}
		\expec{X_k \cond X_{k-1}} &\approx f^M(X_{k-1}), \\
		\cov{X_k \cond X_{k-1}} &\approx Q^M(X_{k-1}).
	\end{split}
\end{equation}
Then by substituting $f^M$ and $Q^M$ into the prediction step in Algorithm~\ref{alg:gfs} we obtain the TME-approximated predictive mean and covariance
\begin{equation}
	\begin{split}
		&\int x_k \, p_{X_k \cond Y_{1:k-1}}(x_k \cond y_{1:k-1}) \diff x_k \\
		&= \iint x_k \, p_{X_k \cond X_{k-1}}(x_k \cond x_{k-1}) \, p_{X_{k-1} \cond Y_{k-1}} (x_{k-1} \cond y_{k-1}) \diff x_{k-1} \diff x_k\\
		&\approx \int f^M(x_{k-1}) \, \mathrm{N} \big(x_{k-1} \cond m^f_{k-1}, P^f_{k-1} \big) \diff x_{k-1}\\
		&= m^-_k, \\
		&\int (x_k - m_k^-) \, (x_k - m_k^-)^\trans \, p_{X_k \cond Y_{1:k-1}}(x_k \cond y_{1:k-1}) \diff x_k \\
		&\approx \int \Big(Q^M(x_{k-1}) + f^M(x_{k-1}) \, \big(f^M(x_{k-1})\big)^\trans \Big) \,\mathrm{N} \big(x_{k-1} \cond m^f_{k-1}, P^f_{k-1} \big) \diff x_{k-1} \\
		&\quad- m_k^- \, (m_k^-)^\trans \\
		&= P^-_k.
	\end{split}
	\label{equ:tme-gfs-pred}
\end{equation}
Similarly, for the cross-covariance $D_{k+1}$ in the smoothing pass we have
\begin{equation}
	\begin{split}
		&\iint x_k \, x_{k+1}^\trans \, p_{X_{k+1} \cond X_{k}}(x_{k+1} \cond x_{k}) \, p_{X_k \cond Y_{1:k}}(x_k \cond y_{1:k}) \diff x_k \diff x_{k+1} - m_k^f \, (m^-_{k+1})^\trans \\
		&\approx \int x_k \, \big(f^M(x_k)\big)^\trans \, \mathrm{N} \big(x_{k} \cond m^f_{k}, P^f_{k} \big)  \diff x_k  - m_k^f \, (m^-_{k+1})^\trans \\
		&= D_{k+1}.
	\end{split}
	\label{equ:tme-gfs-D}
\end{equation}
We formally define the TME Gaussian filter and smoother in the following algorithm.

\begin{algorithm}[TME Gaussian filter and smoother]
	\label{alg:tme-gfs}
	The algorithm is the same as Algorithm~\ref{alg:gfs}, except that the computations for the prediction (i.e., $m^-_k$ and $P^-_k$) and cross-covariance (i.e., $D_{k+1}$) are replaced by Equations~\eqref{equ:tme-gfs-pred} and~\eqref{equ:tme-gfs-D}, respectively.
\end{algorithm}

The expectations in Equations~\eqref{equ:tme-gfs-pred} and~\eqref{equ:tme-gfs-D} are usually computed by quadrature integration methods (e.g., sigma-point methods), since the approximations $f^M$ and $Q^M$ are usually non-linear functions. 

\subsection{Filter stability}
\label{sec:tme-stability}
The filter stability in this context refers to the error bound of the filtering estimates in the mean-square sense. For Kalman filters, some classical stability results are already shown, for example, by~\citet{Jazwinski1970} and~\citet{Anderson1981}. As for non-linear filters, their stability analyse has also been studied extensively in recent decades. For example, \citet{Reif1999} analyse the stability of extended Kalman filters, while the stability of more general Gaussian filters are found in~\citet{Kazufumi2000, Xiong2006}. There are also stability analysis that are model-specific. For instance, \citet{Blomker2013} and~\citet{LawK2014} analyse the stability of a class of Gaussian filters on the Navier--Stokes equation and a Lorenz model, respectively. In the remainder of this section, we rely on the stability results in~\citet{Toni2020} which apply for a wide class of non-linear filters and non-linear state-space models including ours. 

In this section, we analyse the stability of the TME Gaussian filters (see, Algorithm~\ref{alg:tme-gfs}) that use sigma-point integration methods for computing the expectations in Equation~\eqref{equ:tme-gfs-pred}. This analysis is necessary, as it is important to know if the TME Gaussian filtering error -- which accumulates in time -- is in some sense bounded. The sources of the error include, for example, TME approximations, Gaussian approximations to the filtering posterior distributions, and numerical integration. 

To proceed, let
\begin{equation}
	X_k = \check{f}(X_{k-1}) + \check{q}(X_{k-1})
\end{equation}
stand for the \textit{exact} discretisation of the SDE in Equation~\eqref{equ:tme-cd-model} for $k=1,2,\ldots$, where $\check{f}(X_{k-1}) \coloneqq \expec{X_k \cond X_{k-1}}$, and $\check{q}(X_{k-1})$ is a zero-mean random variable whose conditional covariance is $\check{Q}(X_{k-1}) \coloneqq \cov{\check{q}(X_k) \cond X_{k-1}}$. The principle of TME Gaussian filters is such that the TME method approximates $X_k$ via the discretisation
\begin{equation}
	X_k \approx f^M(X_{k-1}) + q^M(X_{k-1}),
\end{equation}
where $q^M(X_{k-1}) \sim \mathrm{N}(0, Q^M(X_{k-1}))$. By Theorem~\ref{thm:tme} or Lemma~\ref{lemma:tme-1-2-moments} we have 
\begin{equation}
	\check{f}(X_{k-1}) = f^M(X_{k-1}) + R_{M}(X_{k-1}),
\end{equation}
where we abbreviate the remainder by $R_{M}(X_{k-1}) \coloneqq R_{M, \phi^{\mathrm{I}}}(X_{k-1}, \Delta t_k)$.

Now suppose that we perform TME Gaussian filtering on a linearly-observed state-space model
\begin{equation}
	\begin{split}
		X_k &= \check{f}(X_{k-1}) + \check{q}(X_{k-1}),\\
		Y_k &= H \, X_k + \xi_k, \quad \xi_k\sim \mathrm{N}(0, R),
	\end{split}
	\label{equ:tme-stability-model}
\end{equation}
defined on a probability space $(\Omega, \FF, \PP)$, where $H \in \R^{d_y \times d}$ and $R\in\R^{d_y\times d_y}$ are constant matrices. Here, we limited ourselves to linear measurement models in order to use the preceding results by~\citet{Toni2020}. 

In the following, sigma-point approximations of Gaussian integrals of the form $\int z(x) \, \mathrm{N}(x \cond m, P) \diff x$ are denoted by $\mathcal{S}_{m, P}(z)$. The sigma-point TME Gaussian filter is such that the predictive mean in Algorithm~\ref{alg:tme-gfs} becomes $m^-_k = \mathcal{S}_{m^f_{k-1}, P^f_{k-1}}(f^M)$. 
\begin{remark}
	Sigma-point approximations of the form $\mathcal{S}_{m, P}(z)$ are weighted summations of $z$ evaluated at integration nodes that are determined by $m$, $P$, and their quadrature rules. For details of these, see, for example, \citet{Sarkka2013}.
\end{remark}

We show the stability of sigma-point TME Gaussian filters in the sense that 
\begin{equation}
	\sup_{k\geq 1}  \expec*{\normbig{X_k - m^f_k}^2_2} < \infty.
	\label{equ:tme-stability-aim}
\end{equation}
\begin{remark}
	Note that if $m^f_k = \expec{X_k \cond Y_{1:k}}$ is obtained exactly, then the mean-square in the equation above is minimised, since $\expec{X_k \cond Y_{1:k}}$ is an orthogonal projection of $X_k$. But in practice, one can only hope for approximating $\expec{X_k \cond Y_{1:k}}$ by using, for example, TME Gaussian filters. The stability analysis here is devoted to show that the TME Gaussian filtering error has a finite (contractive) bound that depends on step $k$.
\end{remark}

We use the following assumptions.

\begin{assumption}
	\label{assump:tme-stability-sys}
	There exist constants $c_M \geq 0$, $c_{\check{q}} \geq 0$, and $c_P \geq 0$ such that $\sup_{k\geq 1} \norm{R_M(X_{k-1})}_2 \leq c_M$ $\PP$-almost surely, $\sup_{k\geq 1} \expec*{\trace*{\check{Q}(X_{k-1})}} \leq c_{\check{q}}$, and $\sup_{k\geq 1} \expecbig{\trace{P^f_k}} \leq c_P$.
\end{assumption}

\begin{assumption}
	\label{assump:tme-stability-sp}
	There exists $c_\mathcal{S}\geq 0$ such that
	\begin{equation}
		\normbig{f^M(x) - \mathcal{S}_{m, p}(f^M)}^2_2 \leq \normbig{ \jacob_x f^M(x) }^2_2 \, \norm{x - m}^2_2 + c_\mathcal{S} \trace{P},
	\end{equation}
	for all $x \in \R^d$, $m \in \R^d$, and positive semi-definite matrix $P \in \R^{d\times d}$.
\end{assumption}

\begin{assumption}
	\label{assump:tme-stability-k}
	There exists $c_K \geq 0$ such that $\sup_{k \geq 1} \norm{I - K_k \, H}_2 \leq c_K$ $\PP$-almost surely, and
	\begin{equation}
		c_f^2 \coloneqq c_K^2 \, \sup_{x \in \R^d} \! \normbig{\jacob_x f^M(x)}_2^2 < \frac{1}{4}.
	\end{equation}
\end{assumption}
Indeed, the assumptions above are in some sense restrictive. In particular, the TME remainder and the covariance of the transition density are required to be bounded by $c_M$ and $c_{\check{q}}$, respectively. In order to satisfy these assumptions, it is sufficient to require that the SDE coefficients are smooth enough and all their derivatives up to a certain order are uniformly bounded (e.g., the Bene\v{s} model in Example~\ref{example-tme-benes}). For more detailed explanations of these assumptions can be found in~\citet{ZhaoTME2020} and~\citet{Toni2020}.

The main result is shown in the following theorem.

\begin{theorem}[TME Gaussian filter stability]
	\label{thm:tme-stability}
	Suppose that Assumptions~\ref{assump:tme-stability-sys} to~\ref{assump:tme-stability-k} hold. Then the sigma-point TME Gaussian filter for system~\eqref{equ:tme-stability-model} is such that 
	\begin{equation}
		\begin{split}
			\expec*{\normbig{X_k - m^f_k}^2_2} &\leq \big(4\, c_f^2\big)^k \trace{P_0}  \\
			&\quad+ \frac{4 \, \big( c_K^2 \,\big(c_\mathcal{S} \, c_P + c_M^2 + c_{\check{q}}\big) + \trace{R}  \, c_P^2 \, \norm{H}_2^2 \, \norm{R^{-1}}^2_2 \big)}{1 - 4\, c_f^2}.
		\end{split}
		\label{equ:tme-stab-bound}
	\end{equation}
\end{theorem}
\begin{proof}
	Define $Z_k \coloneqq I - K_k \, H$. By substituting the sigma-point TME Gaussian filtering steps and the model~\eqref{equ:tme-stability-model} in $X_k - m^f_k$, we get
	\begin{equation}
		\begin{split}
			X_k - m^f_k &= \check{f}(X_{k-1}) + \check{q}(X_{k-1}) - m^-_k - K_k \, (Y_k - H \, m^-_k), \\
			&= Z_k \, \Big( \check{f}(X_{k-1}) - \mathcal{S}_{m^f_{k-1}, P^f_{k-1}}(f^M) + \check{q}(X_{k-1}) \Big) - K_k \, \xi_k \\
			&= Z_k \, \Big( f^M(X_{k-1}) - \mathcal{S}_{m^f_{k-1}, P^f_{k-1}}(f^M)\Big) \\
			&\quad+Z_k \, R_{M}(X_{k-1}) + Z_k \, \check{q}(X_{k-1}) - K_k \, \xi_k.
		\end{split}
	\end{equation}
	Then 
	\begin{align}
		\expec*{\normbig{X_k - m^f_k}_2^2} &\leq 4 \, \expec*{ \normBig{Z_k \, \Big( f^M(X_{k-1}) - \mathcal{S}_{m^f_{k-1}, P^f_{k-1}}(f^M)\Big)}^2_2 } \label{equ:tme-stab-e}\\
		&\quad+ 4 \, \Big( \expecbig{\norm{Z_k \, R_{M}(X_{k-1})}^2_2} + \expecbig{\norm{Z_k \, \check{q}(X_{k-1})}^2_2} + \expecbig{\norm{K_k \, \xi_k}_2^2} \Big). \nonumber
	\end{align}
	Now, by substituting the bounds 
	\begin{equation}
		\begin{split}
			\expec*{ \normBig{Z_k \, \Big( f^M(X_{k-1}) - \mathcal{S}_{m^f_{k-1}, P^f_{k-1}}(f^M)\Big)}^2_2 } &\leq c_f^2 \, \expec*{\normbig{X_{k-1} - m^f_{k-1}}_2^2} \\
			&\quad+ c_K^2 \, c_\mathcal{S} \, c_P, \\
			\expecbig{\norm{Z_k \, R_{M}(X_{k-1})}^2_2} &\leq c_K^2 \, c_M^2, \\
			\expecbig{\norm{Z_k \, \check{q}(X_{k-1})}^2_2} &\leq c_K^2 \, c_{\check{q}}, \\
			\expecbig{\norm{K_k \, \xi_k}_2^2} &\leq \trace{R}  \, c_P^2 \, \norm{H}_2^2 \, \norm{R^{-1}}^2_2,
		\end{split}
	\end{equation}
	following Assumptions~\ref{assump:tme-stability-sys},~\ref{assump:tme-stability-sp}, and~\ref{assump:tme-stability-k} into Equation~\eqref{equ:tme-stab-e}, we obtain the recursive inequality
	\begin{equation}
		\begin{split}
			\expec*{\normbig{X_k - m^f_k}_2^2} &\leq 4 \, c_f^2 \, \expec*{\normbig{X_{k-1} - m^f_{k-1}}_2^2} \\
			&\quad+ 4 \, \big( c_K^2 \,\big(c_\mathcal{S} \, c_P + c_M^2 + c_{\check{q}}\big) + \trace{R}  \, c_P^2 \, \norm{H}_2^2 \, \norm{R^{-1}}^2_2 \big).
		\end{split}
	\end{equation}
	The assumption $4 \, c_f^2< 1$ concludes the bound in Equation~\eqref{equ:tme-stab-bound}.
\end{proof}

Stability analysis of Gaussian smoothers that use the TME method can be found in~\citet{ZhaoZTMEsmoother2021}.

\subsection{Signal estimation and target tracking examples}
\label{sec:tme-gfs-examples}
This section presents a few applications of TME Gaussian filters and smoothers on signal estimation and target tracking problems. In the examples below, we uniformly use the expansion order $M=3$, and we use the Gauss--Hermite quadrature method (of order 3) to approximate the Gaussian expectations in Equations~\eqref{equ:tme-gfs-pred} and~\eqref{equ:tme-gfs-D}.

\begin{example}[Bene\v{s}]
	\label{example:tme-benes-filter-smoother}
	Consider the Bene\v{s} model in Example~\ref{example-tme-benes}, and also consider a linear measurement model $Y_k = X(t_k) + \xi_k$, where $\xi_k \sim \mathrm{N}(0, 0.5)$. We simulate a pair of a signal and its measurements at times $\lbrace t_k=0.01 \, k \colon k=0,1,\ldots, 500 \rbrace$, then we apply the TME-3 Gaussian filter and smoother to estimate the signal from the measurements. The results are plotted in Figure~\ref{fig:tme-benes-filter-smoother}.
\end{example}

\begin{figure}[t!]
	\centering
	\includegraphics[width=.99\linewidth]{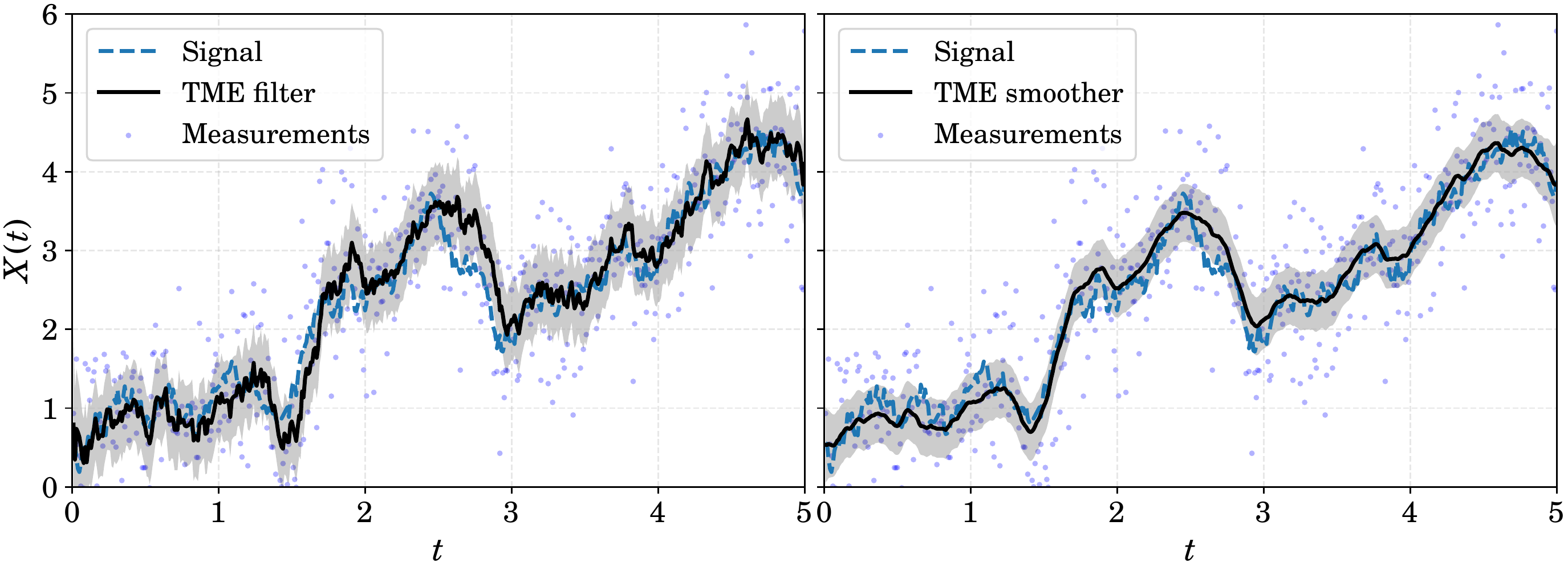}
	\caption{TME Gaussian filtering and smoothing for the Bene\v{s} model in Example~\ref{example:tme-benes-filter-smoother}. Shaded area stands for $0.95$ confidence interval.}
	\label{fig:tme-benes-filter-smoother}
\end{figure}

\begin{figure}[t!]
	\centering
	\includegraphics[width=.9\linewidth]{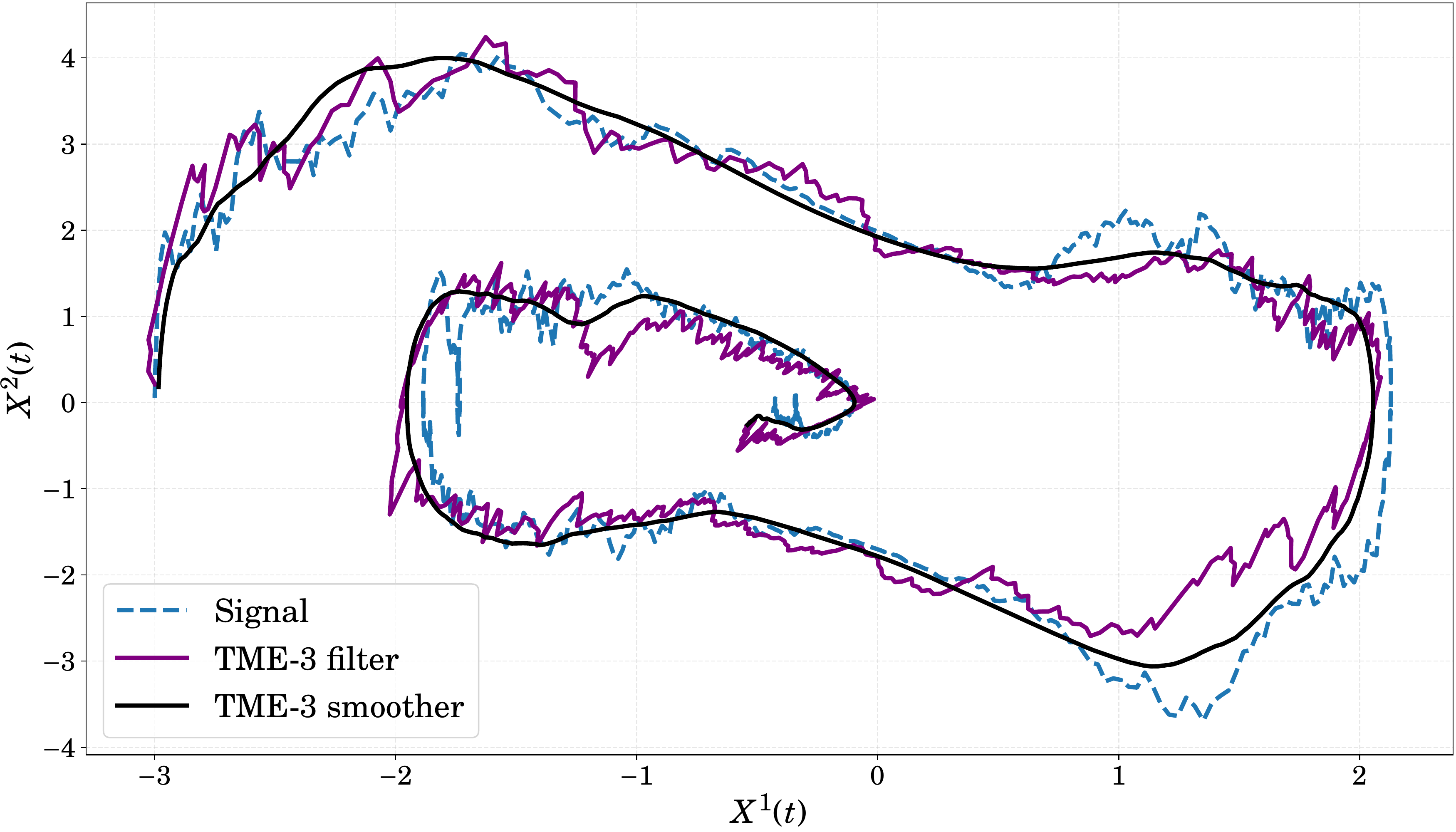} \\
	\includegraphics[width=.494\linewidth]{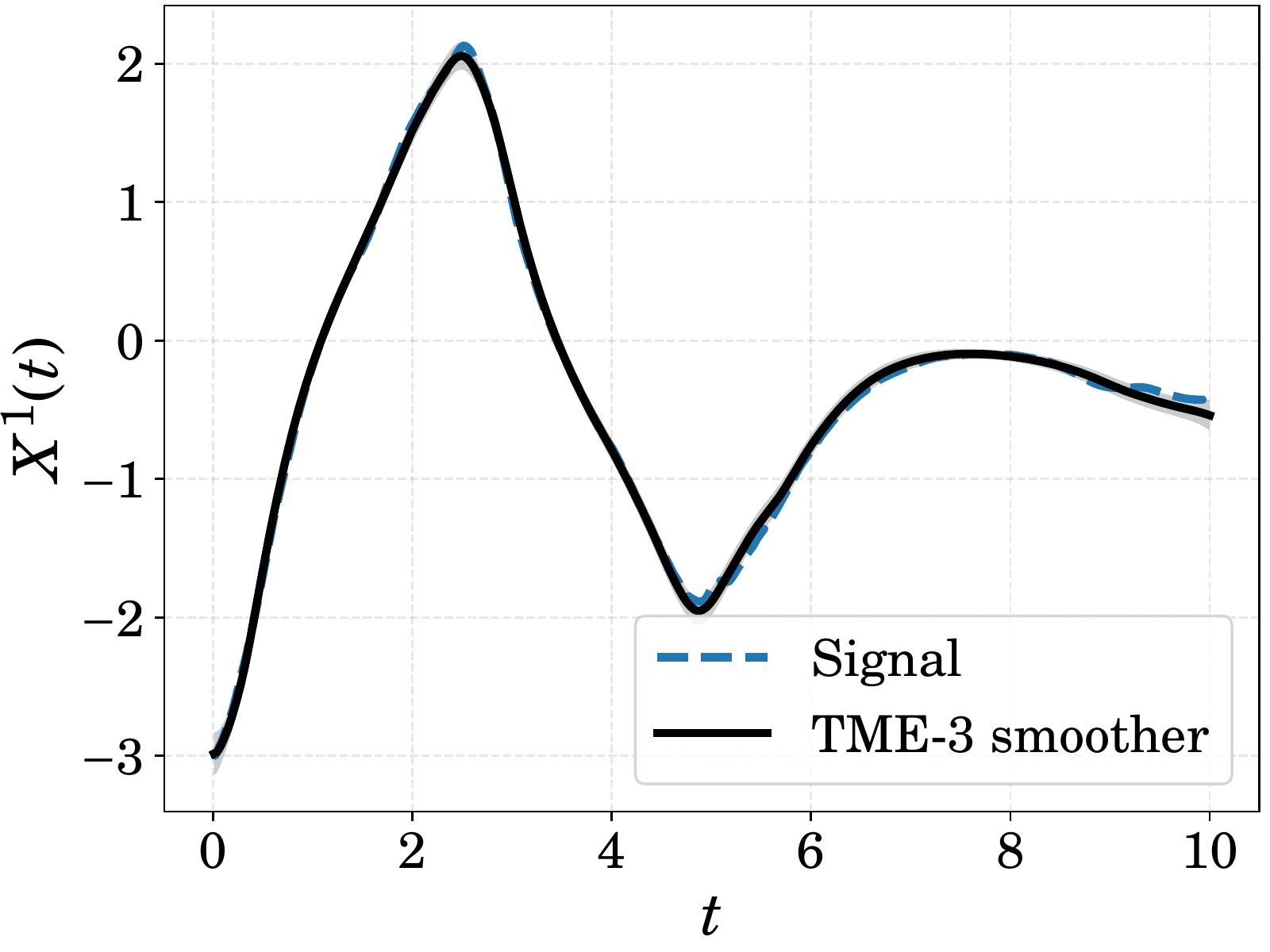}
	\includegraphics[width=.494\linewidth]{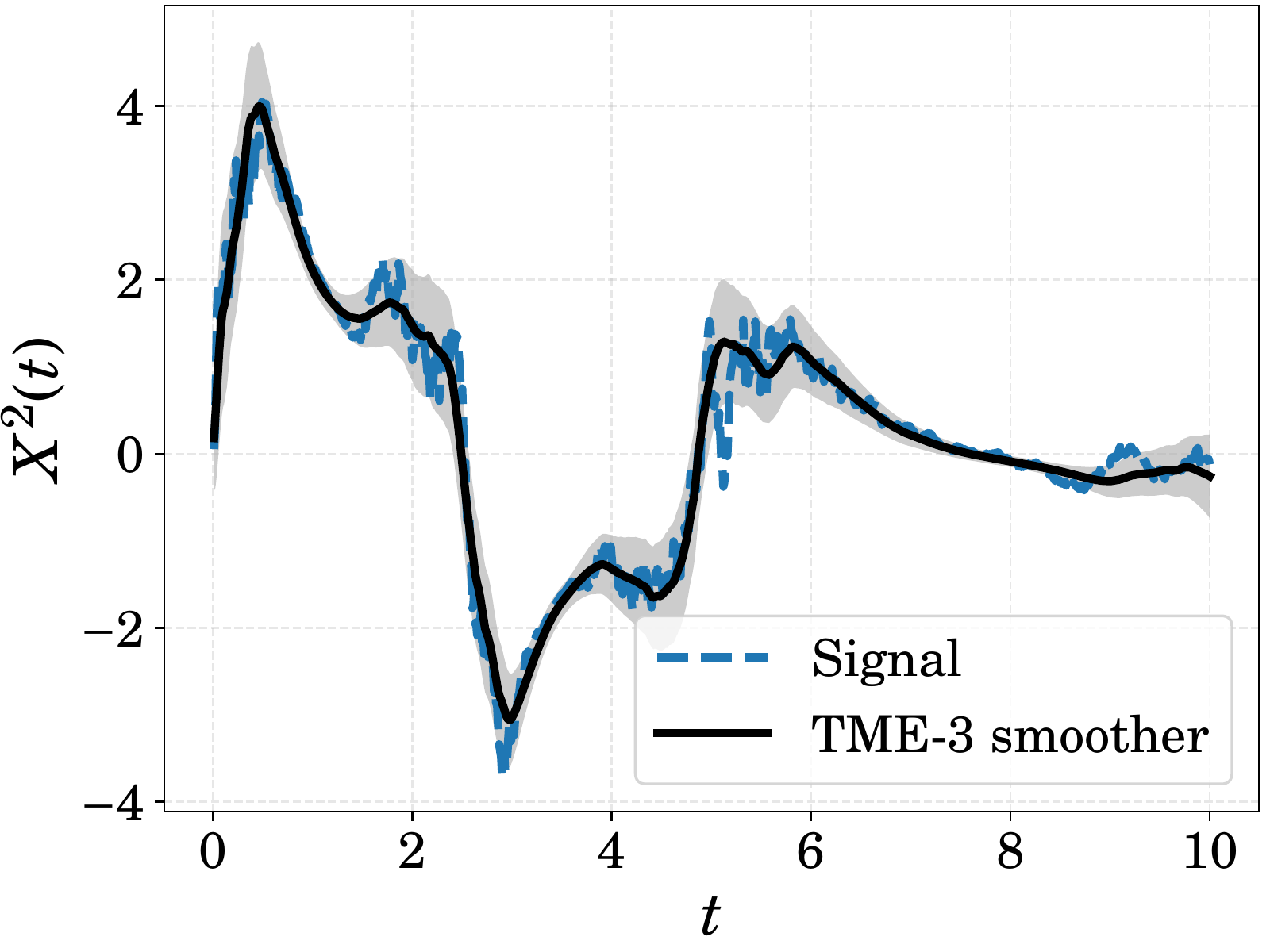}
	\caption{TME Gaussian filtering and smoothing for the Duffing--van der Pol model in Example~\ref{example:tme-duffing-filter-smoother}. }
	\label{fig:tme-duffing-filter-smoother}
\end{figure}

\begin{example}[Duffing--van der Pol]
	\label{example:tme-duffing-filter-smoother}
	Consider a continuous-discrete state-space model
	\begin{equation}
	\begin{split}
	\diff X^1(t) &= X^2(t) \diff t, \\
	\diff X^2(t) &= \Big(X^1(t) \, \big(\kappa - \big(X^1(t)\big)^2\big) - X^2(t) \Big) \diff t + X^1(t) \diff W(t), \\
	Y_k &= X^1(t_k) + 0.1 \, X^2(t_k) + \xi_k,
	\end{split}
	\end{equation}
	starting from the initial values $X^1(t_0) = -3$ and $X^2(t_0)=0$, where $\kappa = 2$ and $\xi_k \sim \mathrm{N}(0, 0.1)$. The non-linear multiplicative SDE above is called a modified stochastic Duffing--van der Pol oscillator equation~\citep{Lord_powell_shardlow_2014, Sarkka2019}. We simulate a pair of a signal and its measurements at times $\lbrace t_k=0.01\,k \colon k=0,1,\ldots,1000 \rbrace$. The results of the TME-3 Gaussian filtering and smoothing for this model is illustrated in Figure~\ref{fig:tme-duffing-filter-smoother}. 
	
	It is worth mentioning that the Euler--Maruyama-based Gaussian smoothing methods on this model may encounter numerical problems because the Euler--Maruyama scheme gives singular covariance approximation.
\end{example}

\begin{figure}[t!]
	\centering
	\includegraphics[width=.494\linewidth]{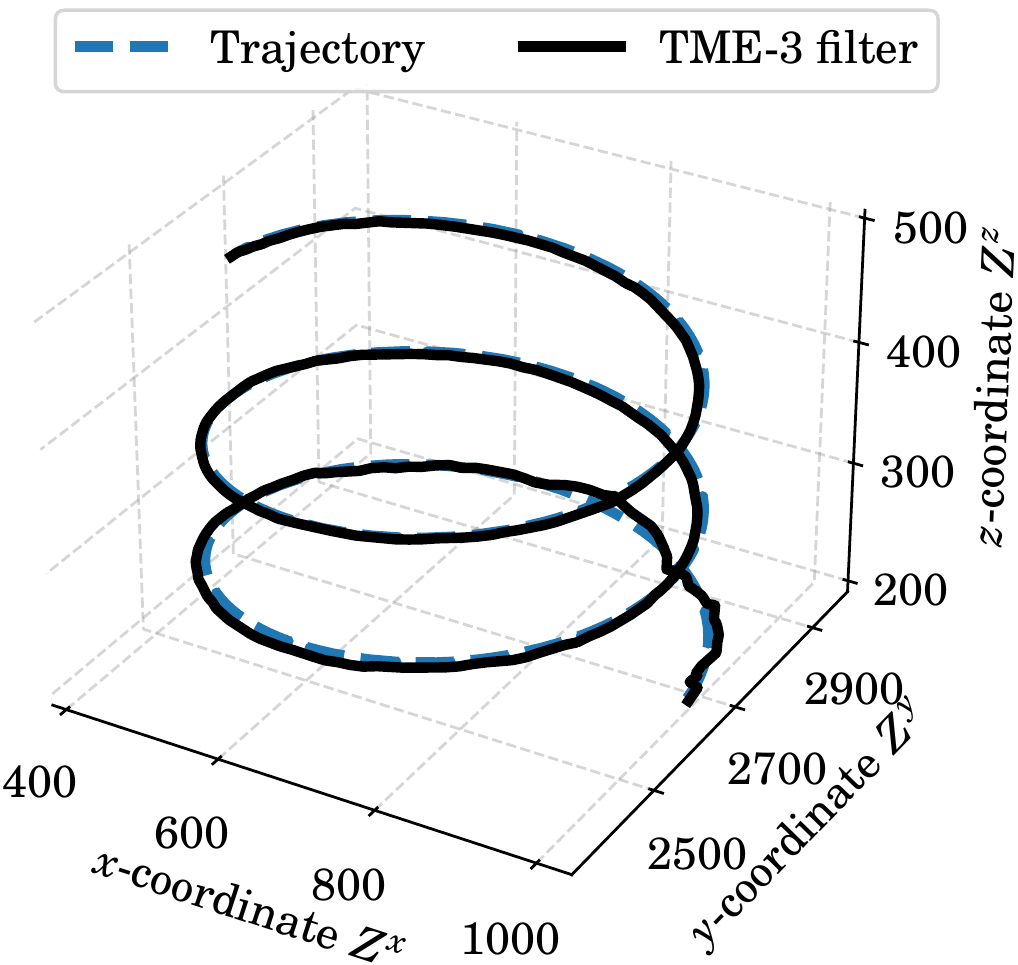}
	\includegraphics[width=.494\linewidth]{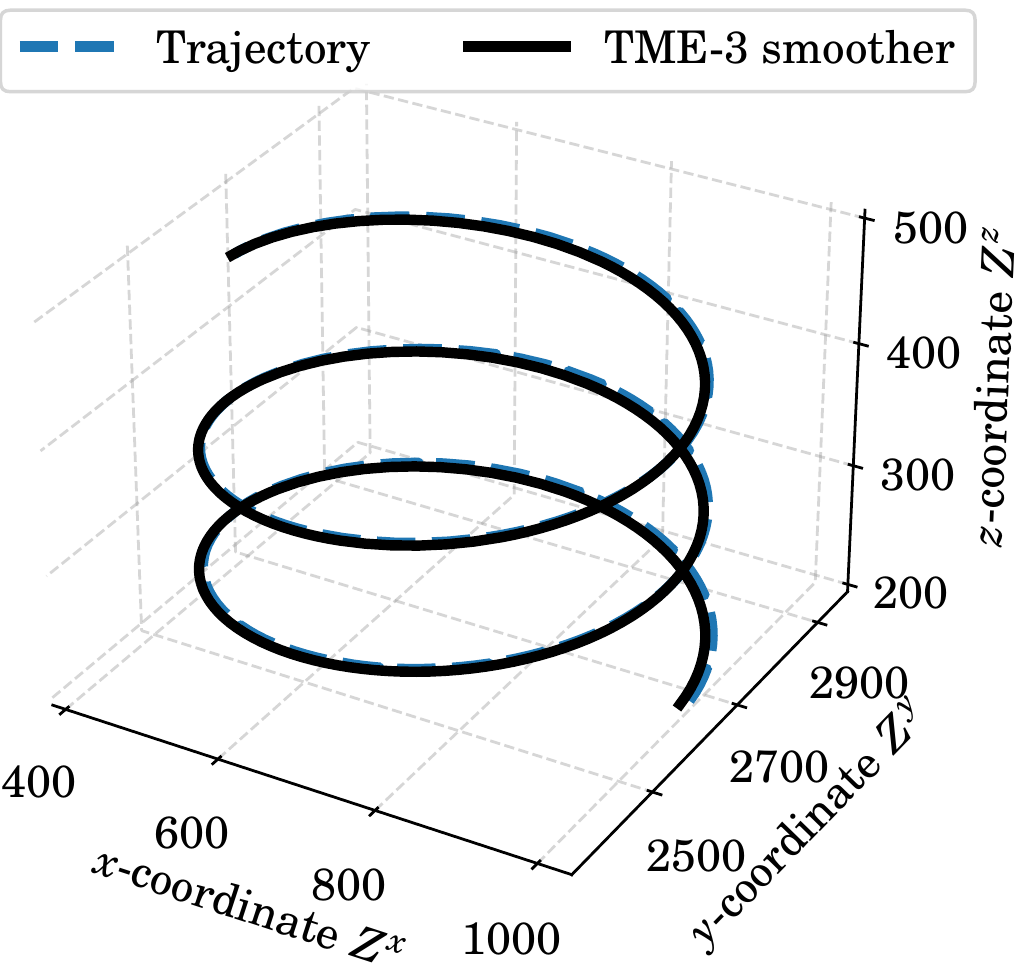}
	\caption{TME Gaussian filtering and smoothing for tracking a target moving as per the 3D coordinated turn model in Example~\ref{example:tme-ct-tracking}. }
	\label{fig:tme-ct3d}
\end{figure}

\begin{example}[3D coordinated turn tracking]
	\label{example:tme-ct-tracking}
	Consider a continuous-discrete model
	\begin{equation}
		\begin{split}
			\diff Z(t) &= a^{\mathrm{CT}}(Z(t)) \diff t + b^{\mathrm{CT}} \diff W(t), \\
			Y_k &= h^{\mathrm{CT}}(Z(t_k)) + \xi_k,
		\end{split}
	\end{equation}
	where the state $Z \colon \T \to \R^7 \coloneqq \begin{bmatrix} Z^x(t) & \dot{Z}^x(t) & Z^y(t) & \dot{Z}^y(t) & Z^z(t) & \dot{Z}^z(t) & \vartheta(t)\end{bmatrix}^\trans$ stands for the 3D Cartesian coordinate and the turn rate of a target. The SDE coefficients and the measurement function are defined by
	\begin{equation}
		\begin{split}
			a^{\mathrm{CT}}(Z(t)) &= 
			\begin{bmatrix}
				\dot{Z}^x(t) & - \vartheta(t) \, \dot{Z}^y(t) & \dot{Z}^y(t) & \vartheta(t) \, \dot{Z}^x(t) & \dot{Z}^x(t) & 0 & 0
			\end{bmatrix}^\trans, \\
			h^{\mathrm{CT}}(Z(t_k)) &= 
			\begin{bmatrix}
				\sqrt{(Z^x(t_k))^2 + (Z^y(t_k))^2 + (Z^z(t_k))^2} \\
				\arctan(Z^y(t_k) \, / \, Z^x(t_k)) \\
				\arctan\big(Z^z(t_k) \, / \, \sqrt{(Z^x(t_k))^2 + (Z^y(t_k))^2}\big)
			\end{bmatrix}.
		\end{split}
	\end{equation}
	For details of this model, we refer the reader to~\citet{ZhaoTME2020}. This model is widely used for manoeuvring target tracking and is very challenging for filtering and smoothing algorithms due to its high dimensionality and non-linearity~\citep{CDCKF2010, BarShalom2002}. A tracking example by using the TME-3 Gaussian filter and smoother is shown in Figure~\ref{fig:tme-ct3d}.
\end{example}

\chapter{State-space deep Gaussian processes}
\label{chap:dssgp}
In this chapter we introduce state-space deep Gaussian processes (SS-DGPs). The chapter starts with a brief review on Gaussian processes (GPs) and their state-space representations in Sections~\ref{sec:gp} and~\ref{sec:ssgp}, respectively. Subsequently, in Section~\ref{sec:ssdgp} deep Gaussian processes and their state-space representations (i.e., SS-DGPs) are defined. In Section~\ref{sec:deep-matern}, we introduce deep \matern processes which are a subclass of SS-DGPs where each GP element in the SS-DGP hierarchy is conditionally a \matern GP. Section~\ref{sec:ssdgp-reg} represents the SS-DGP regression problems as continuous-discrete filtering and smoothing problems. Finally, Section~\ref{sec:l1-r-dgp} illustrates how to solve $L^1$-regularised SS-DGP regression problems.

The content of this chapter is based on Publications~\cp{paperSSDGP} and~\cp{paperRNSSGP}. 

\section{Gaussian processes}
\label{sec:gp}
Gaussian processes (GPs) are a class of stochastic processes with finite-dimensional Gaussian distributions. More precisely, an $\R^d$-valued stochastic process $U\colon\T\to\R^d$ is said to be a GP if the following definition is satisfied. 
\begin{definition}[Gaussian process]
	\label{def:gp}
	A stochastic process $U\colon \T\to \R^d$ on some probability space is called a Gaussian process on $\T$ if for every integer $k\geq 0$ and real numbers $t_1<t_2<\cdots<t_k \in \T$, the random variables $U(t_1), U(t_2), \ldots, U(t_k)$ are jointly Gaussian~\citep[see, e.g.,][Section 2.9]{Karatzas1991}. 
\end{definition}
\begin{remark}
In the spirit of this thesis, we restrict Definition~\ref{def:gp} to temporal GPs only, however, it is possible to define GPs on more general domains~\citep{Carl2006GPML}.
\end{remark}

Since multivariate normal distributions are entirely determined by their means and covariances, Definition~\ref{def:gp} is usually interpreted by the shorthand notation
\begin{equation}
	U(t) \sim \GP\big(m(t), C(t, t')\big),
	\label{equ:gp-notation}
\end{equation}
where $m \colon \T \to\R^d$ and $C\colon \T \times \T \to\R^{d \times d}$ stand for the mean and covariance functions of the process, respectively. Under this notation, the finite-dimensional probability density function of $U$ at time instances $t_1, t_2,\ldots, t_k\in\T$ is given by
\begin{equation}
	\begin{split}
		&p_{U(t_1), U(t_2), \ldots, U(t_k)}(u_1, u_2, \ldots, u_k) \\
		&= \mathrm{N} 
		\begin{pmatrix}
			\begin{bmatrix}
				u_1 \\ u_2 \\ \vdots \\ u_k
			\end{bmatrix} \condBigg
			\begin{bmatrix}
				m(t_1) \\ m(t_2) \\ \vdots \\ m(t_k)
			\end{bmatrix}, 
			\begin{bmatrix}
				C(t_1, t_1) & C(t_1, t_2) & \cdots & C(t_1, t_k) \\
				C(t_2, t_1) & C(t_2, t_2) & \cdots & C(t_2, t_k) \\
				\vdots & \vdots & \ddots & \vdots \\
				C(t_k, t_1) & C(t_k, t_2) & \cdots & C(t_k, t_k)
			\end{bmatrix}
		\end{pmatrix}.
	\end{split}
\end{equation}
There are numerous possible choices for the covariance function $C$, and researchers and practitioners can choose one or the other depending on their applications. One of the most popular family of covariance functions to model continuous functions with varying degrees of regularity is given by the Whittle--Mat\'{e}rn covariance function~\citep{Matern1960}
\begin{equation}
	C_{\mathrm{Mat.}}(t, t') = \frac{\sigma^2 \, 2^{1-\nu}}{\varGamma(\nu)} \, \Bigg(\frac{\sqrt{2 \, \nu} \, \abs{t - t'}}{\ell}\Bigg)^\nu \, \mBesselsec \Bigg(\frac{\sqrt{2 \, \nu} \, \abs{t - t'}}{\ell}\Bigg),
	\label{equ:cov-matern}
\end{equation}
where $\ell$ and $\sigma$ are scale parameters, $\varGamma$ is the Gamma function, $\mBesselsec$ is the modified Bessel function of the second kind, and $\nu \in \big\lbrace \frac{1}{2}, \frac{3}{2}, \ldots \big\rbrace$. The smoothness of $U$ is controlled by the value of $\nu$. For example, if $\nu = \frac{3}{2}$, then $t\mapsto U(t)$ will be differentiable almost surely.  

Without loss of generality, we assume from now on that $m(t) = 0$ for all $t\in\T$, that is
\begin{equation}
	U(t) \sim \GP\big(0, C(t, t')\big).
\end{equation}
The covariance function $C$ thus entirely determines the properties of $U$, such as its continuity and stationarity. 
\begin{remark}
	\label{remark:stationary-gp}
	A stochastic process $U$ is said to be stationary if its finite-dimensional distribution is invariant under translation. That is,
	\begin{equation}
	p_{U(t_1+\tau),\ldots,U(t_k+\tau)}(u_1, \ldots, u_k) = p_{U(t_1),\ldots,U(t_k)}(u_1, \ldots, u_k), \nonumber
	\end{equation} 
	for all $k\geq 1$, $t_1<\cdots<t_k \in \T$, and $t_1+\tau<\cdots<t_k+\tau \in \T$~\citep{Karatzas1991}. Since GPs are characterised by their mean and covariance functions, we say that a zero-mean GP is stationary if $C(t+\tau, t'+\tau)$ does not depend on $\tau$, or equivalently, $C(t, t')$ is only a function of the time difference $t-t'$. 
\end{remark}
Stationarity is an important concept to keep in mind as many widely used covariance functions, such as the Mat\'{e}rn family and the radial basis function (RBF) lead to stationary GPs. However, as mentioned in Introduction, these stationary GPs might not be suitable priors in a number of applications.

\subsection*{Batch GP regression}
\label{sec:gp-reg-batch}
Consider a GP regression model
\begin{equation}
	\begin{split}
		U(t) &\sim \GP\big(0, C(t, t')\big), \\
		Y_k &= U(t_k) + \xi_k, \quad \xi_k \sim\mathrm{N}(0, \Xi_k),
	\end{split}
	\label{equ:batch-gp-reg}
\end{equation}
where we have a set of measurement data $y_{1:T} = \lbrace y_k \colon k=1,2,\ldots, T\rbrace$ at times $t_1, t_2, \ldots, t_T\in\T$. Let us denote by $C_{1:T}$ the (Gram) matrix obtained by evaluating the covariance function $C$ on the Cartesian grid $(t_1, t_2, \ldots, t_T) \times (t_1, t_2, \ldots, t_T)$. Let us also define $\Xi_{1:T} \coloneqq \diag{\Xi_1, \Xi_2, \ldots, \Xi_T}$ and $U_{1:T} \coloneqq \big\lbrace U(t_1), \allowbreak U(t_2), \ldots, U(t_T) \big\rbrace$.

Using Bayes' rule, and Gaussian identities, one can prove that the joint batch posterior probability density $p_{U_{1:T} \cond Y_{1:T}}(u_{1:T} \cond y_{1:T})$ is Gaussian. More specifically, the mean and covariance of the batch posterior density are given by
\begin{equation}
	\expec{U_{1:T} \cond y_{1:T}} = C_{1:T} \, (C_{1:T} + \Xi_{1:T})^{-1} \, y_{1:T}
	\label{equ:gp-reg-m}
\end{equation}
and
\begin{equation}
	\cov{U_{1:T} \cond y_{1:T}} = C_{1:T} - C_{1:T}\,(C_{1:T} + \Xi_{1:T})^{-1} \, C_{1:T},
	\label{equ:gp-reg-P}
\end{equation}
respectively. With a slight modification of the two equations above, the mean and covariance of the posterior density at test points (i.e., interpolation/extrapolation) can also be obtained in closed-form~\citep[see, e.g.,][Section 2.2]{Carl2006GPML}.
\begin{remark}
	The batch term in the name comes from the fact that the posterior density is solved jointly at $t_1, t_2, \ldots, t_T$ by using the full covariance matrix $C_{1:T}$.
\end{remark}
In Figure~\ref{fig:gp-fail}, we illustrate two examples of this batch GP regression using a \matern $\nu=3\,/\,2$ covariance function of the form in Equation~\eqref{equ:cov-matern}. 

It is worth pointing out two numerical problems of batch GP regressions. First, the computational complexity for computing the posterior mean and covariance is $O(T^3)$. This is due to the necessity of solving a system of equations of size $T$. This makes standard GP regression computationally expensive for large-scale datasets. This prompted researchers to introduce a number of alternatives (e.g., sparse GPs) that alleviate this prohibitive complexity. We refer the reader to Section~\ref{sec:literature-review} for a short review on this topic.

Another problem is that if the data times $t_1,t_2,\ldots, t_T$ are densely located (i.e., $t_{k} - t_{k-1}$ is numerically small for $k=1,2,\ldots,T$), or when some of them are identical, then the covariance matrix $C_{1:T}$ might be numerically close to singular~\citep[see, e.g.,][]{Ababou1994, Ranjan2011}. This numerical problem does not in general affect the numerical computation of Equations~\eqref{equ:gp-reg-m} and~\eqref{equ:gp-reg-P}, as the minimum eigenvalue of $C_{1:T} + \Xi_{1:T}$ is greater than the minimum eigenvalue of $\Xi_{1:T}$. However, it affects any procedure that needs to compute the matrix inverse of $C_{1:T}$ (e.g., maximum a posterior estimate of GP regression), or that the GP is observed without measurement noises~\citep{Ranjan2011}. It may also affect making samples from GP by means of Cholesky decomposition of $C_{1:T}$.

State-space representations of GPs, as formulated in the following section, can be used to avoid the two problems above.

\section{State-space Gaussian processes}
\label{sec:ssgp}
In this section, we introduce state-space representations of GPs. Namely, we represent GPs as solutions of linear SDEs. In order to do this, let $U \colon \T \to \R^d$ be the solution of a linear SDE
\begin{equation}
	\begin{split}
		\diff U(t) &= A(t) \, U(t) \diff t + B(t) \diff W(t), \\
		U(t_0) &= U_0,
	\end{split}
	\label{equ:ssgp}
\end{equation}
where coefficients $A \colon \T \to \R^{d \times d}$ and $B \colon \T \to \R^{d \times w}$ are deterministic time-dependent functions, $W\colon\T\to\R^w$ is a Wiener process, and ${U}_0 \sim \mathrm{N}(m_0, P_0)$. For the sake of simplicity, let us from now on assume that these coefficients are regular enough so that the SDE above is well-defined (see, e.g., Theorem~\ref{thm:linear-sde-solution} for sufficient conditions).

It turns out that the solution $U$ of the SDE in Equation~\eqref{equ:ssgp} verifies the axioms of Gaussian processes (given in Definition~\ref{def:gp}) on $\T$~\citep[see,][Section 5.6]{Karatzas1991}. Moreover, its mean $t\mapsto \expec{U(t)}$ and covariance $t\mapsto\cov{U(t)}$ functions are solutions of the following linear ODEs
\begin{equation}
	\begin{split}
		\frac{\diff m(t)}{\diff t} &= A(t) \, m(t), \\
		\frac{\diff P(t)}{\diff t} &= A(t) \, P(t) + P(t) \, A(t)^\trans + B(t) \, B(t)^\trans,
	\end{split}
\end{equation}
for every $t\in\T$ starting from the initial values $m(t_0) = m_0$ and $P(t_0) = P_0$. Note that if the initial mean $m_0 = 0$ then $m(t) = 0$ for all $t\in\T$, so that $U$ will be a zero-mean GP.

Compared to the batch GP representation in Equation~\eqref{equ:gp-notation}, state-space representations do not need to explicitly specify their mean and covariance functions. These functions are instead implicitly defined by the SDE coefficients. Finding the state-space representation of a GP with desired covariance function is possible as well~\citep[see, e.g.,][]{Hartikainen2010, Simo2013SSGP, Solin2016}.

Suppose that the coefficients $A(t) = A$ and $B(t) = B$ are constants, and all the real parts of the eigenvalues of $A$ are negative. Let $m_0=0$, and let $P_0$ solve the Lyapunov equation
\begin{equation}
	A \, P + P \, A^\trans + B \, B^\trans = 0,
	\label{equ:lyapunuv}
\end{equation}
then 
\begin{equation}
	U(t) \sim \mathrm{GP}\big(0, \cov{{U}(t), {U}(t')}\big)\nonumber
\end{equation} 
is a zero-mean stationary GP, and its covariance function is given by
\begin{equation}
	\cov{{U}(t), {U}(t')} = 
	\begin{cases}
		P_0 \, e^{\abs{t-t'} \, A^\trans}, & t < t' \in \T, \\
		e^{\abs{t-t'} \, A} \, P_0, & t' \leq t \in \T.
	\end{cases} \nonumber
\end{equation}
See, for example, \citet[][Theorem 6.7]{Karatzas1991}, \citet[][Section 3.7]{Pavliotis2014}, or~\citet[][Section 6.5]{Sarkka2019} for details.

\subsection*{State-space GP regression}
Due to the fact that state-space GPs (SS-GPs) are solutions of SDEs, they verify the Markov property. This is key in allowing to perform GP regression sequentially for $k=1,2,\ldots,T$ without computing the full covariance matrix $C_{1:T}$. To see this, let us consider a GP regression problem in the state-space form
\begin{equation}
	\begin{split}
		\diff {U}(t) &= A(t) \, {U}(t) \diff t + B(t) \diff W(t), \quad U(t_0) = U_0,\\
		Y_k &= H_k \, {U}(t_k) + \xi_k, \quad \xi_k \sim \mathrm{N}(0, \Xi_k).
	\end{split}
\end{equation}
We aim to compute the posterior density $p_{{U}(t_k) \cond Y_{1:T}}({u}_k \cond y_{1:T})$ for $k=1,2,\ldots, T$ instead of the joint posterior density $p_{U_{1:T} \cond Y_{1:T}}(u_{1:T} \cond y_{1:T})$. This state-space GP regression problem is equivalent to the continuous-discrete smoothing problem in Section~\ref{sec:rts}~\citep{Sarkka2019}. Therefore, one can apply Kalman filters and RTS smoothers (see, Algorithm~\ref{alg:kfs}) to carry out the state-space GP regression at hand exactly. Figure~\ref{fig:gp-kfs-eq} illustrates an example showing the equivalence between batch and state-space GP regression on a toy model. 

The computational complexity of state-space GP regression is $O(T)$, whereas the batch GP regression is $O(T^3)$. As an example, the batch and state-space GP regression shown in Figure~\ref{fig:gp-kfs-eq} take around $37$~s and $0.1$~s, respectively, on a computer with $T=10,000$ measurements.  Furthermore, by using prefix-sum algorithms, state-space GP regression can be solved in logarithmic $O(\log(T))$ time~\citep{Corenflos2021SSGP, Simo2021Parallel}. 

\begin{figure}[t!]
	\centering
	\includegraphics[width=.75\linewidth]{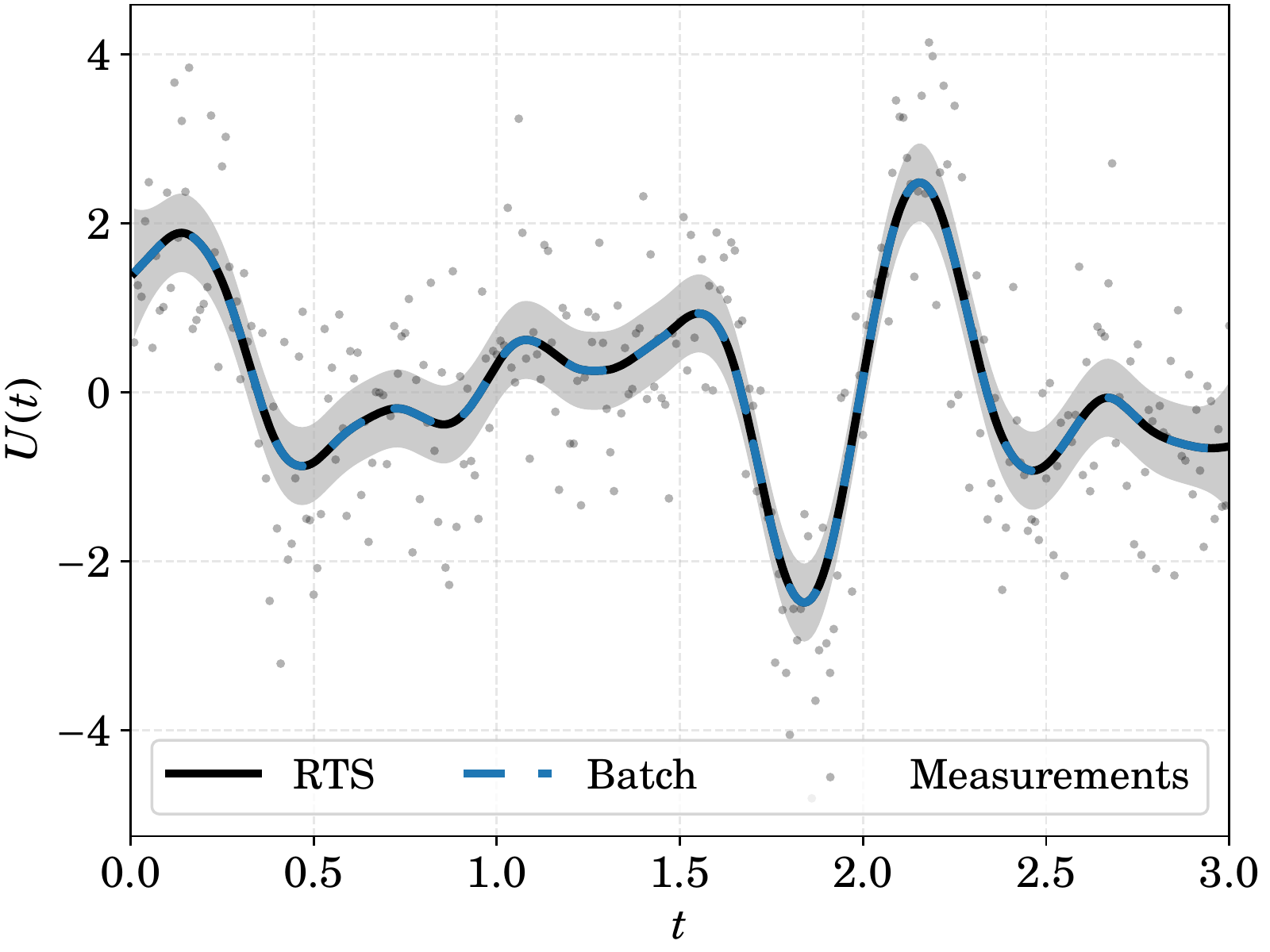}
	\caption{Batch and state-space GP regression on a toymodel with a \matern $\nu=3\, / \, 2$ covariance function and zero mean function. These two regression methods recover the same posterior densities (the lines and shaded area stand for the posterior mean and 0.95 confidence interval, respectively). }
	\label{fig:gp-kfs-eq}
\end{figure}

It is worth mentioning that not all GPs are Markov processes, hence, not all GPs have analytical state-space representations. As an example, \citet{Rozanov1977, Rozanov1982} show that certain stationary Gaussian processes/fields are Markovian if and only if the reciprocal of their spectral densities are polynomials. For instance, GPs using the RBF covariance function are not Markovian, but it is possible to approximate them up to an arbitrary order by using their approximate state-space representations~\citep{Simo2013SSGP}.

\section{State-space deep Gaussian processes (SS-DGPs)}
\label{sec:ssdgp}
State-space deep Gaussian processes (SS-DGPs) are stochastic processes that parametrise multiple conditional GPs hierarchically. This hierarchical construction makes SS-DGPs suitable priors for modelling irregular function in many applications. To see this, let us first consider a GP
\begin{equation}
	U(t) \sim \mathrm{GP}\big(0, C(t, t'; \ell(t)) \big),\nonumber
\end{equation}
where the covariance function $C(t, t';\ell(t))$ has an unknown (random) parameter $\ell(t)\in\R_{>0}$ (i.e., a time-varying length scale). When the parameter $\ell(t)$ does not depend on $t$, it can be assigned by human experts or automatically learnt from data by, for example, maximum likelihood estimation (MLE), maximum a posteriori (MAP), variational inference, or Markov chain Monte Carlo (MCMC)~\citep{Carl2006GPML}. However, the assumption that $\ell$ being independent of $t$ might not be reasonable for a number of applications that exhibit time-varying features. A way to mitigate this issue is, for example, to consider putting another GP prior on the length scale parameter, that is
\begin{equation}
	\ell(t) \sim \mathrm{GP}\big(0, C(t, t'; \ell_2) \big),\nonumber
\end{equation}
where $\ell_2$ is another length scale parameter. This hierarchical feature is meaningful in the sense that it allows the characteristics of $U$ to change over time, since its length scale $t \mapsto \ell(t)$ now is a stochastic process of $t$. It is then of interest to ask if this hierarchical recursion can be continued up to a given depth $L$:
\begin{equation}
	\begin{split}
		\ell_2(t) &\sim \mathrm{GP}\big(0, C(t, t'; \ell_3(t)) \big), \\
		\ell_3(t) &\sim \mathrm{GP}\big(0, C(t, t'; \ell_4(t)) \big), \\
		&\vdots\\
		\ell_{L}(t) &\sim \mathrm{GP}\big(0, C(t, t'; \ell_{L+1}) \big),
		\label{equ:cascading-ell}
	\end{split}
\end{equation}
where the final leaf $\ell_{L+1}$ is a constant.
This construction leads to a class of deep Gaussian processes (DGPs, see, Section~\ref{sec:literature-review} for background). 

In the rest of this chapter, we formulate the hierarchy in Equation~\eqref{equ:cascading-ell} in more abstract form in order to define DGPs. Thereupon we leverage this definition to represent DGPs as solutions of SDEs in order to arrive at SS-DGPs.

\subsection*{Deep Gaussian processes}
\label{sec:dgp-motiv-def}
Equation~\eqref{equ:cascading-ell} exemplifies a DGP where the length scale parameters \textit{only} are considered as GPs. In graph theory, this type of DGP hierarchy corresponds to a path graph~\citep{Gross2019} where the length scale parameters are vertices that ordered in a line/path. This type of DGP construction is the most studied case in the parametrisation-based DGP community~\citep{Roininen2016, Salimbeni2017ns, Emzir2020}. 

However, in principle, a GP can take any number of parameters. Thus, in order to abstract DGPs, we need to think of a DGP as a joint process defined over a set of conditional GPs. These conditional GPs are not necessarily limited to representing length scale parameters only. In order to do so, we need to introduce an indexing system and a few notations. Let $U^i_j \colon \T \to \R^{d_i}$ denote a GP indexed by an integer $i\in\N$. This superscript $i$ means that $U^i_j$ is the $i$-th GP element in a (yet to be defined) collection of GPs. The subscript $j$ in $U^i_j$ means that the GP $U^i_j$ is a parent of the $j$-th GP element (i.e., the $j$-th GP element is parametrised by the $i$-th element). The terminology ``parent'' follows from probabilistic graph model conventions~\citep{Koller2009}. The fact that GP element does not have any child means that it does not parametrise any other GP therefore, we define its subscript $j$ to be $j=0$. This is always true for the first element $U^1_0$ as we shall see later in the definition of the collection of these GP elements. 

Additionally, in order to give a well-defined graph, we restrict $j<i$ so that a GP element can only parametrise \emph{one} of its \emph{preceding} elements. This implies that a GP element can have multiple parents but no more than one child. Without this restriction, one might have two elements, for instance, $U^2_1$ and $U^1_2$ depending on each other, that is not within the scope of this thesis.
\begin{remark}
	The set of dependencies between the conditional GPs can be thought of as a collection of directed trees where the head of each tree has $j$ subscript $j=0$, and the notation $U_j^i$ implies that there is an edge pointing from $U_j^i$ to $U^j_k$ for some $k$. See, Figure~\ref{fig:dgp-examples-graph} for an illustration.
\end{remark}

Suppose that we have $L$ GPs $U^1_0, U^2_{j_2},\ldots, U^L_{j_L}$ and a set $J = \lbrace j_i \in \N \colon i=1,2\ldots,L, \, \, 0\leq j_i < i\rbrace$ that describes the conditional dependencies of these GPs. We define the collection of all these GPs as
\begin{equation}
	\mathcal{V} \coloneqq \mathcal{V}^L_J = \big\lbrace U^i_{j_i} \colon i= 1,2,\ldots, L, \,\, j_i \in J \big\rbrace,
	\label{equ:dgp-set-V}
\end{equation}
and we will call these conditional GPs the GP elements of $\mathcal{V}$. Based on this collection, we define a DGP as a vector-valued process composed of all the GP elements in $\mathcal{V}$.

\begin{definition}[Deep Gaussian process]
	\label{def:dgp}
	Let $\mathcal{V}$ be a collection of $L$ $\R^{d_i}$-valued conditional GPs defined by Equation~\eqref{equ:dgp-set-V}. An $\R^{\sum^L_{i=1} d_i}$-valued stochastic process $V \colon \T \to \R^{\sum^L_{i=1} d_i}$ is said to be a deep Gaussian process on $\T$ with respect to $\mathcal{V}$ if $V$ is a permutation of all the elements of $\mathcal{V}$.
\end{definition}

\begin{remark}
	Note that in the special case $L=1$, a DGP reduces to a standard GP. 
\end{remark}

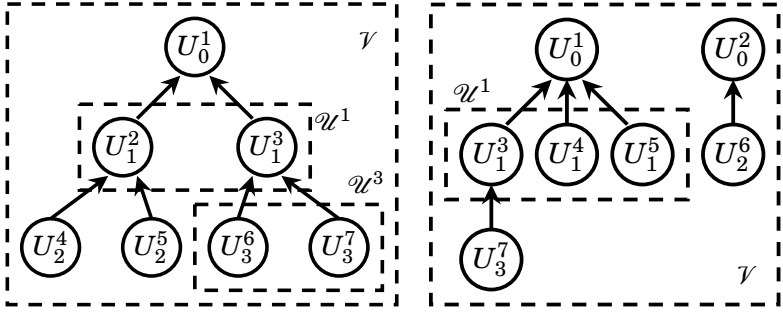
\begin{figure}[t!]
	\centering
	\resizebox{.42\linewidth}{!}{%
		\tikzset{every picture/.style={line width=0.75pt}} 

\begin{tikzpicture}[x=0.6pt,y=0.6pt,yscale=-1,xscale=1]

\draw  [line width=1.5]  (200,30) .. controls (200,18.95) and (208.95,10) .. (220,10) .. controls (231.05,10) and (240,18.95) .. (240,30) .. controls (240,41.05) and (231.05,50) .. (220,50) .. controls (208.95,50) and (200,41.05) .. (200,30) -- cycle ;
\draw  [line width=1.5]  (150,100) .. controls (150,88.95) and (158.95,80) .. (170,80) .. controls (181.05,80) and (190,88.95) .. (190,100) .. controls (190,111.05) and (181.05,120) .. (170,120) .. controls (158.95,120) and (150,111.05) .. (150,100) -- cycle ;
\draw  [line width=1.5]  (250,100) .. controls (250,88.95) and (258.95,80) .. (270,80) .. controls (281.05,80) and (290,88.95) .. (290,100) .. controls (290,111.05) and (281.05,120) .. (270,120) .. controls (258.95,120) and (250,111.05) .. (250,100) -- cycle ;
\draw [line width=1.5]    (207.17,52.83) -- (180,80) ;
\draw [shift={(210,50)}, rotate = 135] [fill={rgb, 255:red, 0; green, 0; blue, 0 }  ][line width=0.08]  [draw opacity=0] (13.4,-6.43) -- (0,0) -- (13.4,6.44) -- (8.9,0) -- cycle    ;
\draw [line width=1.5]    (232.83,52.83) -- (260,80) ;
\draw [shift={(230,50)}, rotate = 45] [fill={rgb, 255:red, 0; green, 0; blue, 0 }  ][line width=0.08]  [draw opacity=0] (13.4,-6.43) -- (0,0) -- (13.4,6.44) -- (8.9,0) -- cycle    ;
\draw  [line width=1.5]  (100,170) .. controls (100,158.95) and (108.95,150) .. (120,150) .. controls (131.05,150) and (140,158.95) .. (140,170) .. controls (140,181.05) and (131.05,190) .. (120,190) .. controls (108.95,190) and (100,181.05) .. (100,170) -- cycle ;
\draw  [line width=1.5]  (170,170) .. controls (170,158.95) and (178.95,150) .. (190,150) .. controls (201.05,150) and (210,158.95) .. (210,170) .. controls (210,181.05) and (201.05,190) .. (190,190) .. controls (178.95,190) and (170,181.05) .. (170,170) -- cycle ;
\draw  [line width=1.5]  (230,170) .. controls (230,158.95) and (238.95,150) .. (250,150) .. controls (261.05,150) and (270,158.95) .. (270,170) .. controls (270,181.05) and (261.05,190) .. (250,190) .. controls (238.95,190) and (230,181.05) .. (230,170) -- cycle ;
\draw  [line width=1.5]  (300,170) .. controls (300,158.95) and (308.95,150) .. (320,150) .. controls (331.05,150) and (340,158.95) .. (340,170) .. controls (340,181.05) and (331.05,190) .. (320,190) .. controls (308.95,190) and (300,181.05) .. (300,170) -- cycle ;
\draw [line width=1.5]    (156.8,122.4) -- (120,150) ;
\draw [shift={(160,120)}, rotate = 143.13] [fill={rgb, 255:red, 0; green, 0; blue, 0 }  ][line width=0.08]  [draw opacity=0] (13.4,-6.43) -- (0,0) -- (13.4,6.44) -- (8.9,0) -- cycle    ;
\draw [line width=1.5]    (181.26,123.79) -- (190,150) ;
\draw [shift={(180,120)}, rotate = 71.57] [fill={rgb, 255:red, 0; green, 0; blue, 0 }  ][line width=0.08]  [draw opacity=0] (13.4,-6.43) -- (0,0) -- (13.4,6.44) -- (8.9,0) -- cycle    ;
\draw [line width=1.5]    (258.74,123.79) -- (250,150) ;
\draw [shift={(260,120)}, rotate = 108.43] [fill={rgb, 255:red, 0; green, 0; blue, 0 }  ][line width=0.08]  [draw opacity=0] (13.4,-6.43) -- (0,0) -- (13.4,6.44) -- (8.9,0) -- cycle    ;
\draw [line width=1.5]    (283.2,122.4) -- (320,150) ;
\draw [shift={(280,120)}, rotate = 36.87] [fill={rgb, 255:red, 0; green, 0; blue, 0 }  ][line width=0.08]  [draw opacity=0] (13.4,-6.43) -- (0,0) -- (13.4,6.44) -- (8.9,0) -- cycle    ;
\draw  [dash pattern={on 5.63pt off 4.5pt}][line width=1.5]  (140,70) -- (300,70) -- (300,130) -- (140,130) -- cycle ;
\draw  [dash pattern={on 5.63pt off 4.5pt}][line width=1.5]  (220,140) -- (350,140) -- (350,200) -- (220,200) -- cycle ;
\draw  [dash pattern={on 5.63pt off 4.5pt}][line width=1.5]  (90,0) -- (360,0) -- (360,210) -- (90,210) -- cycle ;

\draw (340.5,25) node    {$\mathcal{V}$};
\draw (316,80.5) node    {$\mathcal{U}^{1}$};
\draw (339,124.5) node    {$\mathcal{U}^{3}$};
\draw (220,30) node  [font=\large]  {$U_{0}^{1}$};
\draw (170,100) node  [font=\large]  {$U_{1}^{2}$};
\draw (270,100) node  [font=\large]  {$U_{1}^{3}$};
\draw (120,170) node  [font=\large]  {$U_{2}^{4}$};
\draw (190,170) node  [font=\large]  {$U_{2}^{5}$};
\draw (250,170) node  [font=\large]  {$U_{3}^{6}$};
\draw (320,170) node  [font=\large]  {$U_{3}^{7}$};

\end{tikzpicture}
	}
	\resizebox{.379\linewidth}{!}{%
		\tikzset{every picture/.style={line width=0.75pt}} 

\begin{tikzpicture}[x=0.6pt,y=0.6pt,yscale=-1,xscale=1]
	
	\draw  [line width=1.5]  (130,70) .. controls (130,58.95) and (138.95,50) .. (150,50) .. controls (161.05,50) and (170,58.95) .. (170,70) .. controls (170,81.05) and (161.05,90) .. (150,90) .. controls (138.95,90) and (130,81.05) .. (130,70) -- cycle ;
	\draw  [line width=1.5]  (80,140) .. controls (80,128.95) and (88.95,120) .. (100,120) .. controls (111.05,120) and (120,128.95) .. (120,140) .. controls (120,151.05) and (111.05,160) .. (100,160) .. controls (88.95,160) and (80,151.05) .. (80,140) -- cycle ;
	\draw  [line width=1.5]  (180,140) .. controls (180,128.95) and (188.95,120) .. (200,120) .. controls (211.05,120) and (220,128.95) .. (220,140) .. controls (220,151.05) and (211.05,160) .. (200,160) .. controls (188.95,160) and (180,151.05) .. (180,140) -- cycle ;
	\draw [line width=1.5]    (137.17,92.83) -- (110,120) ;
	\draw [shift={(140,90)}, rotate = 135] [fill={rgb, 255:red, 0; green, 0; blue, 0 }  ][line width=0.08]  [draw opacity=0] (13.4,-6.43) -- (0,0) -- (13.4,6.44) -- (8.9,0) -- cycle    ;
	\draw [line width=1.5]    (162.83,92.83) -- (190,120) ;
	\draw [shift={(160,90)}, rotate = 45] [fill={rgb, 255:red, 0; green, 0; blue, 0 }  ][line width=0.08]  [draw opacity=0] (13.4,-6.43) -- (0,0) -- (13.4,6.44) -- (8.9,0) -- cycle    ;
	\draw [line width=1.5]    (150,94) -- (150,120) ;
	\draw [shift={(150,90)}, rotate = 90] [fill={rgb, 255:red, 0; green, 0; blue, 0 }  ][line width=0.08]  [draw opacity=0] (13.4,-6.43) -- (0,0) -- (13.4,6.44) -- (8.9,0) -- cycle    ;
	\draw  [line width=1.5]  (130,140) .. controls (130,128.95) and (138.95,120) .. (150,120) .. controls (161.05,120) and (170,128.95) .. (170,140) .. controls (170,151.05) and (161.05,160) .. (150,160) .. controls (138.95,160) and (130,151.05) .. (130,140) -- cycle ;
	\draw  [line width=1.5]  (80,210) .. controls (80,198.95) and (88.95,190) .. (100,190) .. controls (111.05,190) and (120,198.95) .. (120,210) .. controls (120,221.05) and (111.05,230) .. (100,230) .. controls (88.95,230) and (80,221.05) .. (80,210) -- cycle ;
	\draw  [line width=1.5]  (240,70) .. controls (240,58.95) and (248.95,50) .. (260,50) .. controls (271.05,50) and (280,58.95) .. (280,70) .. controls (280,81.05) and (271.05,90) .. (260,90) .. controls (248.95,90) and (240,81.05) .. (240,70) -- cycle ;
	\draw  [line width=1.5]  (240,140) .. controls (240,128.95) and (248.95,120) .. (260,120) .. controls (271.05,120) and (280,128.95) .. (280,140) .. controls (280,151.05) and (271.05,160) .. (260,160) .. controls (248.95,160) and (240,151.05) .. (240,140) -- cycle ;
	\draw [line width=1.5]    (100,164) -- (100,171) -- (100,190) ;
	\draw [shift={(100,160)}, rotate = 90] [fill={rgb, 255:red, 0; green, 0; blue, 0 }  ][line width=0.08]  [draw opacity=0] (13.4,-6.43) -- (0,0) -- (13.4,6.44) -- (8.9,0) -- cycle    ;
	\draw [line width=1.5]    (260,94) -- (260,120) ;
	\draw [shift={(260,90)}, rotate = 90] [fill={rgb, 255:red, 0; green, 0; blue, 0 }  ][line width=0.08]  [draw opacity=0] (13.4,-6.43) -- (0,0) -- (13.4,6.44) -- (8.9,0) -- cycle    ;
	\draw  [dash pattern={on 5.63pt off 4.5pt}][line width=1.5]  (60,40) -- (290,40) -- (290,240) -- (60,240) -- cycle ;
	\draw  [dash pattern={on 5.63pt off 4.5pt}][line width=1.5]  (70,110) -- (230,110) -- (230,170) -- (70,170) -- cycle ;
	
	\draw (150,70) node  [font=\large]  {$U_{0}^{1}$};
	\draw (100,140) node  [font=\large]  {$U_{1}^{3}$};
	\draw (200,140) node  [font=\large]  {$U_{1}^{5}$};
	\draw (150,140) node  [font=\large]  {$U_{1}^{4}$};
	\draw (100,210) node  [font=\large]  {$U_{3}^{7}$};
	\draw (260,70) node  [font=\large]  {$U_{0}^{2}$};
	\draw (260,140) node  [font=\large]  {$U_{2}^{6}$};
	\draw (73,87.4) node [anchor=north west][inner sep=0.75pt]    {$\mathcal{U}^{1}$};
	\draw (261,212.4) node [anchor=north west][inner sep=0.75pt]    {$\mathcal{V}$};

\end{tikzpicture}
	}
	\caption{Two DGP ($L=7$) examples in graph illustration.}
	\label{fig:dgp-examples-graph}
\end{figure}

It is also natural to define another set 
\begin{equation}
	\mathcal{U}^i = \big\lbrace U^k_{i} \in \mathcal{V} \colon k = 1,2,\ldots, L \big\rbrace
	\label{equ:dgp-set-Ui}
\end{equation}
that collects all the parent GPs of the $i$-th GP element in $\mathcal{V}$. It follows from Lemma~\ref{lemma:dgp-graph-partition} that all the collections of parent GPs form a partition of the set of all GP elements.

\begin{lemma}[Partition]
	\label{lemma:dgp-graph-partition}
	Let $ \mathcal{U}^0, \mathcal{U}^1, \ldots, \mathcal{U}^L$ be collections of parent GPs as defined by Equation~\eqref{equ:dgp-set-Ui}. These collections satisfy the axiom of a partition.
	\begin{enumerate}[I.]
		\item (Pairwise disjointness) For every $m,n\in\lbrace 0, 1,\ldots,L-1\rbrace$ and $m\neq n$,
		\begin{equation}
			\mathcal{U}^m \cap \mathcal{U}^n = \emptyset.
		\end{equation}
		\item (Exhaustiveness)
		\begin{equation}
			\bigcup_{i=0}^{L-1} \mathcal{U}^i = \mathcal{V}.
	\end{equation}
	\end{enumerate}
\end{lemma}
\begin{remark}
	Note that $\mathcal{U}^L=\emptyset$ by construction.
\end{remark}
\begin{proof}
	In order to prove the first property, suppose that there exists a pair $m,n\in\lbrace 0, 1,\ldots,L-1\rbrace$ and $m\neq n$ such that $\mathcal{U}^m \cap \mathcal{U}^n$ is non-empty. This implies that there is a GP element pointing simultaneously to $U^m_{j_m}$ and to $U^n_{j_n}$, which violates the definition of a GP element. 

	Following Equations~\eqref{equ:dgp-set-V} and~\eqref{equ:dgp-set-Ui}, we have $\bigcup^{L-1}_{i=0} \mathcal{U}^i \subseteq \mathcal{V}$. Suppose that there exists a GP element that is in $\mathcal{V}$ but not in $\bigcup^{L-1}_{i=0} \mathcal{U}^i$, then this GP element is not a parent of any GP elements (i.e., it must be in $\mathcal{U}^0$) which violates the hypothesis.
\end{proof}

We mention that the indexing system for DGPs here is simplified compared to Publication~\cp{paperSSDGP} which additionally used an unnecessary index denoting the depth of the GP element in the hierarchy. Figure~\ref{fig:dgp-examples-graph} illustrates two graphical examples of DGPs to clarify the indexing and notations used here. 

\subsection*{Batch representations of DGPs}
\label{sec:batch-ss-dgps}
Following Definition~\ref{def:dgp}, we can use the following shorthand batch notation to represent a DGP $V\colon \T\to\R^{\sum^L_{i=1} d_i}$ with $L$ conditional GPs: 

\begin{equation}
	\begin{split}
		U^1_0(t) \condbig \mathcal{U}^1 &\sim \mathrm{GP}\big(0, C^1(t, t'; \mathcal{U}^1)\big), \\
		U^2_{j_2}(t) \condbig \mathcal{U}^2 &\sim \mathrm{GP}\big(0, C^2(t, t'; \mathcal{U}^2)\big), \\
		&\vdots\\
		U^i_{j_i}(t) \condbig \mathcal{U}^i &\sim \mathrm{GP}\big(0, C^i(t, t'; \mathcal{U}^i)\big), \\
		&\vdots\\
		U^L_{j_L}(t) &\sim \mathrm{GP}\big(0, C^L(t, t')\big),
		\label{equ:batch-dgp}
	\end{split}
\end{equation}
where $C^i\colon \T \times \T \to \R^{d_i \times d_i}$ is the covariance function of $U^i_{j_i}$ parametrised by the GPs in $\mathcal{U}^i$, and
\begin{equation}
	V(t) \coloneqq \begin{bmatrix} U^1_0(t) & U^2_{j_2}(t) & \cdots & U^L_{j_L}(t) \end{bmatrix}^\trans. \nonumber
\end{equation}
Thanks to the conditional hierarchy structure of the model, the probability density function
\begin{equation}
	\begin{split}
		p_{V(t)}(v, t) &\coloneqq p_{U^1_0(t),\ldots,U^L_{j_L}(t)}\big(u^1_0,\ldots,u^L_{j_L}, t\big) \\
		&= \prod^L_{i=1} p_{U^i_{j_i}(t) \cond \mathcal{U}^i} \big(u^i_{j_i}, t \cond \mathcal{U}^i \big)
		\label{equ:batch-dgp-density}
	\end{split}
\end{equation}
of $V$ can factorise over the probability densities of the GP elements $p_{U^i_{j_i}(t) \cond \mathcal{U}^i} \allowbreak \big(u^i_{j_i}, t \cond \mathcal{U}^i \big)$ for $i=1,\ldots, L$. Notice that for the sake of readability, we slightly abused the notation in Equation~\eqref{equ:batch-dgp-density}, in the sense that $\mathcal{U}^i$, appearing in the argument of $p_{U^i_{j_i}(t) \cond \mathcal{U}^i} \big(u^i_{j_i}, t \cond \mathcal{U}^i \big)$, actually stands for the realisation of all the GPs contained in $\mathcal{U}^i$. 

In order for the DGP $V$ represented by Equation~\eqref{equ:batch-dgp} to be well-defined, its covariance functions $C^1,\ldots, C^L$ must be chosen suitably. Many conventional covariance functions -- such as the Mat\'{e}rn $C_{\mathrm{Mat.}}$ in Equation~\eqref{equ:cov-matern} -- mostly fail to be positive definite if one replaces their parameters with time dependent functions. To allow for time-varying parameters, a typical choice is to use
\begin{equation}
	\begin{split}
		&C_{\mathrm{NS}}(t, t'; \ell, \sigma) \\
		&= \frac{\sigma(t) \, \sigma(t') \big(\ell(t) \, \ell(t')\big)^{\frac{1}{4}} \, \sqrt{2}}{\varGamma(\nu) \, 2^{\nu-1} \, \sqrt{\ell(t) + \ell(t')}} \left(\sqrt{\frac{8 \, \nu \, (t-t')^2}{\ell(t) + \ell(t')}}\right)^{\!\!\nu} \mBesselsec\!\!\left(\sqrt{\frac{8 \, \nu \, (t-t')^2}{\ell(t) + \ell(t')}}\right),
		\label{equ:cov-ns-matern}
	\end{split}
\end{equation}
which is a non-stationary generalisation of the Mat\'{e}rn family by~\citet{Paciorek2004, Paciorek2006}. \citet{Gibbs} introduces a similar formulation for constructing a non-stationary RBF covariance function. More non-stationary covariance function examples using time-varying parameters can also be found in, for example, \citet{Higdon1999non, Snoek2014, Remes2017}.

\subsection*{State-space representations of DGPs}
Another way to represent a DGP as defined in Definition~\ref{def:dgp} is through the use of SDEs. The idea consists in forming a (non-linear) system of SDE representations of all the GP elements appearing in the hierarchy. More precisely, let $U^1_0, U^2_{j_2}, \ldots, U^i_{j_i}, \ldots, U^L_{j_L}$ be $\R^{d_i}$-valued GPs that satisfy the following SDEs
\begin{equation}
	\begin{split}
		\diff U^1_0(t) &= A^1\big(t; \mathcal{U}^1\big) \, U^1_0 \diff t + B^1\big(t; \mathcal{U}^1\big) \diff W^1(t), \\
		\diff U^2_{j_2}(t) &= A^2\big(t; \mathcal{U}^2\big) \, U^2_{j_2} \diff t + B^2\big(t; \mathcal{U}^2\big) \diff W^2(t), \\
		&\vdots \\
		\diff U^i_{j_i}(t) &= A^i\big(t; \mathcal{U}^i\big) \, U^i_{j_i} \diff t + B^i\big(t; \mathcal{U}^i\big) \diff W^i(t), \\
		&\vdots \\
		\diff U^L_{j_L}(t) &= A^L(t) \, U^L_{j_L} \diff t + B^L(t) \diff W^L(t),
		\label{equ:ss-dgps-sde-split}
	\end{split}
\end{equation}
respectively.
In Equation~\eqref{equ:ss-dgps-sde-split}, $W^i \colon \T \to \R^{w_i}$ for $i=1,2,\ldots, L$ are $w_i$-dimensional Wiener processes, and $A^i \colon \T \to \R^{d_i \times d_i}$ and $B^i \colon \T \to \R^{d_i \times w_i}$ for $i=1,2,\ldots, L-1$ are stochastic processes that are parametrised by the GPs in $\mathcal{U}^i$. The $L$-th coefficients $A^L\colon\T\to\R^{d_L}$ and $B^L\colon\T\to\R^{d_L \times w_L}$, on the other hand, are deterministic, since $\mathcal{U}^L = \emptyset$ by definition. For the sake of simplicity, we collapse Equation~\eqref{equ:ss-dgps-sde-split} into a matricial form
\begin{equation}
	\begin{split}
		\diff V(t) &= a(V(t)) \diff t + b(V(t)) \diff W(t), \\
		V(t_0) &= V_0,
		\label{equ:ss-dgps-sde}
	\end{split}
\end{equation}
where $V(t) \coloneqq \begin{bmatrix} U^1_0(t) & U^2_{j_2}(t) & \cdots & U^L_{j_L}(t) \end{bmatrix}^\trans \in \R^{\sum_{i=1}^L d_i}$, and the SDE coefficients are defined by
\begin{equation}
	a(V(t)) \coloneqq 
	\begin{bmatrix}
		A^1\big(t; \mathcal{U}^1\big) & & & \\
		& A^2\big(t; \mathcal{U}^2\big) & & & \\
		& & \ddots & \\
		& & & A^L\big(t; \mathcal{U}^L\big)
	\end{bmatrix} \, V(t)
\end{equation}
and
\begin{equation}
	b(V(t)) \coloneqq
	\begin{bmatrix}
		B^1\big(t; \mathcal{U}^1\big) & & & \\
		& B^2\big(t; \mathcal{U}^2\big) & & & \\
		& & \ddots & \\
		& & & B^L\big(t; \mathcal{U}^L\big)
	\end{bmatrix}.
\end{equation}
The vector-valued Wiener process appearing in Equation~\eqref{equ:ss-dgps-sde} is similarly defined by $W(t) \coloneqq \begin{bmatrix}
W^1(t) & W^2(t) & \cdots & W^L(t)
\end{bmatrix}^\trans \in \R^{\sum_{i=1}^L w_i}$. 

A DGP $V\colon \T\to \R^{\sum^L_{i=1}d_i}$ that is characterised as per SDE~\eqref{equ:ss-dgps-sde} is called a state-space deep Gaussian process (SS-DGP). Compared to batch representations of DGPs, one specifies the SDE coefficients $A^i$ and $B^i$ for $i=1,2,\ldots, L$ and the initial condition $V_0$ instead of explicitly specifying the covariance functions of DGPs. In Section~\ref{sec:deep-matern} we present some concrete examples of how to select these SDE coefficients so that each GP element of the SS-DGPs is conditionally a \matern GP.

\section{Existence and uniqueness of SS-DGPs}
\label{sec:ssdgps-solution}
In the previous sections, we have defined SS-DGPs as SDE represented DGPs. However, the solution existence and uniqueness of the SDE in Equation~\eqref{equ:ss-dgps-sde-split} has still not been proven. In this section, we provide sufficient conditions on the SDE coefficients in SDE~\eqref{equ:ss-dgps-sde-split} so that the strong existence and pathwise uniqueness hold for the SDE. 

In particular, one must understand that the hierarchical nature of SS-DGPs makes a direct application of Theorem~\ref{thm:linear-sde-solution} slightly unsound. Indeed, the system of SDEs~\eqref{equ:ss-dgps-sde-split} is not a linear system when seen as a multidimensional SDE. However, the individual GP elements SDEs are (conditionally on their parents in the DGP hierarchy) linear. 

\begin{theorem}
	Let $W^i\colon \T \to \R^{w_i}$ and $U^i(t_0)$ for $i=1,2,\ldots, L$ be Wiener processes and initial random variables defined on filtered probability spaces $\big(\Omega^i, \FF^i, \FF^i_t, \PP^i \big)$ for $i=1,2,\ldots, L$, where their filtrations $\lbrace \FF^i_t\colon 1,2,\ldots,L\rbrace$ are generated by their Wiener processes and initial variables. Suppose that functions $A^i$ and $B^i$ for $i=1,2,\ldots, L$ in Equation~\eqref{equ:ss-dgps-sde-split} are locally bounded measurable, then the multidimensional SDE~\eqref{equ:ss-dgps-sde-split}, or equivalently, \eqref{equ:ss-dgps-sde} has a strong solution and the pathwise uniqueness holds.
\end{theorem}
\begin{proof}
	By Theorem~\ref{thm:linear-sde-solution} and the conditions of this theorem, the SDEs in Equation~\eqref{equ:ss-dgps-sde-split} are exactly the same with the integral equations
	\begin{align}
		U^1_0(t) &= \Lambda^1\big(t;\mathcal{U}^1\big) \, U^1_0(t_0) + \Lambda^1\big(t;\mathcal{U}^1\big) \int^t_{t_0} \big(\Lambda^1\big(s;\mathcal{U}^1\big)\big)^{-1} \, B^1\big(s;\mathcal{U}^1\big) \diff W^1(s), \nonumber\\
		U^2_{j_2}(t) &= \Lambda^2\big(t;\mathcal{U}^2\big) \, U^2_{j_2}(t_0) + \Lambda^2\big(t;\mathcal{U}^2\big) \int^t_{t_0} \big(\Lambda^2\big(s;\mathcal{U}^2\big)\big)^{-1} \, B^2\big(s;\mathcal{U}^2\big) \diff W^2(s), \nonumber\\
		&\vdots \label{equ:ss-dgp-sde-integral}\\
		U^i_{j_i}(t) &= \Lambda^i\big(t;\mathcal{U}^i\big) \, U^i_{j_i}(t_0) + \Lambda^i\big(t;\mathcal{U}^i\big) \int^t_{t_0} \big(\Lambda^i\big(s;\mathcal{U}^i\big)\big)^{-1} \, B^i\big(s;\mathcal{U}^i\big) \diff W^i(s), \nonumber\\
		&\vdots \nonumber\\
		U^L_{j_L}(t) &= \Lambda^L(t) \, U^L_{j_L}(t_0) + \Lambda^L(t)\, \int^t_{t_0} \big(\Lambda^L(s)\big)^{-1} \, B^L(s) \diff W^L(s), \nonumber
	\end{align}
	where $\Lambda^i$ for $i=1,2,\ldots, L$ are defined as per Theorem~\ref{thm:linear-sde-solution}. Hence, the joint process $V(t) \coloneqq \begin{bmatrix} U^1_0(t) & U^2_{j_2}(t) & \cdots & U^L_{j_L}(t) \end{bmatrix}^\trans$ is an $\FF_t$-adapted process defined on the product space $(\Omega, \FF, \FF_t, \PP)$, where $\Omega = \Omega^1 \times \cdots \times \Omega^L$, $\FF$ and $\FF_t$ are the product sigma-algebras and filtrations~\citep{ReneMeasure2017}, and $\PP(E^1 \times \cdots \times E^L) = \PP^1(E^1) \, \cdots \, \PP^L(E^L)$ for every $E^1\in\Omega^1, \ldots, E^L\in\Omega^L$. Noting that the other properties in Definition~\ref{def:strong-solution} are also verified, Equation~\eqref{equ:ss-dgp-sde-integral} is a strong solution of the multidimensional SDE~\eqref{equ:ss-dgps-sde-split}. The pathwise uniqueness of SDE~\eqref{equ:ss-dgps-sde-split} follows from the fact that the pathwise uniqueness holds for the linear SDEs of all the GP elements~\citep[see,][Lemma 7]{Zhao2021RSSGP}.
\end{proof}

The theorem above shows that in order to give a well-defined SS-DGP we only needs to ensure the SDE coefficients be locally bounded measurable functions. This condition is substantially weaker compared to the classical ones, such as the global Lipschitz and linear growth conditions~\citep{Karatzas1991, Friedman1975, Mao2008, Shen2006}, because we have leveraged the hierarchical nature of SS-DGP. From now on, unless otherwise specified, we will assume that this condition holds whenever we construct an SS-DGP.

Thanks to the Markov property, probability densities of SS-DGPs can factorise in the time dimension. Suppose that we have temporal instances $t_1 \leq t_2 \leq \cdots \leq t_T \in\T$, then the probability density function of $V$ on these time instances reads
\begin{equation}
	p_{V_{1:T}}(v_{1:T}) = p_{V_1}(v_1)\,\prod^T_{k=1} p_{V_{k+1} \cond  V_k}(v_{k+1} \cond v_{k}), \nonumber
\end{equation}
where we denote $V_{1:T} \coloneqq \lbrace V(t_1), V(t_2), \ldots, V(t_T)\rbrace$. We can also factorise the probability density above in the GP element variable like in Equation~\eqref{equ:batch-dgp-density} as well. 

\subsection*{Covariance functions of SS-DGPs}
The equivalence between batch and state-space DGP representations can be stated in terms of equivalence of covariance functions. In particular, conditionally on its parents in the DGP hierarchy, we can express the covariance function of a GP element as a function of its SDE coefficients.  
\begin{theorem}
	\label{thm:ss-dgp-cov}
	Let $V(t)$ be an SS-DGP governed by the SDE in Equation~\eqref{equ:ss-dgps-sde} on some probability space $( \Omega, \FF, \PP)$. Also let $\FF^i \subset \FF$ be the sub-sigma-algebra generated by the GPs in $\mathcal{U}^i(t)$ for all $t\in\T$. Then the covariance function of the $i$-th GP element is
	\begin{equation}
		\begin{split}
			C^i_{\mathrm{SS}}(t, t'; \mathcal{U}^i) &\coloneqq \covbig{U^i_{j_i}(t), U^i_{j_i}(t') \cond \mathcal{F}^i}\\
			&=\cu{\Lambda}^i(t,t_0) \, \covbig{U^i_{j_i}(t_0) \cond \mathcal{U}^i(t_0)} \big(\cu{\Lambda}^i(t, t_0) \big)^\trans \\
			&\quad+\int^{t \,\wedge \, t'}_{t_0} \cu{\Lambda}^i(t, s) \, B^i\big(s; \mathcal{U}^i(s)\big) \, B^i\big(s; \mathcal{U}^i(s)\big)^\trans \, \big(\cu{\Lambda}^i(t, s) \big)^\trans \diff s,
			\label{equ:ss-dgp-cov}
		\end{split}
	\end{equation}
	where $\cu{\Lambda}^i(t,s) = \Lambda^i(t) \, \big(\Lambda^i(s)\big)^{-1}$ for $t,s\in\T$, and $\Lambda^i(t)$ is generated by $A^i$ as per Theorem~\ref{thm:linear-sde-solution}.
\end{theorem}
\begin{proof}
	By It\^{o}'s formula and Theorem~\ref{thm:linear-sde-solution}, we have that
	\begin{equation}
		\begin{split}
			U^i_{j_i}(t) &= \Lambda^i(t)\,U^i_{j_i}(t_0) + \Lambda^i(t)\int^t_{t_0} \big(\Lambda^i(s)\big)^{-1} \, B^i\big(s; \mathcal{U}^i(s)\big) \, \diff W^i(s)\\
			&\coloneqq \cu{\Lambda}^i(t,t_0) \, U^i_{j_i}(t_0) + \int^t_{t_0} \cu{\Lambda}^i(t,s) \, B^i\big(s; \mathcal{U}^i(s)\big) \, \diff W^i(s)
			\label{equ:ss-dgp-cov-U}
		\end{split}
	\end{equation}
	with respect to $\FF^i$. Note that $\Lambda(t_0) = I$ as per Equation~\eqref{equ:solution-linear-sdes-initial-condition}.
	Hence, by It\^{o} isometry and by substituting $U^i_{j_i}(t)$ into
	\begin{equation}
		\begin{split}
			&\covbig{U^i_{j_i}(t), U^i_{j_i}(t') \cond \mathcal{F}^i} \\
			&= \expecbig{U^i_{j_i}(t)\, \big(U^i_{j_i}(t')\big)^\trans \cond \FF^i} - \expecbig{U^i_{j_i}(t) \cond \FF^i} \, \big(\expecbig{U^i_{j_i}(t') \cond \FF^i}\big)^\trans,
		\end{split}
	\end{equation}
	we arrive at Equation~\eqref{equ:ss-dgp-cov}. For details, see,~\citet{Zhao2021RSSGP}.
\end{proof}

\begin{remark}
	The matrix $\cu{\Lambda}^i(t, s)$ above is often referred to as the transition matrix in control theory~\citep{Brogan2011}. Although in general $\cu{\Lambda}^i(t, s)$ does not have a closed-form representation, Peano--Baker series in Theorem~\ref{thm:peano-baker} can be used to approximate it successively~\citep{Baake2011, DaCunha2005}. One can also use Magnus expansions, if an exponential representation of transition matrix (i.e., $\cu{\Lambda}^i(t,s)=\exp(\cdot)$) is required, but the convergence usually requires strict conditions on $A^i$~\citep{Moan2008MagnusConv}. 
	
	However, if $A^i$ is self-commuting for all $t\in\T$, then the transition matrix simplifies to $\cu{\Lambda}^i(t, s) = \exp\big(\int^t_{t_0} A^i\big( s; \mathcal{U}^i(s) \big) \diff s\big)$. 
\end{remark}

The converse of Theorem~\ref{thm:ss-dgp-cov} is also available to some extent in the sense that the covariance functions in batch DGPs can be translated into the SDE coefficients of state-space DGPs. For how to proceed on this, we refer the reader to~\citet{Hartikainen2010, Simo2013SSGP}.

\section{Numerical simulation of SS-DGPs}
\label{sec:ssdgps-lcd}
In this section we discuss the numerical simulation of the SDEs describing SS-DGPs. In particular we present approximate discretisation methods that leverage the hierarchical nature of SS-DGPs, then we discuss alternatives that would result in exact simulations. 

\subsection*{Discretisation of SDEs}

In order to simulate SS-DGPs, it is very common to consider discretisations of their SDEs. In particular, we focus on the Gaussian increment-based explicit discretisations of the form
\begin{equation}
	V_k \approx f_{k-1}(V_{k-1}) + q_{k-1}(V_{k-1}),
	\label{equ:ss-dgp-disc}
\end{equation}
where $V_k \coloneqq V(t_k)$ and $q_{k-1} \sim \mathrm{N}(0, Q_{k-1}(V_{k-1}))$, and the functions $f_{k-1}$ and $Q_{k-1}$ depend on the discretisation scheme used.

Unfortunately, many commonly used discretisation methods fail to provide valid numerical schemes for SS-DGPs. For instance, the Euler--Maruyama method yields singular covariance $Q_{k-1}$ for smooth Mat\'{e}rn SS-DGPs (see, e.g., Example~\ref{example:ssdgp-m32}). While higher-order It\^{o}--Taylor expansions, such as Milstein's method, exist, they are only numerically efficient for constant, diagonal, or more generally, commutative dispersion function $b$~\citep[see, the definition of commutative noise in][Chapter 10]{Kloeden1992}. However, dispersion functions of SS-DGPs may not always verify these conditions (e.g., Example~\ref{example:ssdgp-m12}).

The Taylor moment expansion (TME) method presented in Section~\ref{sec:tme} does not suffer from the problems of high-order It\^{o}--Taylor expansions, but on the other hand it requires sufficient smoothness on the SDE coefficients. Moreover, the resulting covariance estimate $Q_{k-1}$ used in the TME-based discretisation in Equation~\eqref{equ:ss-dgp-disc} may be singular. While the smoothness of the coefficients is a necessary price to pay, the possible singularity of the estimated covariance $Q_{k-1}$ can be addressed. We refer the reader back to Section~\ref{sec:tme-pd} for methods to do so.

In this thesis, we additionally present an ad-hoc discretisation approach by leveraging the hierarchical structure of SS-DGPs and explicit solutions of linear SDEs (e.g., Equation~\eqref{equ:ss-dgp-sde-integral}). The idea relies on approximating the SDE of each GP element between two time steps $t_{k-1}$ and $t_k$ by a time-invariant SDE, the coefficients of which depend on the values of its parent GPs at $t_{k-1}$. This idea roots in the so-called local linearisation methods as in~\citet{Ozaki1993, Sarkka2019}. By using this approach, the transition matrix $\cu{\Lambda}$, as defined in Theorem~\ref{thm:ss-dgp-cov}, reduces to a matrix exponential. The following algorithm shows how this hierarchical discretisation can be used in practice.
\begin{algorithm}[Locally conditional discretisation]
	\label{alg:local-cond-disc}
	Starting from any $t_{k-1}\in\T$, locally conditional discretisation (LCD) approximates the solution of SDEs~\eqref{equ:ss-dgps-sde} at time $t_k\in\T$ by
	\begin{equation}
		\begin{split}
			U^i_{j_i}(t_k) &\approx \widetilde{\cu{\Lambda}}^i(t_k, t_{k-1}) \, U^i_{j_i}(t_{k-1}) + \int^{t_k}_{t_{k-1}} \widetilde{\cu{\Lambda}}^i(t_k, s) \, B^i\big(t_{k-1}; \mathcal{U}^i(t_{k-1})\big) \diff W^i(s),
			\label{equ:lcd-solution}
		\end{split}
	\end{equation}
	for $i=1,2,\ldots, L$, where $\widetilde{\cu{\Lambda}}^i(t_k, s) \coloneqq \exp\big( (t_k - s) \, A^i (t_{k-1}; \mathcal{U}^i(t_{k-1}))\big)$, and $A^i$ and $B^i$ are defined in Equation~\eqref{equ:ss-dgps-sde-split}.
\end{algorithm}
\begin{remark}
	Except in special cases (such as \matern SS-DGPs presented later), Equation~\eqref{equ:lcd-solution} needs to be solved numerically. This can be done, for example, by using the methods highlighted around Equations~\eqref{equ:disc-linear} and~\eqref{equ:disc-coeff-exp}. 
\end{remark}

It is worth mentioning that the non-stationary Gaussian state-space model introduced by~\citet{YaoweiLi2020} coincides with the LCD approximation to the \matern class of SS-DGPs (see, Section~\ref{sec:deep-matern}). 

\begin{figure}[t!]
	\centering
	\includegraphics[width=.95\linewidth]{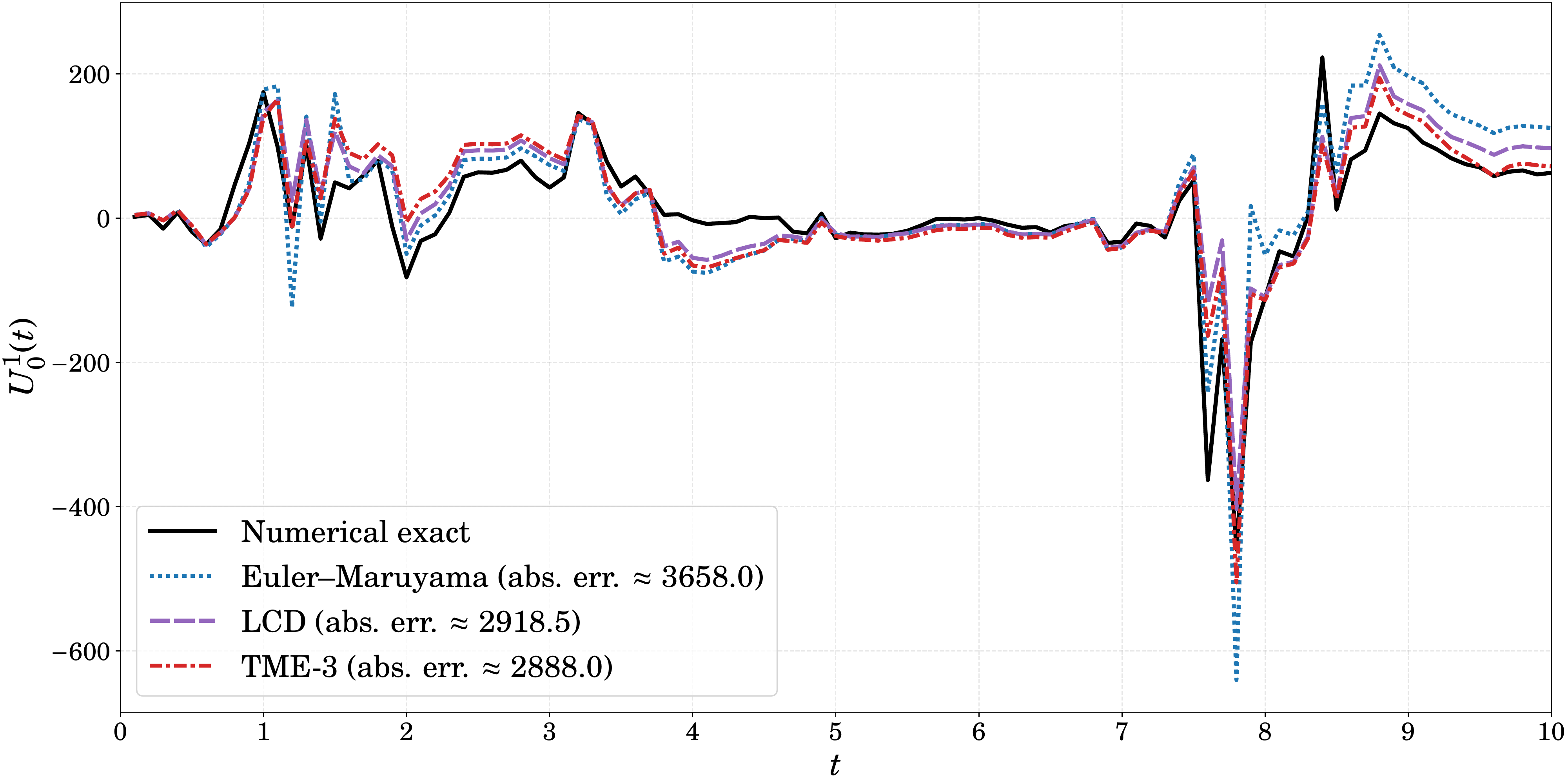}
	\caption{Comparison of different discretisation schemes on the \matern $\nu=1\,/\,2$ SS-DGP defined in Example~\ref{example:ssdgp-m12}, where we let parameters $\ell_2=\ell_3=1$ and $\sigma_2=\sigma_3=0.1$. The numbers in the legend are the cumulative absolute errors with respect to the (numerically) exact discretisation.}
	\label{fig:disc-err-dgp-m12}
\end{figure}

Figure~\ref{fig:disc-err-dgp-m12} illustrates a comparison amongst the Euler--Maruyama, TME, and LCD methods on a \matern SS-DGP. On this example, both LCD and TME methods outperform Euler--Maruyama substantially, especially in the ``high-frequency'' portions of this SDE trajectory.

\subsection*{Exact simulation methods}

Apart from discretisation-based simulations, there also exist exact simulation methods~\citep{Beskos2005, Kessler2012, Blanchet2020}. Although these methods can avoid the discretisation errors, they are usually limited to specific types of SDEs, which may not apply to all SS-DGPs. As an example, the method introduced by~\citet{Beskos2005} requires that the dispersion coefficient be constant, which is usually not the case in SS-DGPs.

Finally, it is worth noting that each sub-SDE in Equation~\eqref{equ:ss-dgps-sde} is a linear SDE conditionally on its parent GPs. Hence, we could borrow the idea of Gibbs sampling~\citep{Robert2004} in order to sample from $U^i_{j_i}$ for $i=L,L-1,\ldots, 1$. While this method was not implemented in the context of this thesis, it is likely to improve on the LCD method and will therefore be a subject of future work.

\section{Deep \matern processes}
\label{sec:deep-matern}
In this section, we present SS-DGPs that are constructed in the Mat\'{e}rn sense. Specifically, we choose the SDE coefficients in Equation~\eqref{equ:ss-dgps-sde-split} in such a way that each GP element is a \matern GP when conditioned on its parent GPs. 

Let us start by considering linear SDEs of the form
\begin{equation}\label{equ:matern-linear-sde}
	\diff U(t) = A \, U(t) \diff t+ B \diff W(t),
\end{equation}
where the initial condition $U(t_0) \sim \mathrm{N}(0, P_0)$ is a Gaussian random variable, and the Wiener process $W\colon \T\to\R$ takes value in $\R$. Let $\nu \in \big\lbrace \frac{1}{2}, \frac{3}{2}, \ldots \big\rbrace$ and $\gamma = \nu + \frac{1}{2}$. Suppose that the state $U\colon \T \to \R^\gamma$ verifies 
\begin{equation}
	U(t) = 
	\begin{bmatrix}
	\overline{U}(t) & \frac{\diff\overline{U}}{\diff t}(t) & \cdots & \frac{\diff^{\gamma-1}\overline{U}}{\diff t^{\gamma-1}}(t) 
	\end{bmatrix}^\trans,
	\label{equ:gp-matern-state}
\end{equation}
and that the coefficients in Equation~\eqref{equ:matern-linear-sde} are given by
\begin{equation}
	A = 
	\begin{bmatrix}
		0 & 1 &   &\\
		  & 0 & 1 &\\
		\vdots & & \ddots &\\
		-\binom{\gamma}{0} \kappa^\gamma & -\binom{\gamma}{1} \kappa^{\gamma-1} & \cdots & -\binom{\gamma}{\gamma-1} \kappa
	\end{bmatrix}, \quad 
	B = 
	\begin{bmatrix}
		0 \\
		0\\
		\vdots\\
		\frac{\sigma \varGamma(\gamma) \, (2\,\kappa)^{\gamma - \frac{1}{2}}}{\sqrt{\varGamma(2\,\gamma-1)}}
	\end{bmatrix},
	\label{equ:matern-sde-coeffs}
\end{equation}
where $\kappa = \sqrt{2 \, \nu} \, / \, \ell$. Furthermore, suppose that the initial covariance $P_0$ solves the corresponding Lyapunov equation (see, Equation~\eqref{equ:lyapunuv}) of the SDE. Then the process $\overline{U} \colon \T \to \R$ in Equation~\eqref{equ:gp-matern-state} is a zero-mean Mat\'{e}rn GP with the covariance function $C_{\mathrm{Mat.}}$ defined in Equation~\eqref{equ:cov-matern}~\citep{Simo2013SSGP, Solin2016}.

\begin{remark}
	The matrix $A$ in Equation~\eqref{equ:matern-sde-coeffs} is Hurwitz~\citep{Khalil2002} as all its eigenvalues have strictly negative real part. However, $A$ is prone to be ill-conditioned if $\nu$ is large, resulting in numerically unstable SDEs. This can be addressed by using balancing algorithms~\citep{Osborne1960, Parlett1971}. 
\end{remark}

Based on the aforementioned \matern SDE representation, we can now construct \matern SS-DGPs by choosing their coefficients $A^i$ and $B^i$ for $i=1,\ldots,L$, as per Equation~\eqref{equ:matern-sde-coeffs}. As an example, suppose that the $i$-th GP element $U^i_{j_i}\in\R^{\gamma}$ in Equation~\eqref{equ:ss-dgps-sde-split} has two parents $U^m_{i}\colon \T\to\R^{d_m}$ and $U^n_{i}\colon \T\to\R^{d_n}$, that encode the length scale and the magnitude parameters, respectively. Then, we can select two suitable transformation functions $g_m \colon \R^{d_m} \to \R_{>0}$ and $g_n \colon \R^{d_n} \to \R_{>0}$, and let 
\begin{equation}
	\ell_i(t) = g_m\big( U^m_{i}(t) \big)
	\label{equ:ss-dgp-ell}
\end{equation}
and
\begin{equation}
	\sigma_i(t) = g_n\big(U^n_{i}(t)\big).
	\label{equ:ss-dgp-sig}
\end{equation} 
Under these notations, the coefficient $A^i$ of $U^i_{j_i}$ reads
\begin{equation}
		A^i(t; \mathcal{U}^i) = A^i\big( U^m_{i}(t)\big) = \begin{bmatrix}
		0 & 1 &   &\\
		& 0 & 1 &\\
		\vdots & & \ddots &\\
		-\binom{\gamma}{0} \kappa^\gamma_i(t) & -\binom{\gamma}{1} \kappa^{\gamma-1}_i(t) & \cdots & -\binom{\gamma}{\gamma-1} \kappa_i(t)
		\end{bmatrix},
\end{equation}
where $\kappa_i(t) = \sqrt{2 \, \nu} \, / \, \ell_i(t)$. Likewise, one can derive the coefficient $B^i(t; \mathcal{U}^i) = B^i\big( U^m_{i}(t), U^n_{i}(t) \big) = \begin{bmatrix}
	0 & 0 & \cdots & \sigma_i(t) \, \varGamma(\gamma) \, (2\,\kappa_i(t))^{\gamma - \frac{1}{2}}\, (\varGamma(2\,\gamma-1))^{\frac{1}{2}}
\end{bmatrix}^\trans \in \R^{\gamma}$.

Transformation functions should also be chosen regular enough so that the solution of the related SDE is well-defined (see, Section~\ref{sec:ssdgps-solution}). For example, in~\citet{Zhao2020SSDGP, Zhao2021RSSGP}, we use $g(u) = \exp(u)$, $g(u) = \arctan(u) + \pi \, / \, 2$, or $g(u) = \log(1 + \exp(u))$.

SDEs of Mat\'{e}rn SS-DGPs are time-homogeneous by construction (i.e., the coefficients $A^i$ and $B^i$ for $i=1,2,\ldots,L$ do not explicitly depend on time). Provided that the transformation functions are chosen suitably as per~\citet[][Definition~7.1.1]{Oksendal2003}, the \matern SS-DGPs are then It\^{o} diffusions. This can bring many useful features, such as the strong Markov property~\citep{Ikeda1992}. 

In the following, we give some concrete examples of Mat\'{e}rn SS-DGPs and plot a few of their simulations. 

\begin{example}[Mat\'{e}rn $\nu=1\,/\,2$ SS-DGP with three GP elements]
	\label{example:ssdgp-m12}
	Let $\nu=1\, / \, 2$. Consider the following SDEs
	\begin{equation}
		\begin{split}
			\diff U^1_0(t) &= -\frac{1}{\ell_1(t)}\,U^1_0(t) \diff t + \frac{\sqrt{2} \, \sigma_1(t)}{\sqrt{\ell_1(t)}} \diff W^1(t),\\
			\diff U^2_1(t) &= -\frac{1}{\ell_2}\,U^2_1(t) \diff t + \frac{\sqrt{2} \, \sigma_2}{\sqrt{\ell_2}} \diff W^2(t),\\
			\diff U^3_1(t) &= -\frac{1}{\ell_3}\,U^3_1(t) \diff t + \frac{\sqrt{2} \, \sigma_3}{\sqrt{\ell_3}} \diff W^3(t),
		\end{split}
		\label{equ:ss-dgp-m12}
	\end{equation}
	where $\ell_1(t) = g\big( U^2_1(t) \big)$ and $\sigma_1 = g\big( U^3_1(t) \big)$ are the length scale and magnitude of $U^1_0(t)$, respectively. The solution $V(t) = \begin{bmatrix}
		U^1_0(t) & U^2_1(t) & U^3_1(t)
	\end{bmatrix}^\trans$ is said to be a Mat\'{e}rn $\nu=1\,/\,2$ SS-DGP. 
\end{example}

\begin{example}[Mat\'{e}rn $\nu=3\,/\,2$ SS-DGP with three GP elements]
	\label{example:ssdgp-m32}
	Let $\nu= 3 \, / \, 2$. Consider the following SDEs
	\begin{equation}
		\begin{split}
			\diff U^1_0(t) &= 
			\begin{bmatrix}
				0 & 1\\
				-\big(\frac{\sqrt{3}}{\ell_1(t)}\big)^2 & \frac{-2 \,\sqrt{3}}{\ell_1(t)}
			\end{bmatrix} \, U^1_0(t) \diff t + 
			\begin{bmatrix}
				0 \\
				2\,\sigma_1(t) \, \big(\frac{\sqrt{3}}{\ell_1(t)}\big)^{\frac{3}{2}}
			\end{bmatrix} \diff W^1(t), \\
			\diff U^2_1(t) &= 
			\begin{bmatrix}
				0 & 1\\
				-\big(\frac{\sqrt{3}}{\ell_2}\big)^2 & \frac{-2 \,\sqrt{3}}{\ell_2}
			\end{bmatrix} \, U^2_1(t) \diff t + 
			\begin{bmatrix}
				0 \\
				2\,\sigma_2 \, \big(\frac{\sqrt{3}}{\ell_2}\big)^{\frac{3}{2}}
			\end{bmatrix} \diff W^2(t), \\
			\diff U^3_1(t) &= 
			\begin{bmatrix}
				0 & 1\\
				-\big(\frac{\sqrt{3}}{\ell_3}\big)^2 & \frac{-2 \,\sqrt{3}}{\ell_3}
			\end{bmatrix} \, U^3_1(t) \diff t + 
			\begin{bmatrix}
				0 \\
				2\,\sigma_3 \, \big(\frac{\sqrt{3}}{\ell_3}\big)^{\frac{3}{2}}
			\end{bmatrix} \diff W^3(t), 
		\end{split}
	\end{equation}
	where $U^1_0(t) = \begin{bmatrix}
		\overline{U}^1_0(t) & \frac{\diff\overline{U}^1_0}{\diff t}(t)
	\end{bmatrix}^\trans$, and similarly for $U^2_1(t)$ and $U^3_1(t)$. The length scale and magnitude of $U^1_0(t)$ are given by $\ell_1(t) = g\big( U^2_1(t) \big)$ and $\sigma_1 = g\big( U^3_1(t) \big)$, respectively, for $g\colon \R^2 \to \R_{>0}$. The solution $V(t) = \begin{bmatrix}
	U^1_0(t) & U^2_1(t) & U^3_1(t)
	\end{bmatrix}^\trans$ is said to be a Mat\'{e}rn $\nu=3\,/\,2$ SS-DGP. 
	
	It is worth noting that for this model the Euler--Maruyama scheme gives a singular discretisation covariance.
\end{example}

\begin{figure}[t!]
	\centering
	\includegraphics[width=.95\linewidth]{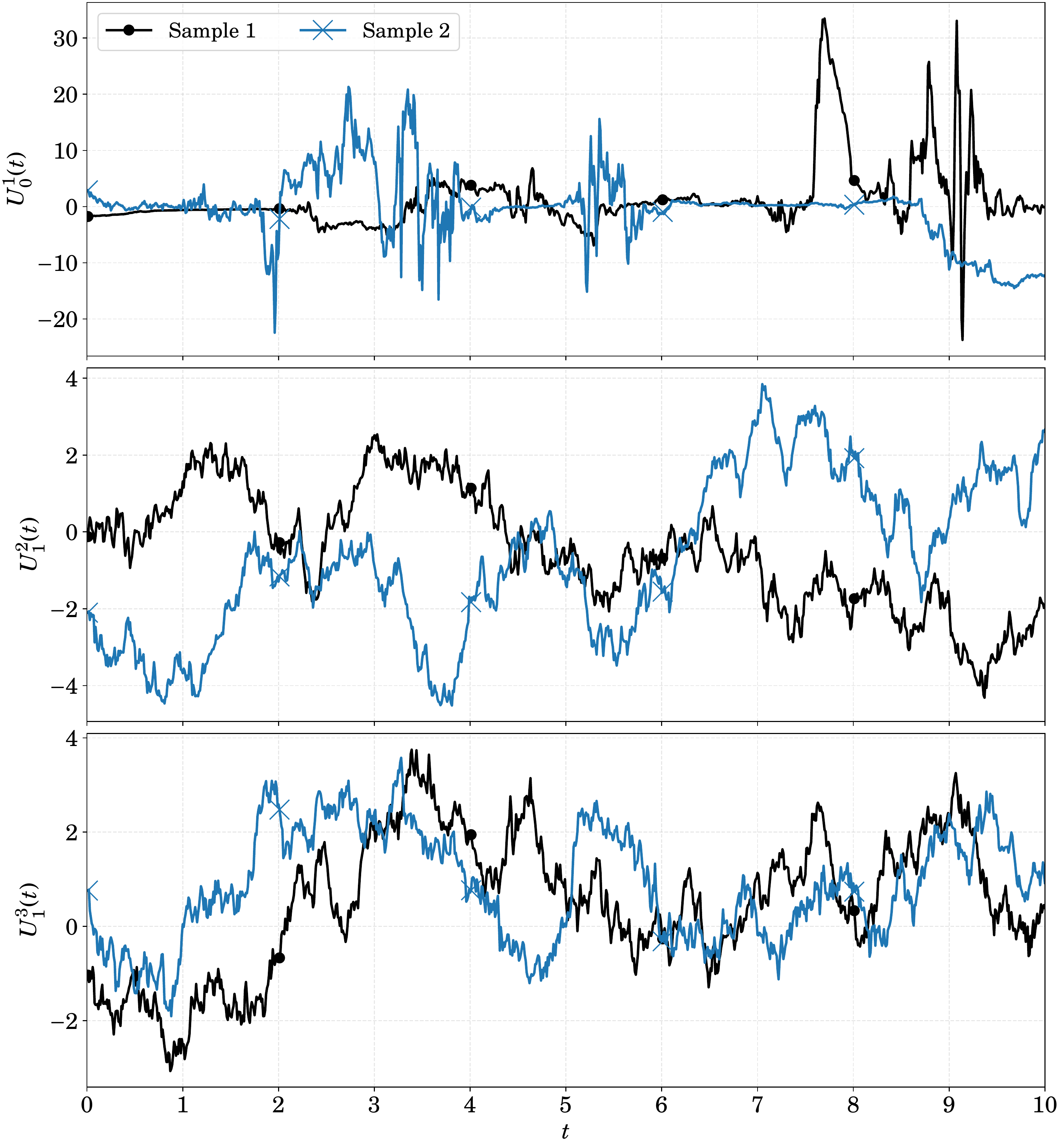}
	\caption{This figure shows two samples (plotted in different colours and markers) drawn from the Mat\'{e}rn $\nu=1\,/\,2$ SS-DGP defined in Example~\ref{example:ssdgp-m12}.}
	\label{fig:ssdgp-m12-samples}
\end{figure}
\begin{figure}[t!]
	\centering
	\includegraphics[width=.95\linewidth]{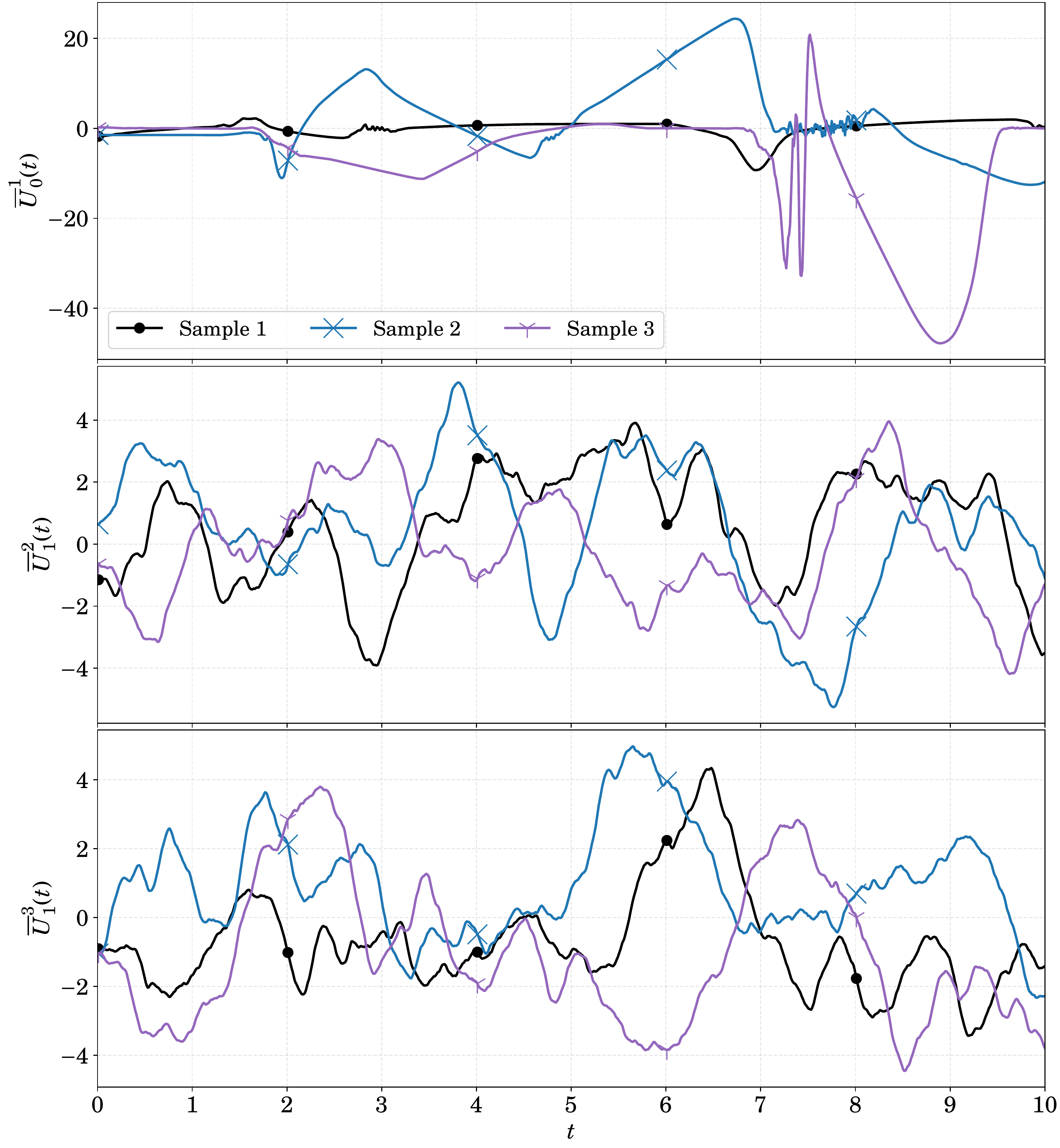}
	\caption{This figure shows three samples (plotted in different colours and markers) drawn from the Mat\'{e}rn $\nu=3\,/\,2$ SS-DGP defined in Example~\ref{example:ssdgp-m32}.}
	\label{fig:ssdgp-m32-samples}
\end{figure}

The Examples~\ref{example:ssdgp-m12} and~\ref{example:ssdgp-m32} feature Mat\'{e}rn SS-DGPs with only three GP elements. This hierarchy/depth can be continued further to represent higher degrees of non-stationarity. 

Figures~\ref{fig:ssdgp-m12-samples} and~\ref{fig:ssdgp-m32-samples} illustrate a few samples drawn from the Mat\'{e}rn SS-DGPs defined in Examples~\ref{example:ssdgp-m12} and~\ref{example:ssdgp-m32}, respectively. More specifically, for the Mat\'{e}rn $\nu=1\,/\,2$ SS-DGP in Example~\ref{example:ssdgp-m12} we use $\ell_2=\ell_3=\sigma_2=\sigma_3=2$ and $g(u) = \exp(u)$, while for the Mat\'{e}rn $\nu=3\,/\,2$ SS-DGP in Example~\ref{example:ssdgp-m32} we use $\ell_2=\ell_3=0.5$, $\sigma_2=\sigma_3=2$ and $g(u) = \exp\big( \big[ 1 \,\,0 \big] \, u \big)$. The initial states are standard Gaussian random vectors with unit covariances in both cases.

From Figures~\ref{fig:ssdgp-m12-samples} and~\ref{fig:ssdgp-m32-samples}, we can observe non-stationary in the behaviour of $U^1_0$. This results from its length scale and magnitude being driven by its parent GPs $U^2_1$ and $U^3_1$. As an example, Sample~2 (blue line) of $\overline{U}^1_0$ in Figure~\ref{fig:ssdgp-m32-samples} exhibits low-magnitude high-frequency jittering around $t\in[7, 8]$ because the length scale $g\big(\overline{U}^2_1\big)$ and magnitude $g\big(\overline{U}^3_1\big)$ are relatively small on $t\in[7, 8]$. On the other hand Sample~3 (magenta line) of $\overline{U}^1_0$ in Figure~\ref{fig:ssdgp-m32-samples} exhibits high-magnitude medium-frequency jittering around $t\in[7, 8]$ because the length scale $g\big(\overline{U}^2_1\big)$ and magnitude $g\big(\overline{U}^3_1\big)$ are relatively average and high, respectively, on $t\in[7, 8]$.

The main usefulness of this Mat\'{e}rn construction is that the resulting Mat\'{e}rn SS-DGPs can provide generic priors for modelling a wide class of continuous functions (with smoothness parameter $\gamma - 1$). These priors are flexible in the sense that they have non-stationary characteristics which can be learnt from data. In Chapter~\ref{chap:apps}, we will show some real applications of Mat\'{e}rn SS-DGPs. 

Apart from the Mat\'{e}rn construction, it is also possible to build SS-DGPs by formulating SDE coefficients in some other meaningful ways. For example, \citet{Solin2014} construct SDEs that represent quasi-periodic oscillators, and~\citet{Syama2018} parametrise SDEs with neural networks for time series forecasting.

\begin{figure}[t!]
	\centering
	\includegraphics[width=.9\linewidth]{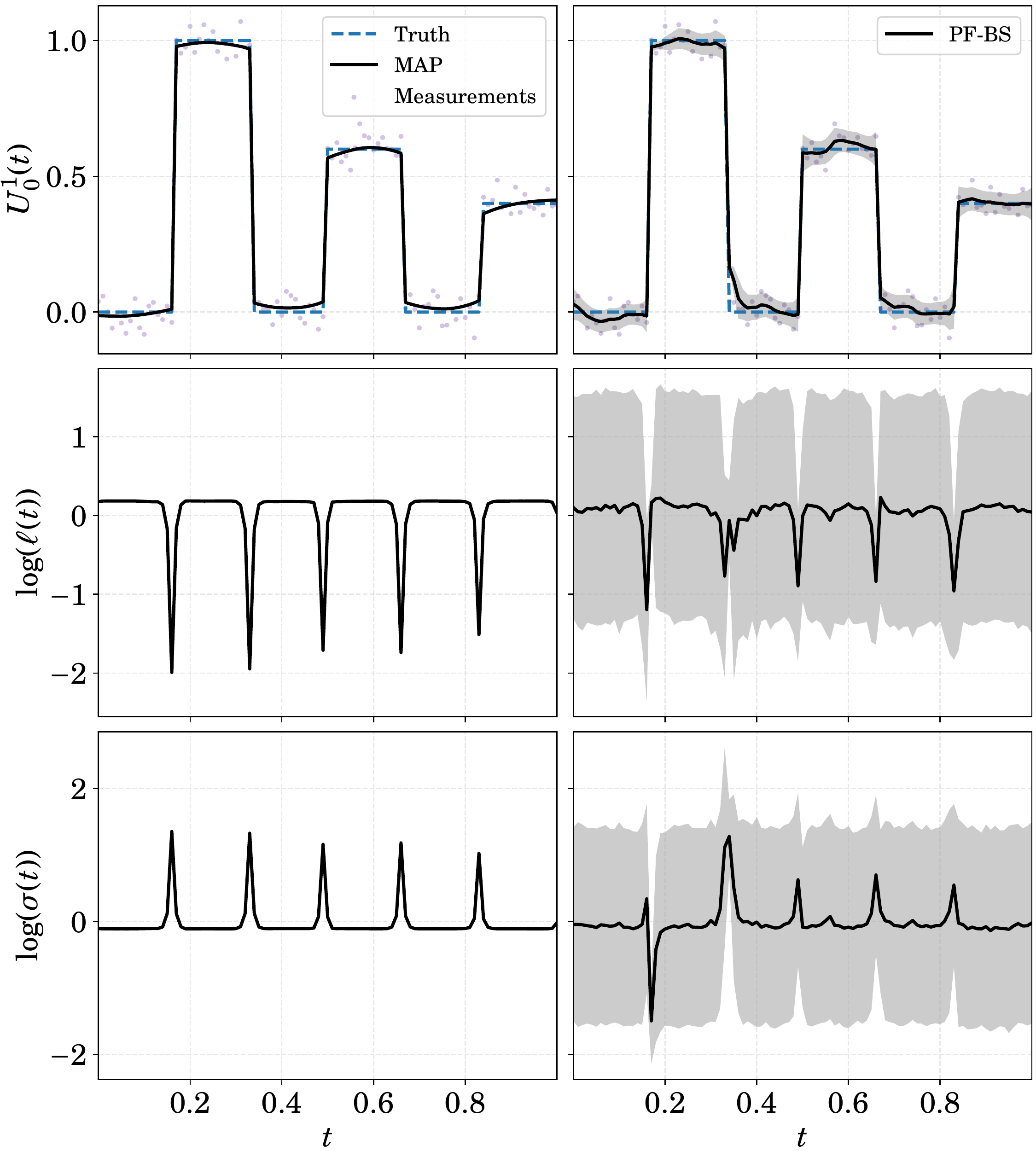}
	\caption{Mat\'{e}rn SS-DGP regression on a rectangular signal. The first column corresponds to the MAP estimate of the SDE state, while the second column corresponds to a full posterior estimate using PF-BS (particle filter and backward simulation smoother). }
	\label{fig:ssdgp-reg-rect}
\end{figure}
\begin{figure}[t!]
	\centering
	\includegraphics[width=.99\linewidth]{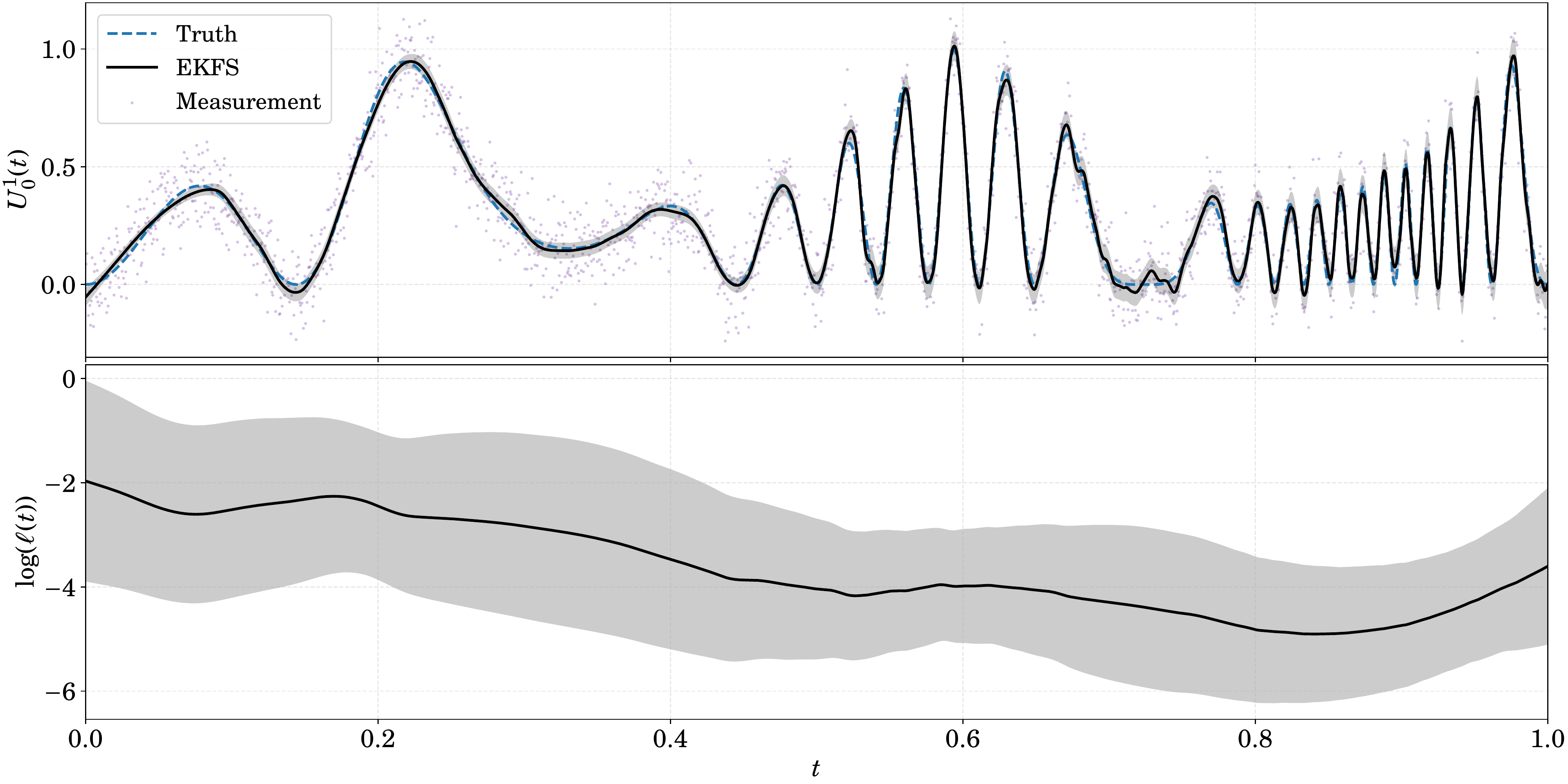}
	\caption{Mat\'{e}rn SS-DGP regression on a composite sinusoidal signal by using EKFS (extended Kalman filter and smoother). }
	\label{fig:ssdgp-reg-sine}
\end{figure}

\section{SS-DGP Regression}
\label{sec:ssdgp-reg}
In this section, we show how to solve SS-DGP regression problems for discrete measurement data. Since SS-DGPs are characterised by SDEs, we view these problems as continuous-discrete Bayesian smoothing problems (see, Section~\ref{sec:cd-smoothing}).

Let $V \colon \T \to \R^{\sum^L_{i=1} d_i}$ be an SS-DGP defined as per Equation~\eqref{equ:ss-dgps-sde}. Suppose that we measure $V$ at $t_1, t_2,\ldots, t_T \in \T$, by a (non-linear) function $h\colon \R^{\sum^L_{i=1} d_i} \to \R^{d_y}$ and additive Gaussian noises $\xi_k\sim\mathrm{N}(0, \Xi_k)$ for $k=1,2,\ldots, T$. We consider the SS-DGP regression problem in its continuous-discrete state-space form
\begin{equation}
	\begin{split}
		\diff V(t) &= a(V(t)) \diff t + b(V(t)) \diff W(t), \quad V(t_0) = V_0,\\
		Y_k &= h(V_k) + \xi_k,
		\label{equ:ssdgp-regression}
	\end{split}
\end{equation}
where we abbreviate $V_k \coloneqq V(t_k)$. Suppose we have a set of measurements $y_{1:T} = \lbrace y_1, y_2, \ldots, y_T\rbrace $, we want to learn the (smoothing) posterior density $p_{V_k \cond Y_{1:T}}(v_k \cond y_{1:T})$ for $k=1,2,\ldots, T$, or more generally $p_{V(t) \cond Y_{1:T}}(v,t \cond y_{1:T})$, for any $t\in\T$. We can then use the methods presented in Section~\ref{sec:cd-smoothing} to solve the regression/continuous-discrete smoothing problem above. 

Thanks to the Markov property of the SS-DGP prior, we can solve this regression problem in linear computational time with respect to $T$ by leveraging Bayesian filtering and smoothing methods. This is in contrast with batch DGPs, where one often needs to solve matrix inversions of dimension $T\times T$.

In Figures~\ref{fig:ssdgp-reg-rect} and~\ref{fig:ssdgp-reg-sine}, we plot some SS-DGP regression examples taken from~\citet{Zhao2020SSDGP}. Compared to the GP regression shown in Figure~\ref{fig:gp-fail}, we can see the advantages of SS-DGPs for fitting irregular data. Also, the estimated length scale and magnitude parameters can explain well the changes of regime in the data generating process. 

It would be also possible to extend SS-DGP regression to classification by modifying the measurement model in Equation~\eqref{equ:ssdgp-regression} accordingly~\citep{Neal1998, Carl2006GPML, Garcia2019}. For example, we can assume that the measurement follows a categorical distribution, with parameters determined by the SS-DGP states~\citep{Carl2006GPML}.

\section{Identifiability analysis of Gaussian approximated SS-DGP regression}
\label{sec:identi-problem}
Gaussian filters and smoothers (GFSs, see, Section~\ref{sec:gaussian-filter-smoother}) are widely used classes of Bayesian filters and smoothers. Moreover, \citet{Zhao2020SSDGP} show that GFSs are particularly efficient for solving SS-DGP regression problems. However, for a certain class of SS-DGPs, GFSs cannot identify (i.e., estimate the posterior density of) their state components as $t\to\infty$. More specifically, in this section, we show how -- under some weak assumptions on the SS-DGP regression model coefficients -- the posterior (cross-)covariance estimates of the regression problem solutions at the measurements times $t_k$ collapse to $0$ as $k\to\infty$.

\begin{figure}[t!]
	\centering
	\includegraphics[width=.7\linewidth]{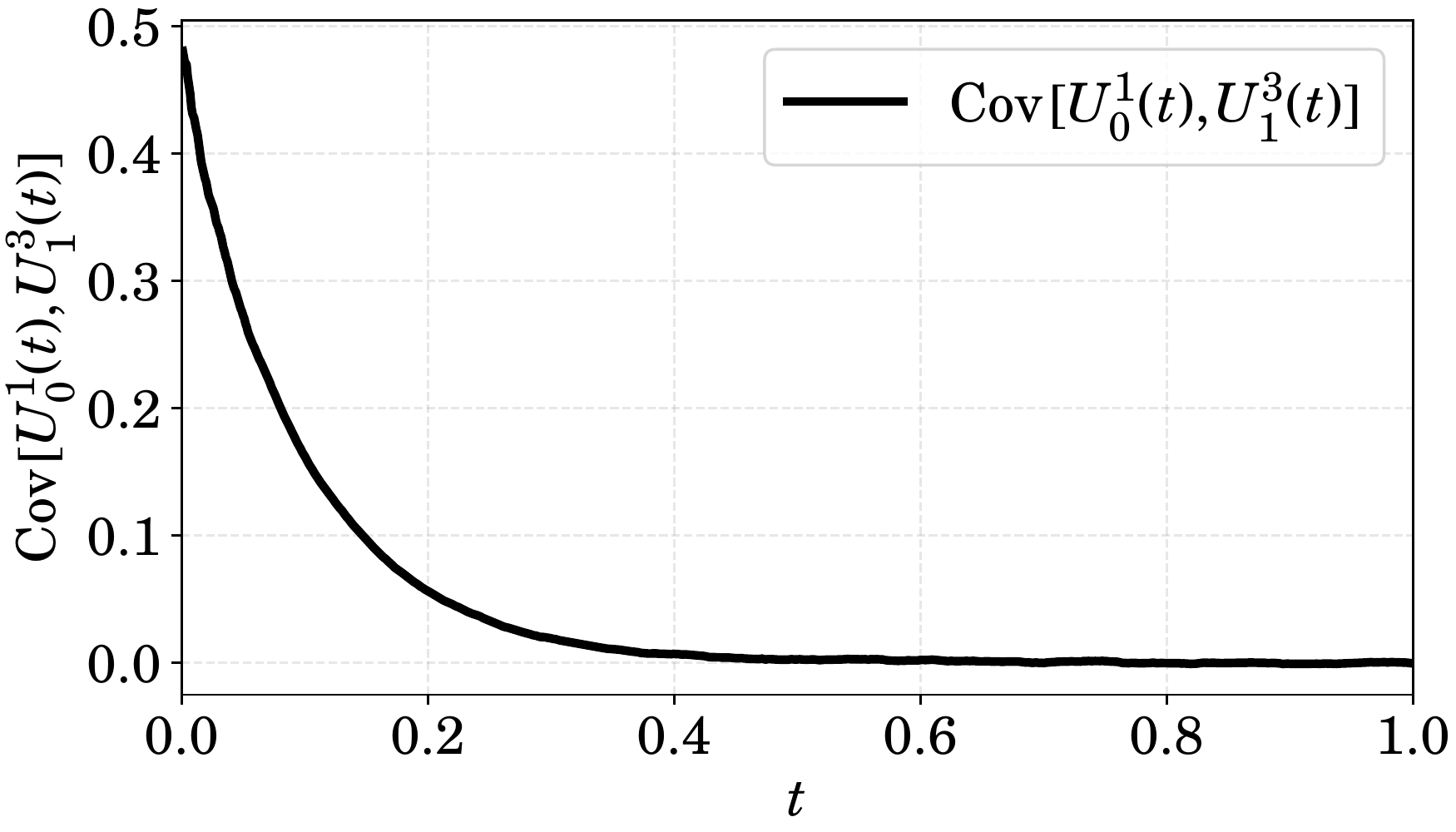}
	\caption{Evolution of $\covbig{U^1_0(t), U^3_1(t)}$ in Example~\ref{example:ssdgp-m12}, estimated by 20,000 independent Monte Carlo runs. Parameters are $\ell_2=\ell_3=\sigma_2=\sigma_3=0.1$, and the initial $\covbig{U^1_0(0), U^3_1(0)}=0.5$. We observe that the covariance (numerically) converges to zero monotonically as $t$ grows.}
	\label{fig:ssdgp-vanishing-cov}
\end{figure}

To explain the problem in short, let us suppose that we have an SS-DGP regression model with the SDE given by Example~\ref{example:ssdgp-m12}, and that we measure the first GP element $U^1_0$ with additive Gaussian noises. Further suppose that we apply GFSs (see, Algorithm~\ref{alg:gfs}) to solve the regression problem. It turns out that this SDE has a vanishing $\covbig{U^1_0(t), U^3_1(t)} \to 0$ as $t\to\infty$ (see, Figure~\ref{fig:ssdgp-vanishing-cov} for a numerical illustration), and that the GFS estimated posterior $\covbig{U^1_0(t_k), U^3_1(t_k) \cond y_{1:k}}$ vanishes to zero as $k\to \infty$ too. Consequently, the Kalman gain for the component $U^3_1$ converges to zero as $k\to\infty$. This means that the posterior distribution of $U^3_1$ estimated by GFSs will use no information from measurements as $k\to\infty$.

In order to formulate the problem, we limit ourselves to a class of SS-DGP regression models for which the dispersion term of the observed GP element is parametrised by another GP element. Formally, we consider $U^1_0\colon \T \to \R$ and $U^2_1 \colon \T \to \R$ that are the solutions of the pair of SDEs
\begin{equation}
	\begin{split}
		\diff U^1_0(t) &= A^1( \psi(t) ) \, U^1_0(t) \diff t + B^1\big(\psi(t), U^2_1(t)\big) \diff W^1(t),\\
		\diff U^2_1(t) &= A^2(\varphi(t)) \, U^2_1(t) \diff t + B^2(\varphi(t)) \diff W^2(t),
		\label{equ:ident-sde}
	\end{split}
\end{equation}
on a filtered probability space $(\Omega, \FF, \FF_t, \PP )$, where the initial conditions, the processes $\psi \colon \T \to \R$ and $\varphi \colon \T \to \R$, and the Wiener processes $W^1 \colon \T \to \R$ and $W^2 \colon \T \to \R$ are mutually independent.
\begin{remark}
	Note that the index $2$ on $U^2_1$ is arbitrary as the definition of a DGP is invariant of reindexing of its components (see, Definition~\ref{def:dgp}).
\end{remark}

\begin{remark}
	The SDEs given by Equation~\eqref{equ:ident-sde} represent a class of SS-DGPs for which an inner GP element $U^2_1$ parametrises the dispersion term of the measured GP element $U^1_0$. Since the parents of $U^1_0$ and $U^2_1$ are not necessarily Gaussian (but are instead conditional Gaussian), we generically name their parents $\psi$ and $\varphi$ which can be any well-defined processes. As an example, in the left figure of Figure~\ref{fig:dgp-examples-graph}, one can imagine $\psi$ as the representation of $U^3_1$, $U^6_3$, and $U^7_3$, while $\varphi$ as the representation of $U^4_2$ and $U^5_2$.
\end{remark}

Let the random variables
\begin{equation}
	Y_k = U^1_0(t_k) + \xi_k, \quad \xi_k \sim \mathrm{N}(0, \Xi_k),
	\label{equ:ident-y}
\end{equation}
for $k=1,2,\ldots$ stand for the measurements at time $t_1, t_2, \ldots$, and assume that $\inf_k \lbrace t_k - t_{k-1} \colon k=1,2,\ldots\rbrace > 0$.

We use the following assumptions. 

\begin{assumption}
	\label{assump:ssdgp-ident-solution}
	The coefficients $A^1 \colon \R \to \R_{<0}$, $B^1 \colon \R \times \R \to \R$, $A^2 \colon \R \to \R_{<0}$, $B^2\colon \R \to \R$, and the initial conditions are chosen regular enough so that the weak uniqueness (see, Definition~\ref{def:weakly-unique}) holds for the SDE in Equation~\eqref{equ:ident-sde} .
\end{assumption}

\begin{assumption}
	\label{assump:ssdgp-ident-init}
	The components $U^1_0$ and $U^2_1$ at the initial time $t_0$ satisfy
	$\expecbig{\absbig{U^1_0(t_0)}} < \infty$, $\expecbig{ \absbig{U^2_1(t_0)}} < \infty$, and $\expecbig{ \absbig{U^1_0(t_0) \, U^i_1(t_0)}} < \infty$.
\end{assumption}

\begin{assumption}
	\label{assump:ssdgp-ident-neg}
	There exists constants $c_{A^1}<0$ and $c_{A^2}<0$ such that $\PP$-almost surely the processes $(A^1 \circ \psi)(t) \leq c_{A^1}$ and $(A^2 \circ \varphi)(t) \leq c_{A^2}$ for all $t\in\T$. 
\end{assumption}

\begin{assumption}
	\label{assump:ssdgp-ident-bound-A-B}
	There exist constants $c_A > 0$ and $c_B>0$ such that, for all $t \in \T$, $\expecbig{\big( A^1(\psi(t))\,U^1_0(t) \big)^2} \leq c_A$ and $\expecbig{\big( B^1\big(\psi(t), U^2_1(t)\big) \big)^2} \geq c_B$.
\end{assumption}

\begin{assumption}
	\label{assump:ssdgp-ident-bound-R}
	There exists $\Xi_{\mathrm{inf}} > 0$ such that for every $k=1,2,\ldots$, either $\Xi_k > \Xi_{\mathrm{inf}}$ or $\Xi_k = 0$.
\end{assumption}

Assumption~\ref{assump:ssdgp-ident-solution} ensures SDEs~\eqref{equ:ident-sde} be well-defined. Assumption~\ref{assump:ssdgp-ident-init} postulates absolute integrability of the initial conditions which is used in the proof of Lemma~\ref{lemma:vanishing-prior-cov}. Assumption~\ref{assump:ssdgp-ident-bound-A-B} aims to yield a positive lower bound for $\varrbig{U^1_0(t)}$ as used in Corollary~\ref{corollary:var-f-bound}. 

Assumption~\ref{assump:ssdgp-ident-neg} is the key assumption to have the prior covariance vanishing. This assumption is pragmatic because it ensures that the mean of $U^1_0$ and $U^2_1$ shrinks to zero over time. Also, if one considers $-1 \,/\, A^1$ and $-1 \, / \, A^2$ as length scales, then this assumption guarantees their positivity.

\begin{lemma}
	\label{lemma:vanishing-prior-cov}
	Under Assumptions~\ref{assump:ssdgp-ident-solution} to~\ref{assump:ssdgp-ident-neg}, 
	\begin{equation}
		\lim_{t\to\infty} \covbig{U^1_0(t), U^2_1(t)} = 0.
		\label{equ:vanish-cov-0}
	\end{equation}
\end{lemma}
\begin{proof}
	By It\^{o}'s formula and the law of total expectation, one can find that 
	\begin{align}
		\covbig{U^1_0(t), U^2_1(t)} &= \expecBig{\expecbig{U^1_0(t_0) \, U^2_1(t_0) \cond \psi(t_0), \varphi(t_0)} e^{\int^t_{t_0} A^1( \psi(s)) + A^2( \varphi(s)) \diff s}} \nonumber\\
		&\quad-\expecBig{\expecbig{U^1_0(t_0) \cond \psi(t_0)} e^{\int^t_{t_0} A^1( \psi(s)) \diff s }} \label{equ:vanish-cov-eq1}\\
		&\qquad\times\expecBig{\expecbig{U^2_1(t_0) \cond \varphi(t_0)} e^{\int^t_{t_0} A^2( \varphi(s)) \diff s }}.\nonumber
	\end{align}
	Next, knowing that $A^1 \circ \psi$ and $A^2 \circ \varphi$ are upper bounded by Assumption~\ref{assump:ssdgp-ident-neg}, we can then apply the triangle inequality and conditional Jensen's inequality~\citep[][Theorem 8.20]{AchimKlenke2014} to bound the three terms in the equation above. As an example, the first term admits
	\begin{equation}
		\begin{split} &\absBig{\expecBig{\expecbig{U^1_0(t_0) \, U^2_1(t_0) \cond \psi(t_0), \varphi(t_0)} \, e^{\int^t_{t_0} A^1( \psi(s)) + A^2(\varphi(s)) \diff s }}}\\
		&\leq \expecBig{\absbig{\expecbig{U^1_0(t_0) \, U^2_1(t_0) \cond \psi(t_0), \varphi(t_0)}} \, \absbig{e^{\int^t_{t_0} A^1( \psi(s)) + A^2(\varphi(s)) \diff s }}} \\
		&\leq \expecBig{\absbig{U^1_0(t_0) \, U^2_1(t_0) }} \, e^{(c_{A^1} + c_{A^2}) \, (t - t_0)}.
		\label{equ:ssdgp-identi-lemma-cov-ineq}
		\end{split}
	\end{equation}
	%
	For the rest two expectation terms in Equation~\eqref{equ:vanish-cov-eq1}, \textit{mutatis mutandis}. Finally, by taking limits on both side of Equation~\eqref{equ:vanish-cov-eq1} and using the bounds above, one arrives at Equation~\eqref{equ:vanish-cov-0}.
\end{proof}
\begin{remark}
	The proof above slightly deviates from the proof given in~\citet[][Lemma 1]{Zhao2020SSDGP} as the original proof uses Assumption~\ref{assump:ssdgp-ident-neg} in a different order, which yields an unnecessarily stricter bound than that of Equation~\eqref{equ:ssdgp-identi-lemma-cov-ineq}. 
\end{remark}
%
%
%
\begin{lemma}
	\label{lemma:var-f-bound}
	Under Assumption~\ref{assump:ssdgp-ident-init}, for every $\epsilon > 0$, there exists $\zeta_\epsilon > 0$ such that
	\begin{equation}
		\begin{split}
			&\varrbig{U^1_0(t)} \\
			&\geq \frac{1}{z(t)} \, \int^t_{t_0} z(s)\,\Big( \expecbig{\big( B^1\big(\psi(s), U^2_1(s)\big) \big)^2} -2\,\epsilon \sqrt{\expecbig{\big( A^1(\psi(s))\,U^1_0(s) \big)^2}}\Big) \diff s,
			\label{equ:var-bound-1}
		\end{split}
	\end{equation}
	where 
	\begin{equation}
		z(t) = \exp \Bigg\lbrace\int^t_{t_0} 2 \, \zeta_\epsilon \,\sqrt{\expecbig{\big( A^1(\psi(s))\,U^1_0(s) \big)^2}} \diff s\Bigg\rbrace. \nonumber
	\end{equation}
\end{lemma}
\begin{proof}
	We give the idea of the proof, for details, see,~\citet{Zhao2020SSDGP}. The first step is to express $\varrbig{U^1_0(t)}$ as the solution of an integral/differential equation. Then by using H\"{o}lder's inequality one can obtain an integral/differential inequality. Finally, by using the integrating factor method on $\big(z(t)\varrbig{U^1_0(t)}\big)$, one can recover the desired bound.
\end{proof}

\begin{remark}
	\label{remark:langenhop-var-f}
	We can also use Theorem~\ref{thm:langenhop} to obtain alternative positive lower bounds by letting $v(x) = \sqrt{x}$ or $v(x) = \epsilon + \zeta_\epsilon \, x$ in Theorem~\ref{thm:langenhop}. The resulting bounds do not involve the dispersion term $\expecbig{\big( B^1\big(\psi(s), U^2_1(s)\big) \big)^2}$ but includes the initial variance $\varrbig{U^1_0(t_0)}$ instead.
\end{remark}

\begin{corollary}
	\label{corollary:var-f-bound}
	Under Assumptions~\ref{assump:ssdgp-ident-solution} and~\ref{assump:ssdgp-ident-bound-A-B}, there exists an $\epsilon >0$ such that
	\begin{equation}
		\varrbig{U^1_0(t)} \geq \frac{(c_B - 2 \, \epsilon \, \sqrt{c_A}) \, (t - t_0)}{\exp\big(2 \, \zeta_\epsilon \, \sqrt{c_A} \, (t-t_0)\big)} > 0,
	\end{equation}
	for all $t\in\T$.
\end{corollary}
\begin{proof}
	The bound follows from Lemma~\ref{lemma:var-f-bound} and Assumption~\ref{assump:ssdgp-ident-bound-A-B}. In order to make the bound positive, one then needs to choose $\epsilon < c_B \, / \, (2 \, \sqrt{c_A})$.
\end{proof}

We can now analyse the limit of the posterior covariance 
\begin{equation}
	\covbig{U^1_0(t_k), U^2_1(t_k) \cond y_{1:k}} \approx P^{1, 2}_k 
\end{equation}
as approximated by Gaussian filters as $k\to \infty$. In order to do so, in Algorithm~\ref{alg:abs-gf}, we consider an abstract general form of Gaussian filters that suppose perfect integration in the prediction step. We use the notations $\covbig{U^1_0(t_k), U^2_1(t_k)}_{z_s}$ and $\varrbig{U^1_0(t_k)}_{z_s}$ to represent the values of $\covbig{U^1_0(t_k), U^2_1(t_k)}$ and $\varrbig{U^1_0(t_k)}$, respectively, at time $t_k$ starting from any initial value $z_s$ at time $t_s<t_k \in \T$.

\begin{algorithm}[Abstract Gaussian filter for $P^{1, 2}_k$]
	\label{alg:abs-gf}
	Suppose that we have initial conditions $P^{1, 2}_0 = \covbig{U^1_0(t_0), U^2_1(t_0)}$ and $P^{1,1}_0=\varrbig{U^1_0(t_0)}$. Starting from $k=1$ the abstract Gaussian filter predicts
	\begin{equation}
		\begin{split}
			\overline{P}^{1, 2}_k &= \covbig{U^1_0(t_k), U^2_1(t_k)}_{P^{1, 2}_{k-1}},\\
			\overline{P}^{1, 1}_k &= \varrbig{U^1_0(t_k)}_{P^{1, 1}_{k-1}},
			\label{equ:abs-gfs-pred}
		\end{split}
	\end{equation}
	and updates
	\begin{equation}
		\begin{split}
			P^{1, 2}_k &= \overline{P}^{1, 2}_k - \frac{\overline{P}^{1, 1}_k \, \overline{P}^{1, 2}_k}{\overline{P}^{1, 1}_k + \Xi_k},\\
			K^{1, 2}_k &= \frac{\overline{P}^{1, 2}_k}{\overline{P}^{1, 1}_k + \Xi_k},
		\end{split}
	\end{equation}
	for $k=1,2,\ldots$
\end{algorithm}
\begin{remark}
	Algorithm~\ref{alg:abs-gf} is a skeleton of Algorithm~\ref{alg:gfs} that is only concerned with the covariance estimates $P^{1, 2}_k$ for $k=1,2,\ldots$. However, this algorithm assumes that the predictions through the SDE are done exactly as per Equation~\eqref{equ:abs-gfs-pred}, which is usually unrealistic in practice. \citet{Zhao2020SSDGP} explain how this abstraction is derived.
\end{remark}

We can finally state the main result of this section.
\begin{theorem}
	\label{thm:vanish-post-cov}
	Suppose that Assumptions~\ref{assump:ssdgp-ident-solution} to~\ref{assump:ssdgp-ident-bound-R} hold. Further assume that $\absbig{\covbig{U^1_0(t_k), U^2_1(t_k)}_{z_{k-1}}} \leq \abs{z_{k-1}}$ for all initial $z_{k-1} \in\R$ and $k=1,2,\ldots$, then Algorithm~\ref{alg:abs-gf} gives
	\begin{equation}
		\lim_{k\to\infty} P^{1, 2}_k = 0.
	\end{equation}
\end{theorem}
\begin{proof}
	The basic idea is to expand the recursion in Algorithm~\ref{alg:abs-gf} for $k=1,2,\ldots$, and by mathematical induction one can prove that
	\begin{equation}
		\absbig{P^{1, 2}_k} \leq \absbig{P^{1, 2}_0} \, \prod^k_{j=1} D_j, \nonumber
	\end{equation}
	where
	\begin{equation}
		D_j = \frac{R_j}{\overline{P}^{1, 1}_j + R_j}.\nonumber
	\end{equation}
	Hence, the limit of $P^{1, 2}_k$ depends on the limit of $\prod^k_{j=1} D_j$. Although $D_j$ is always less than $1$, the infinite product $\prod^\infty_{j=1} D_j$ does not necessarily converge to zero (e.g., Vi\`{e}te's formula). However, Lemma~\ref{lemma:var-f-bound} and Assumption~\ref{assump:ssdgp-ident-bound-R} ensure that $\overline{P}^{1, 1}_j$ is lower bounded uniformly by some positive $P_{\inf}$, so that $D_j < \frac{R_j}{R_j + P_{\inf}}$, for all $j > 0$. Assumption~\ref{assump:ssdgp-ident-bound-R} then allows to conclude. For details, see~\citet{Zhao2020SSDGP}.
\end{proof}

The consequence of Theorem~\ref{thm:vanish-post-cov} is that the Kalman gain $K^{1, 2}_k$ in Algorithm~\ref{alg:abs-gf} will also converge to zero as $k\to\infty$. It means that the Kalman update for the state $U^2_1(t_k)$ will not use information from measurements in the limit $k\to\infty$.

\section{$L^1$-regularised batch and state-space DGP regression}
\label{sec:l1-r-dgp}
Constrained/regularised regression, for example, the sparsity-inducing least absolute shrinkage and selection operator~\citep[LASSO,][]{Tibshirani1996} method is an important topic in statistics, machine learning, and inverse problems~\citep{Kaipio2005, Hastie2015}. Heuristically, sparsity in DGPs may also yield several benefits, in particular for modelling discontinuous signals for which the length scale around the discontinuities should jump from a high value to almost zero (see, e.g., Figure~\ref{fig:ssdgp-reg-rect}). In this section, we show how $L^1$-regularisation can be interpreted and implemented in the context of batch and state-space DGP regressions.

\subsection*{Regularised batch DGP regression}

For the sake of exposition and to keep notations simple, we will restrict ourselves to a shallow DGP, with only one observed GP element depending on two latent GPs:
\begin{equation}
	\begin{split}
		U^1_0(t) \condbig \mathcal{U}^1 &\sim \mathrm{GP}\big(0, C^1(t, t'; \mathcal{U}^1)\big), \\
		U^2_1(t) &\sim \mathrm{GP}\big(0, C^2(t, t')\big),\\
		U^3_1(t) &\sim \mathrm{GP}\big(0, C^3(t, t')\big),\\
		Y_k &= U^1_0(t_k) + \xi_k, \quad \xi_k \sim \mathrm{N}(0, \Xi_k),
		\label{equ:r-dgp-reg-model}
	\end{split}
\end{equation}
where $\mathcal{U}^1 = \big\lbrace U^2_1, U^3_1 \big\rbrace$, and we let the DGP $V(t) = \begin{bmatrix} U^1_0(t) & U^2_1(t) & U^3_1(t)
\end{bmatrix}^\trans$. Suppose that at times $\lbrace t_k \in\T \colon k=1,2,\ldots, T \rbrace$ we have measurements $y_{1:T} = \lbrace y_k\colon k=1,2,\ldots, T \rbrace$, the DGP regression aims to learn the posterior density
\begin{equation}
	\begin{split}
		&p_{V_{1:T} \cond Y_{1:T}}(v_{1:T} \cond y_{1:T}) \\
		&\propto p_{Y_{1:T} \cond V_{1:T}}(y_{1:T} \cond v_{1:T}) \, p_{V_{1:T}}(v_{1:T}) \\
		&= p_{Y_{1:T} \cond U^1_{0, 1:T}}\big(y_{1:T} \cond u^1_{0,1:T}\big) \, p_{U^1_{0, 1:T} \cond U^2_{1, 1:T}, U^3_{1, 1:T}}\big( u^1_{0, 1:T} \cond u^2_{1, 1:T}, u^3_{1, 1:T} \big)\\
		&\quad\times p_{U^2_{1, 1:T}}\big( u^2_{1, 1:T} \big) \, p_{U^3_{1, 1:T}}\big( u^3_{1, 1:T} \big),
		\label{equ:r-dgp-post}
	\end{split}
\end{equation}
where we define $V_{1:T} \coloneqq \lbrace V(t_k)\colon k=1,2,\ldots, T\rbrace$, $U^1_{0, 1:T} \coloneqq \big\lbrace U^1_{0}(t_k)\colon k=1,2,\allowbreak\ldots, T \big\rbrace$, and similarly for $U^2_{1,1:T}$ and $U^3_{1,1:T}$. Now let us introduce three regularisation-inducing matrices $\Phi^1\in \R^{T \times T}$, $\Phi^2\in \R^{T\times T}$, and $\Phi^3\in \R^{T\times T}$. We are interested in learning the posterior density~\eqref{equ:r-dgp-post} under an $L^1$-regularisation of the GP elements, that is by introducing the penalty terms
\begin{equation}
	\normbig{ \Phi^1 \, u^1_{0, 1:T}}_1, \quad \normbig{ \Phi^2 \, u^2_{1, 1:T}}_1, \quad \text{and} \quad \normbig{ \Phi^3 \, u^3_{1, 1:T}}_1.
\end{equation}
In other words, we encourage the $\Phi$-transformed variables to be sparse in the $L^1$ norm sense. For example, if we let $\Phi$ to be the identity matrix (respectively, a finite difference matrix), then the resulting penalty will correspond to increasing elementwise sparsity (respectively, reducing the total variation of the function).

We consider a maximum a posterior (MAP) approach for solving Equation~\eqref{equ:r-dgp-post}, and we express the regularised DGP regression problem as a penalised optimisation problem. Namely, by taking the negative log of Equation~\eqref{equ:r-dgp-post}, we get an objective function
\begin{align}
	\mathcal{L}^{\mathrm{B}} &\coloneqq \mathcal{L}^{\mathrm{B}}(v_{1:T}) \coloneqq \mathcal{L}^{\mathrm{B}}(u^1_{0, 1:T}, u^2_{1, 1:T}, u^3_{1, 1:T}) \nonumber\\
	&= \normbig{u^1_{0, 1:T} - y_{1:T}}^2_{\Xi_{1:T}} + \normbig{u^1_{0, 1:T}}^2_{C^1_{1:T}} + \log \,\det\big(2 \, \pi \, C^1_{1:T}\big) \label{equ:r-dgp-batch-obj}\\
	&\quad+ \normbig{u^2_{1, 1:T}}^2_{C^2_{1:T}} + \log \,\det\big(2 \, \pi \, C^2_{1:T}\big) + \normbig{u^3_{1, 1:T}}^2_{C^3_{1:T}} + \log \,\det\big(2 \, \pi \, C^3_{1:T}\big)\nonumber,
\end{align}
where we omit the factor $1 \, / \, 2$ and let $v_{1:T} \coloneqq \big\lbrace u^1_{0, 1:T}, u^2_{1, 1:T}, u^3_{1, 1:T} \big\rbrace$ for simplicity. In the above Equation~\eqref{equ:r-dgp-batch-obj}, notation $\norm{x}_G = (x \, G^{-1} \, x)^{1 \, / \, 2}$ stands for the $G$-weighted Euclidean norm given a non-singular matrix $G$. We write $C^1_{1:T} \in \R^{T \times T}$ for the matrix obtained by evaluating the covariance function $C^1$ on the Cartesian grid $(t_1, \ldots, t_T) \times (t_1, \ldots, t_T)$, and similarly for $C^2_{1:T}$ and $C^3_{1:T}$. The noise covariance $\Xi_{1:T}$ is the diagonal matrix of $\Xi_1,\ldots, \Xi_T$. We now introduce the regularisation term 
\begin{equation}
	\begin{split}
		\mathcal{L}^{\mathrm{B-REG}} &\coloneqq \mathcal{L}^{\mathrm{B-REG}}(v_{1:T})\\
		&= \lambda_1 \, \normbig{ \Phi^1 \, u^1_{0, 1:T}}_1 + \lambda_2 \, \normbig{ \Phi^2 \, u^2_{1, 1:T}}_1 + \lambda_3 \, \normbig{ \Phi^3 \, u^3_{1, 1:T}}_1,
		\label{equ:b-reg}
	\end{split}
\end{equation}
where the positive parameters $\lambda_1$, $\lambda_2$, and $\lambda_3$ stand for the strength of regularisation. The regularised batch DGP (R-DGP) regression aims at solving
\begin{equation}
	v_{1:T} = \argmin_{v_{1:T}} \mathcal{L}^{\mathrm{B}} + \mathcal{L}^{\mathrm{B-REG}}.
	\label{equ:r-dgp-argmin}
\end{equation}
\begin{remark}
	It is important to recall that the covariance matrix $C^1_{1:T}$ in Equation~\eqref{equ:r-dgp-batch-obj} depends on the objective variables $u^2_{1, 1:T}$ and $u^3_{1, 1:T}$, hence, $\mathcal{L}^\mathrm{B}$ can be non-convex. 
\end{remark}

\subsection*{ADMM solution of regularised batch DGP regression}
There are many approaches to solving penalised optimisation problems of the form given in Equation~\eqref{equ:r-dgp-argmin}~\citep{Ruszczynski2006, Wright2006}, however $\mathcal{L}^{\mathrm{B}}$ is usually non-convex, and $\mathcal{L}^\mathrm{B-REG}$ is not differentiable-everywhere. Heuristically, we can interpret the gradients of $\mathcal{L}^\mathrm{B-REG}$ at non-differential points by subgradients and then use gradient descent (GD) methods to find the minima. These GD-based approaches, however, can suffer from a slow convergence rate~\citep{Hastie2015} which limits their applicability in practice.

\citet{Zhao2021RSSGP} propose using the alternating direction method of multipliers \citep[ADMM, ][]{Boyd2004Convex} for solving the optimisation problem in Equation~\eqref{equ:r-dgp-argmin}. The idea is to split the complicated optimisation problem in Equation~\eqref{equ:r-dgp-argmin} into simpler subproblems. 
To do so, let us introduce auxiliary variables $\theta_{1:T} \coloneqq \big\lbrace \theta^1_{1:T}, \theta^2_{1:T}, \theta^3_{1:T}\big\rbrace$, where $\theta^1_{1:T}\in\R^T$, $\theta^2_{1:T}\in\R^T$, and $\theta^3_{1:T}\in\R^T$. We rewrite Equation~\eqref{equ:r-dgp-argmin} as an equality constrained problem
\begin{align}
	v_{1:T} &= \argmin_{v_{1:T}} \normbig{u^1_{0, 1:T} - y_{1:T}}^2_{\Xi_{1:T}} + \normbig{u^1_{0, 1:T}}^2_{C^1_{1:T}} + \log \,\det\big(2 \, \pi \, C^1_{1:T}\big) \nonumber\\
	&\quad+ \normbig{u^2_{1, 1:T}}^2_{C^2_{1:T}} + \log \,\det\big(2 \, \pi \, C^2_{1:T}\big) + \normbig{u^3_{1, 1:T}}^2_{C^3_{1:T}} + \log \,\det\big(2 \, \pi \, C^3_{1:T}\big) \nonumber\\
	&\quad+ \lambda_1 \, \normbig{\theta^1_{1:T}}_1 + \lambda_2 \, \normbig{\theta^2_{1:T}}_1 + \lambda_3 \, \normbig{\theta^3_{1:T}}_1 \nonumber\\
	&\text{subject to } \label{equ:r-dgp-equal-constrained}\\
	&\quad\theta^1_{1:T}=\Phi^1 \, u^1_{0, 1:T}, \quad \theta^2_{1:T}=\Phi^2 \, u^2_{1, 1:T}, \quad \theta^3_{1:T}=\Phi^3 \, u^3_{1, 1:T}.\nonumber
\end{align}
Then let us introduce multiplier variables $\eta_{1:T} \coloneqq \big\lbrace \eta^1_{1:T}, \eta^2_{1:T}, \eta^3_{1:T}\big\rbrace$, where $\eta^1_{1:T}\in\R^T$, $\eta^2_{1:T}\in\R^T$, and $\eta^3_{1:T}\in\R^T$. We can construct the augmented Lagrangian function $\mathcal{L}(v_{1:T}, \theta_{1:T}, \eta_{1:T})$ associated with problem~\eqref{equ:r-dgp-equal-constrained} as
\begin{align}
	&\mathcal{L}^\mathrm{A}(v_{1:T}, \theta_{1:T}, \eta_{1:T}) \nonumber\\
	&= \normbig{u^1_{0, 1:T} - y_{1:T}}^2_{\Xi_{1:T}} + \normbig{u^1_{0, 1:T}}^2_{C^1_{1:T}} + \log \,\det\big(2 \, \pi \, C^1_{1:T}\big) \nonumber\\
	&\quad+ \normbig{u^2_{1, 1:T}}^2_{C^2_{1:T}} + \log \,\det\big(2 \, \pi \, C^2_{1:T}\big) + \normbig{u^3_{1, 1:T}}^2_{C^3_{1:T}} + \log \,\det\big(2 \, \pi \, C^3_{1:T}\big) \nonumber\\
	&\quad+ \lambda_1  \, \normbig{\theta^1_{1:T}}_1 + \big(\eta^1_{1:T}\big)^\trans \, \big(\Phi^1 \, u^1_{0,1:T} - \theta^1_{1:T}\big) \label{equ:r-dgp-lag}\\
	&\quad+ \lambda_2  \, \normbig{\theta^2_{1:T}}_1 + \big(\eta^2_{1:T}\big)^\trans \, \big(\Phi^2 \, u^2_{1,1:T} - \theta^2_{1:T}\big) \nonumber\\
	&\quad+ \lambda_3  \, \normbig{\theta^3_{1:T}}_1 + \big(\eta^3_{1:T}\big)^\trans \, \big(\Phi^3 \, u^3_{1,1:T} - \theta^3_{1:T}\big) \nonumber\\
	&\quad+ \frac{\rho_1}{2} \,  \normbig{\Phi^1 \, u^1_{0,1:T} - \theta^1_{1:T}}^2_2 + \frac{\rho_2}{2} \,  \normbig{\Phi^2 \, u^2_{1,1:T} - \theta^2_{1:T}}^2_2 + \frac{\rho_3}{2} \,  \normbig{\Phi^3 \, u^3_{1,1:T} - \theta^3_{1:T}}^2_2,\nonumber
\end{align}
where $\rho_1>0$, $\rho_2>0$, and $\rho_3>0$ are penalty parameters. The ADMM method works by generating a sequence of estimates $\big\lbrace v^{(i)}_{1:T}, \theta^{(i)}_{1:T}, \eta^{(i)}_{1:T} \colon i=0,1,\ldots\big\rbrace$ to iteratively approximate the optimal $\lbrace v_{1:T}, \theta_{1:T}, \eta_{1:T}\rbrace$ of Equation~\eqref{equ:r-dgp-lag}, as shown in the following algorithm.

\begin{algorithm}[ADMM for R-DGP regression]
	\label{alg:admm}
	Let $\big\lbrace v^{(0)}_{1:T}, \theta^{(0)}_{1:T}, \eta^{(0)}_{1:T}\big\rbrace$ be a given initial estimate. Then for $i=0,1,\ldots$, the ADMM algorithm updates the estimate by solving the following subproblems iteratively
	\begin{align}
		v^{(i+1)}_{1:T} &= \argmin_{v_{1:T}} \normbig{u^1_{0, 1:T} - y_{1:T}}^2_{\Xi_{1:T}} + \normbig{u^1_{0, 1:T}}^2_{C^1_{1:T}} + \log \,\det\big(2 \, \pi \, C^1_{1:T}\big) \nonumber\\
		&\quad+ \normbig{u^2_{1, 1:T}}^2_{C^2_{1:T}} + \log \,\det\big(2 \, \pi \, C^2_{1:T}\big) + \normbig{u^3_{1, 1:T}}^2_{C^3_{1:T}} + \log \,\det\big(2 \, \pi \, C^3_{1:T}\big) \nonumber\\
		&\quad+ \big(\eta^{1, (i)}_{1:T}\big)^\trans \, \big(\Phi^1 \, u^1_{0,1:T} - \theta^{1, (i)}_{1:T}\big) + \frac{\rho_1}{2}  \, \normbig{\Phi^1 \, u^1_{0,1:T} - \theta^{1, (i)}_{1:T}}^2_2 \label{equ:r-dgp-sub-f}\\
		&\quad+ \big(\eta^{2, (i)}_{1:T}\big)^\trans \, \big(\Phi^2 \, u^2_{1,1:T} - \theta^{2, (i)}_{1:T}\big) + \frac{\rho_2}{2}  \, \normbig{\Phi^2 \, u^2_{1,1:T} - \theta^{2, (i)}_{1:T}}^2_2 \nonumber\\
		&\quad+ \big(\eta^{3, (i)}_{1:T}\big)^\trans \, \big(\Phi^3 \, u^3_{1,1:T} - \theta^{3, (i)}_{1:T}\big) + \frac{\rho_3}{2} \,  \normbig{\Phi^3 \, u^3_{1,1:T} - \theta^{3, (i)}_{1:T}}^2_2, \nonumber
	\end{align}
	\begin{equation}
		\begin{split}
			\theta^{(i+1)}_{1:T} &= \argmin_{\theta_{1:T}} \lambda_1 \normbig{\theta^1_{1:T}}_1 + \frac{\rho_1}{2}  \, \normBig{\Phi^1 \, u^{1, (i+1)}_{0,1:T} - \theta^1_{1:T} + \frac{1}{\rho_1}\,\eta^{1, (i)}}_2^2\\
			&\quad+ \lambda_2 \normbig{\theta^2_{1:T}}_1 + \frac{\rho_2}{2} \,  \normBig{\Phi^2 \, u^{2, (i+1)}_{1,1:T} - \theta^2_{1:T} + \frac{1}{\rho_2}\,\eta^{2, (i)}}_2^2\\
			&\quad+ \lambda_3 \normbig{\theta^3_{1:T}}_1 + \frac{\rho_3}{2} \,  \normBig{\Phi^3 \, u^{3, (i+1)}_{1,1:T} - \theta^3_{1:T} + \frac{1}{\rho_3}\,\eta^{3, (i)}}_2^2,
			\label{equ:r-dgp-sub-theta}
		\end{split}
	\end{equation}
	and
	\begin{equation}
		\begin{split}
			\eta^{1, (i+1)}_{1:T} &= \eta^{1, (i)}_{1:T} + \rho_1\, \big( \Phi^1 \, u^{1, (i+1)}_{0,1:T} - \theta^{1, (i+1)}_{1:T} \big), \\
			\eta^{2, (i+1)}_{1:T} &= \eta^{2, (i)}_{1:T} + \rho_2\, \big( \Phi^2 \, u^{2, (i+1)}_{1,1:T} - \theta^{2, (i+1)}_{1:T} \big), \\
			\eta^{3, (i+1)}_{1:T} &= \eta^{3, (i)}_{1:T} + \rho_3\, \big( \Phi^3 \, u^{3, (i+1)}_{1,1:T} - \theta^{3, (i+1)}_{1:T} \big). 
		\end{split}
	\end{equation}
\end{algorithm}

The subproblem in Equation~\eqref{equ:r-dgp-sub-f} is a standard unconstrained optimisation problem which can be solved numerically by a vast number of non-linear optimisers. For a review of such optimisers, we refer the reader to \citet{Wright2006}. As for the subproblem in Equation~\eqref{equ:r-dgp-sub-theta}, one can use the soft thresholding scheme in order to obtain a closed-form solution~\citep{Hastie2015, Boyd2011admm}.

\subsection*{Convergence analysis of Algorithm~\ref{alg:admm}}
The goal is now to analyse whether the sequence generated by Algorithm~\ref{alg:admm} converges to a local minimum. For notational convenience we concatenate the objective variables in $v_{1:T}$ in a vector $\overline{v}_{1:T} \in \R^{3 \, T}$ defined by $\overline{v}_{1:T} \coloneqq \begin{bmatrix}
	u^1_{0, 1:T} & u^2_{1, 1:T} & u^3_{1, 1:T}
\end{bmatrix}^\trans$.

We consider the following assumptions on the DGP and the minimisation problem parameters.

\begin{assumption}
	\label{assump:r-dgp-mineig-C}
	The covariance matrix $C^1_{1:T}$ is strictly positive definite. That is, $\mineig(C^1_{1:T})$ has a positive lower bound uniformly for all $u^2_{1, 1:T}\in\R^T$ and $u^3_{1,1:T}\in\R^T$.
\end{assumption}

\begin{assumption}
	\label{assump:r-dgp-rho}
	The penalty parameters $\rho_i$ and sparsity parameters $\Phi_i$ for $i=1,2,3$ satisfy
	\begin{equation}
		\frac{\rho_i}{2} \, \big(\mineig(\Phi_i)\big)^2 - \frac{c_v}{2} \geq 0,
	\end{equation}
	where the constant $c_v$ is defined in Lemma~\ref{lemma:r-dgp-lip}.
\end{assumption}

Assumption~\ref{assump:r-dgp-mineig-C} is an important prerequisite for Theorem~\ref{thm:admm-converge} because the proof in~\citet{Zhao2021RSSGP} requires that $(C^1_{1:T})^{-1}$ be bounded so that the Lagrangian function in Equation~\eqref{equ:r-dgp-lag} admits a lower bound independent of its arguments.

The constant $c_v$ in Assumption~\ref{assump:r-dgp-rho} is a fixed number determined by $\mathcal{L}^\mathrm{B}$ (and therefore by the DGP model itself) and is independent of data $y_{1:T}$. This gives a lower bound on the free penalty parameters $\rho_1$, $\rho_2$, and $\rho_3$ for the problem to be well defined.

\begin{remark}
	\citet{Zhao2021RSSGP} considers the specific case of a non-stationary \matern covariance function $C_{\mathrm{NS}}$ defined as per Equation~\eqref{equ:cov-ns-matern}. Due to this choice, additional assumptions need to be introduced so as to proceed with the convergence analysis of the problem in Equation~\eqref{equ:r-dgp-argmin}.
\end{remark}

\begin{lemma}[Lipschitz condition]
	\label{lemma:r-dgp-lip}
	Suppose that there exists a constant $c_v>0$ such that the norm $\norm{\hessian_{\overline{v}_{1:T}} \mathcal{L}^B}_2 \leq c_v$ for all $\overline{v}_{1:T} \in \R^{3 \, T}$. Then for every two vectors $\overline{v}^{1}_{1:T}, \overline{v}^{2}_{1:T} \in \R^{3 \, T}$,
	\begin{equation}
		\begin{split}
			&\absBig{\mathcal{L}^\mathrm{B}\big(\overline{v}^{1}_{1:T}\big) - \mathcal{L}^\mathrm{B}\big(\overline{v}^{2}_{1:T}\big) - \big( \overline{v}^{1}_{1:T} - \overline{v}^{2}_{1:T} \big)^\trans\,\nabla_{\overline{v}_{1:T}} \mathcal{L}^\mathrm{B}\big(\overline{v}^{2}_{1:T}\big)} \\
			&\quad\leq \frac{c_v}{2} \, \normbig{\overline{v}^{1}_{1:T} - \overline{v}^{2}_{1:T}}^2_2.
		\end{split}
	\end{equation}
\end{lemma}
\begin{proof}
	See, Lemma 1.2.2 and 1.2.3 in~\citet{Nesterov2004}.
\end{proof}

\begin{theorem}
	\label{thm:admm-converge}
	Suppose that Assumptions~\ref{assump:r-dgp-mineig-C} and \ref{assump:r-dgp-rho} hold. Further assume that the subproblem in Equation~\eqref{equ:r-dgp-sub-f} has a stationary point. Then the sequence $\big\lbrace v^{(i)}_{1:T}, \theta^{(i)}_{1:T}, \eta^{(i)}_{1:T} \colon i=0,1,\ldots \big\rbrace$ generated by Algorithm~\ref{alg:admm} converges to a local minimum.
\end{theorem}
\begin{proof}
	The key is to prove that the sequence $\big\lbrace \mathcal{L}^\mathrm{A}\big(v^{(i)}_{1:T}, \theta^{(i)}_{1:T}, \allowbreak \eta^{(i)}_{1:T}\big)\colon i=0,1,\ldots \big\rbrace$ is non-increasing and lower bounded over $i=0,1,\ldots$ Using the convexity of subproblem~\eqref{equ:r-dgp-sub-theta} we can then prove the convergence of Algorithm~\ref{alg:admm}~\citep[see, e.g.,][]{Boyd2004Convex, Nesterov2018optimization}. We refer the reader to~\citet{Zhao2021RSSGP} for details.
\end{proof}

\subsection*{Regularised SS-DGP regression}
We can also derive the state-space versions of regularised DGPs. However, the resulting method turns out to be very similar to that of batch DGPs. We therefore only sketch out the basic idea in this section, and refer to \citet{Zhao2021RSSGP} for details.

Consider the state-space representation of the DGP defined in Equation~\eqref{equ:r-dgp-reg-model}. The first step is to derive the state-space version of the MAP objective function in Equation~\eqref{equ:r-dgp-batch-obj}. One approach is to discretise the SDE as in Equation~\eqref{equ:ss-dgp-disc}, then we can factorise the SDE prior density over $t_1,\ldots, t_T$. As shown in~\citet{Zhao2021RSSGP}, the state-space (approximate) MAP objective function reads
\begin{equation}
	\begin{split}
		\mathcal{L}^\mathrm{S}(v_{1:T}) = v_0^\trans \, P_0^{-1}\,v_0 + \sum^T_{k=1} \Big[ &\norm{y_k - H \, v_k}^2_{\Xi_k} + \norm{v_k - f_{k-1}(v_{k-1})}^2_{Q_{k-1}(v_{k-1})} \\
		&+ \log \, \det\big(2 \, \pi \, Q_{k-1}(v_{k-1})\big)\Big],
	\end{split}
\end{equation}
where the matrix $H$ selects $U^1_0$ from $V$. Now, let $\Psi^1$, $\Psi^2$, and $\Psi^3$ be three suitable matrices that, respectively, select the components $U^1_0$, $U^2_1$, and $U^3_1$ from $V$. We can now introduce the regularisation term
\begin{equation}
	\mathcal{L}^{\mathrm{S-REG}}(v_{1:T}) = \sum^T_{k=0} \sum^3_{i=1} \lambda_i \, \norm{\Psi^i \, v_k}_1.
\end{equation}
One can then analogously derive the augmented Lagrangian function and the corresponding ADMM algorithm in their state-space versions.

\begin{remark}
	The computational complexities of $\mathcal{L}^\mathrm{S}$ and $\mathcal{L}^\mathrm{B}$ are linear and cubic with respect to $T$, respectively. Hence, the state-space version is significantly computationally cheaper than the batch version when the number of measurements is large. 
\end{remark}

\begin{remark}
	$\mathcal{L}^\mathrm{B}$ and $\mathcal{L}^\mathrm{S}$ are generally not equal, since discretisations of SS-DGPs often involve approximations.
\end{remark}

\subsection*{Uncertainty quantification}
The regularised DGP regression method presented in this section is rooted in the MAP framework, which provides point estimates of the quantities at hand, and ignores the uncertainty in the solution. This can be partially remedied by using, for example, Laplace's method to approximate $p_{V_{1:T} \cond Y_{1:T}}(v_{1:T} \cond y_{1:T})$ around the MAP estimate with a Gaussian density~\citep{Bishop2006}. However, computing the Hessian (of dimension $T \times T$) can be computationally intensive and limits the applicability for high dimensional problems such as ours. 

\begin{figure}[t!]
	\centering
	\includegraphics[width=.8\linewidth]{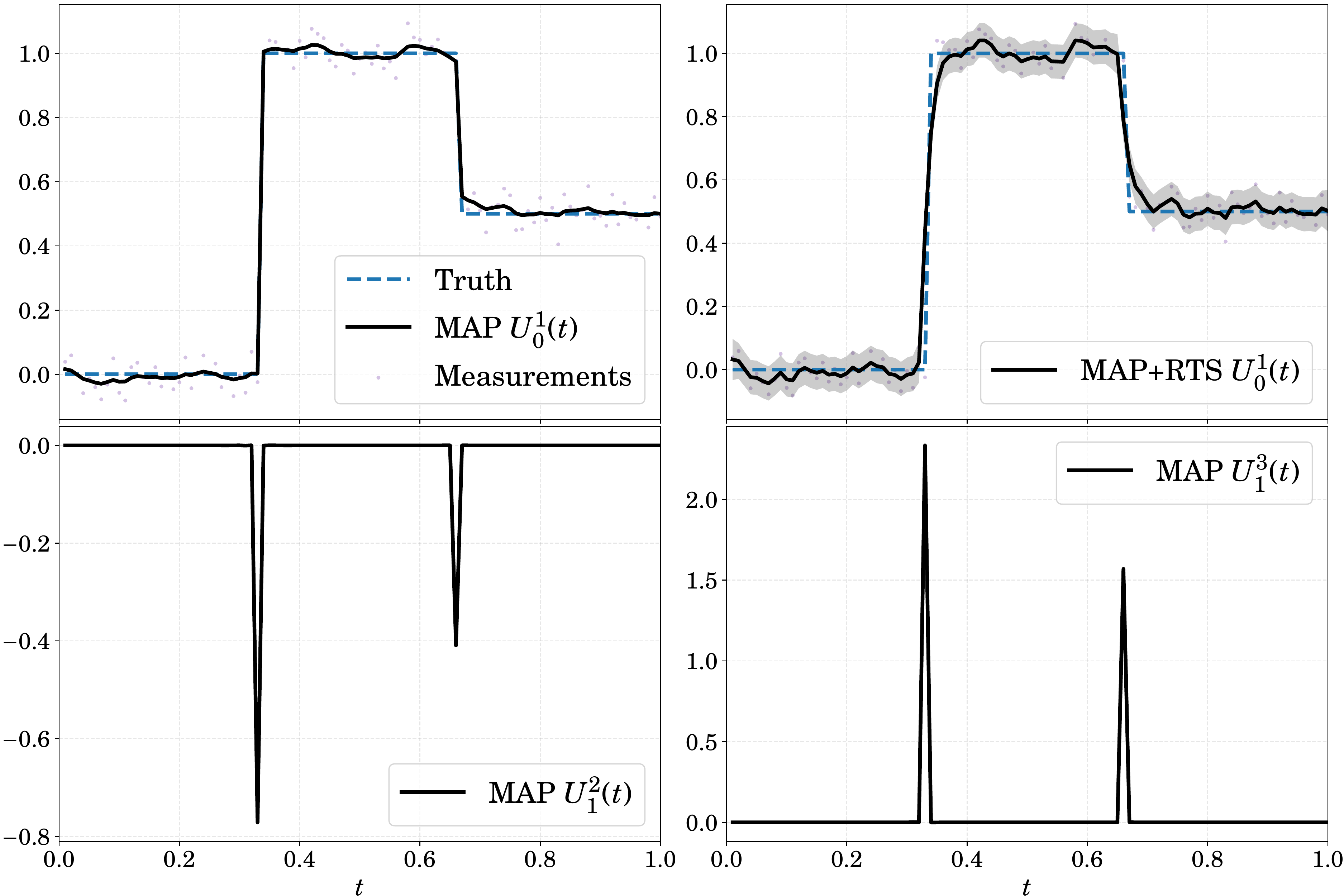}
	\caption{Regularised SS-DGP regression on a rectangular signal. The uncertainty is quantified by using Equation~\eqref{equ:r-ssdgp-uncertainty} and an RTS smoother.}
	\label{fig:r-ssdgp-admm}
\end{figure}

Another solution is to solve the subproblems of ADMM by using Bayesian solvers instead of deterministic optimisers, resulting in an estimate of the uncertainty in the form of a posterior distribution on the solution. For instance, it is known that the iterated extended Kalman smoother is in some sense equivalent to the Gauss--Newton method~\citep{Bell1994, Simo2020IEKFS}, and \citet{Rui2019ieks} showed that this connection could be extended to the ADMM method. For a review of these, we refer the reader to the discussion in \citet{Rui2020thesis}.

However, if we are only interested in the marginal posterior density $p_{U^1_{0, 1:T} \cond Y_{1:T}}\big(u^1_{0, 1:T} \cond y_{1:T}\big)$ instead of the full density $p_{V_{1:T} \cond Y_{1:T}}(v_{1:T} \cond y_{1:T})$, then we can leverage the hierarchical nature of DGPs to approximate the marginal density efficiently. In order to do so, we can write the approximation
\begin{equation}
	\begin{split}
		&p_{U^1_{0, 1:T} \cond U^2_{1, 1:T}, U^3_{1, 1:T}, Y_{1:T}}\big(u^1_{0, 1:T} \cond u^2_{1, 1:T}, u^3_{1, 1:T}, y_{1:T}\big) \\
		&\approx p_{U^{1}_{0, 1:T} \cond Y_{1:T}}\big(u^1_{0, 1:T} \cond u^{2, \star}_{1, 1:T}, u^{3, \star}_{1, 1:T}, y_{1:T}\big),
		\label{equ:r-ssdgp-uncertainty}
	\end{split}
\end{equation}
where $u^{2, \star}_{1, 1:T}, u^{3, \star}_{1, 1:T}$ stand for the MAP estimates of $U^{2}_{1, 1:T}\mid y_{1:T}$ and $U^{3}_{1, 1:T} \mid y_{1:T}$. Afterwards, computing $p_{U^1_{0, 1:T} \cond Y_{1:T}}\big(u^1_{0, 1:T} \cond y_{1:T}\big)$ simply consists in solving a standard GP regression problem, which can be obtained in closed form~\citep{Zhao2021RSSGP}.

Figure~\ref{fig:r-ssdgp-admm} illustrates such an example of regularised SS-DGP, where we set the sparsity inducing matrices to be identity matrices~\citep{Zhao2021RSSGP}. The latent states $U^2_1$ and $U^3_1$ exhibit spiking behaviours, being almost zero except at the two discontinuities.

\chapter{Applications}
\label{chap:apps}
In this chapter, we present the experimental results in Publications~\cp{paperDRIFT}, \cp{paperKFSECG}, \cp{paperKFSECGCONF}, \cp{paperSSDGP}, and~\cp{paperMARITIME}. These works are mainly concerned with the applications of state-space (deep) GPs. Specifically, in Section~\ref{sec:drift-est} we show how to use the SS-GP regression method to estimate unknown drift functions in SDEs. Similarly, under that same state-space framework, in Section~\ref{sec:spectro-temporal} we show how to estimate the posterior distributions of the Fourier coefficients of signals. Sections~\ref{sec:apps-ssdgp} and~\ref{sec:maritime} illustrate how SS-DGPs can be used to model real-world signals, such as gravitational waves, accelerometer recordings of human motion, and maritime vessel trajectories.

\section{Drift estimation in stochastic differential equations}
\label{sec:drift-est}
Consider a scalar-valued stochastic process $X \colon \T \to \R$ governed by a stochastic differential equation
\begin{equation}
	\diff X(t) = a(X(t)) \diff t + b \diff W(t), \quad X(t_0) = X_0,
	\label{equ:drift-est-sde}
\end{equation}
where $b\in\R$ is a constant, $W\colon \T\to\R$ is a Wiener process, and $a\colon \R\to\R$ is an \emph{unknown} drift function. Suppose that we have measurement random variables $X(t_1), X(t_2), \ldots, X(t_T)$ of $X$ at time instances $t_1, t_2, \ldots, t_T\in\T$, the goal is to estimate the drift function $a$ from these measurements.

One way to proceed is to assume a parametric form of function $a=a_\vartheta(\cdot)$ and estimate its parameters $\vartheta$ by using, for example, maximum likelihood estimation~\citep{Zmirou1986, Yoshida1992, Kessler1997, Sahalia2003} or Monte Carlo methods~\citep{Roberts2001, Beskos2006}.

In this chapter, we are mainly concerned with the GP regression approach for estimating the unknown $a$~\citep{Papaspiliopoulos2012, Ruttor2013, Garcia2017, Batz2018, Opper2019}. The key idea of this approach is to assume that the unknown drift function is distributed according to a GP, that is
\begin{equation}
	a(x) \sim \mathrm{GP}(0, C(x, x')).
\end{equation}
Having at our disposal measurements $X(t_1), X(t_2), \ldots, X(t_T)$ observed directly from SDE~\eqref{equ:drift-est-sde}, we can formulate the problem of estimating $a$ as a GP regression problem. In order to do so, we discretise the SDE in Equation~\eqref{equ:drift-est-sde} and thereupon define the measurement model as 
\begin{equation}
	Y_k \coloneqq X(t_{k}) - X(t_{k-1}) = \check{f}_{k-1}(X(t_{k-1})) + \check{q}_{k-1}(X(t_{k-1}))
\end{equation}
for $k=1,2,\ldots, T$, where the function $\check{f}_{k-1}\colon \R \to \R$ and the random variable $\check{q}_{k-1}$ represent the exact discretisation of $X$ at $t_{k}$ from $t_{k-1}$. We write the GP regression model for estimating the drift function by
\begin{equation}
	\begin{split}
		a(x) &\sim \mathrm{GP}(0, C(x, x')),\\
		Y_k &= \check{f}_{k-1}(X_{k-1}) + \check{q}_{k-1}(X_{k-1}).
		\label{equ:drift-est-reg-model}
	\end{split}
\end{equation}
The goal now is to estimate the posterior density of $a(x)$ for all $x\in\R$ from a set of data $y_{1:T} = \lbrace x_k - x_{k-1} \colon k=1,2,\ldots, T \rbrace$. 

However, the exact discretisation of non-linear SDEs is rarely possible. In practice, we often have to approximate $\check{f}_{k-1}$ and $\check{q}_{k-1}$ by using, for instance, Euler--Maruyama scheme, Milstein's method, or more generally It\^{o}--Taylor expansions~\citep{Kloeden1992}. As an example, application of the Euler--Maruyama method to Equation~\eqref{equ:drift-est-sde} gives 
\begin{equation}
	\begin{split}
		\check{f}_{k-1} &\approx a(x) \, \Delta t_k, \\
		\check{q}_{k-1} &\approx b \, \delta_k, 
	\end{split}
\end{equation}
where $\Delta t_k \coloneqq t_{k} - t_{k-1}$ and $\delta_k \sim \mathrm{N}(0, \Delta t_k)$. 

However, the discretisation by the Euler--Maruyama scheme can sometimes be crude, especially when the discretisation step is relatively large, making the measurement representation obtained from it inaccurate.~\citet{ZhaoZheng2020Drift} show that if the prior of $a$ is chosen of certain regularities, it is possible to leverage high-order It\^{o}--Taylor expansions in order to discretise the SDE with higher accuracy. As an example, suppose that the GP prior $a$ is twice-differentiable almost surely. Then, the It\^{o}--Taylor strong order 1.5 (It\^{o}-1.5) method~\citep{Kloeden1992} gives
\begin{equation}
	\begin{split}
		\check{f}_{k-1}(x) &\approx a(x) \, \Delta t_{k} + \frac{1}{2} \Big( \frac{\diff a}{\diff x}(x) \, a(x) + \frac{1}{2}\, \frac{\diff^2 a}{\diff x^2}(x) \, b^2 \Big) \, \Delta t_k^2,\\
		\check{q}_{k-1}(x) &\approx b\,\delta_{1,k} + \frac{\diff a}{\diff x}(x) \, b \, \delta_{2, k},
		\label{equ:drift-est-ito15}
	\end{split}
\end{equation}
where
\begin{equation}
	\begin{bmatrix}
		\delta_{1,k} \\
		\delta_{2,k}
	\end{bmatrix} \sim \mathrm{N}\left(
	\begin{bmatrix}
		0\\0
	\end{bmatrix}, 
	\begin{bmatrix}
		\frac{(\Delta t_k)^3}{3} & \frac{(\Delta t_k)^2}{2}\\
		\frac{(\Delta t_k)^2}{2} & \Delta t_k
	\end{bmatrix}\right).
\end{equation}
Indeed, using a higher order It\^{o}--Taylor expansion can lead to a better measurement representation, however, this in turn requires more computations and limits the choice of the prior model. It is also worth mentioning that if one uses the approximations of high order It\^{o}--Taylor expansions -- such as the one in Equation~\eqref{equ:drift-est-ito15} -- the resulting measurement representation in the GP regression model~\eqref{equ:drift-est-reg-model} is no longer linear with respect to $a$. Consequently, the GP regression solution may not admit a closed-form solution. 

One problem of this GP regression-based drift estimation approach is that the computation can be demanding if the number of measurements $T$ is large. Moreover, if the measurements are densely located then the covariance matrices used in GP regression may be numerically close to singular. These two issues are already discussed in Introduction and Section~\ref{sec:ssgp}. In addition, the GP regression model is not amenable to high order It\^{o}--Taylor expansions, as these expansions result in non-linear measurement representations and require to compute the derivatives of $a$ up to a certain order.

\begin{figure}[t!]
	\centering
	\includegraphics[width=.99\linewidth]{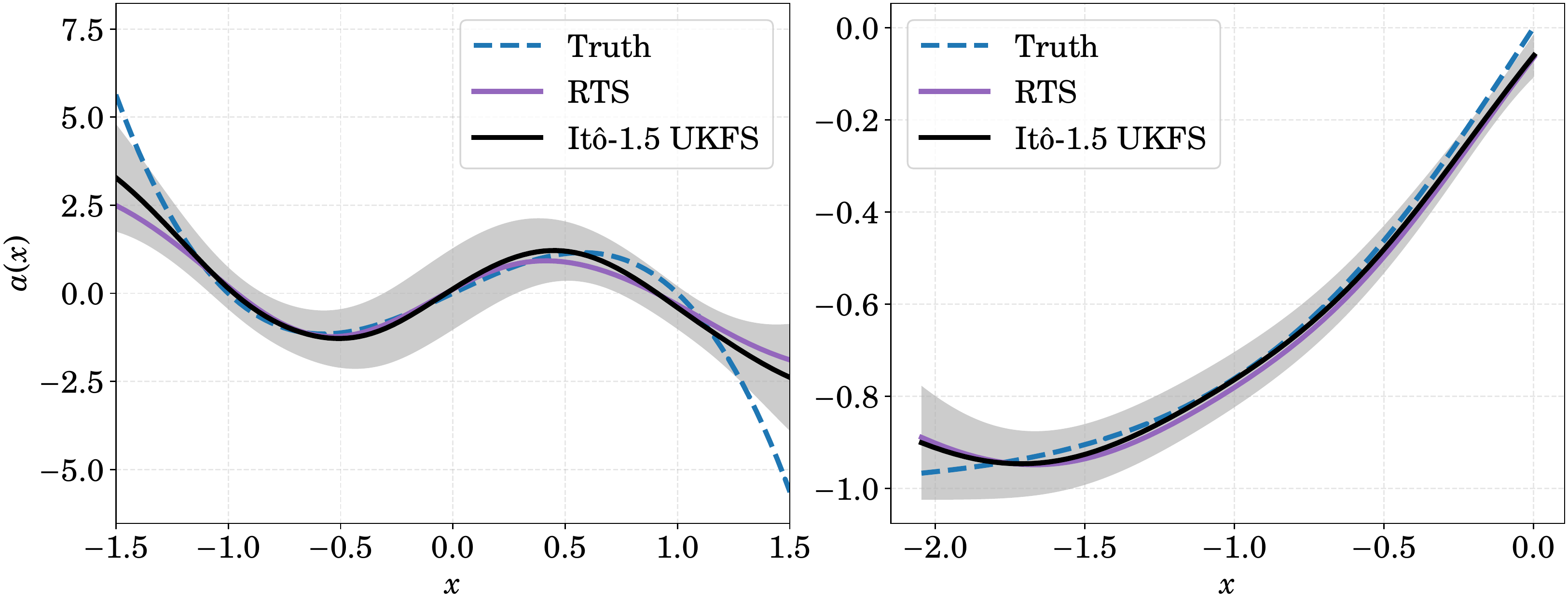}
	\caption{Estimation of drift functions $a(x)=3\,(x-x^3)$ (left) and $a(x)=\tanh(x)$ (right) by \citet{ZhaoZheng2020Drift}. UKFS stands for unscented Kalman filter and RTS smoother (UKFS). Shaded area stands for 0.95 confidence interval associated with the UKFS estimation.}
	\label{fig:drift-est}
\end{figure}

\citet{ZhaoZheng2020Drift} address the problems above by considering solving the GP regression problem in Equation~\eqref{equ:drift-est-reg-model} under the state-space framework. More precisely, they put an SS-GP prior over the unknown $a$ instead of a standard batch GP. The main benefit of doing so for this application is that the SS-GP regression solvers are computationally more efficient for large-scale measurements compared to the standard batch GP regression (see, Introduction and Section~\ref{sec:ssgp}). Moreover, in order to use high order It\^{o}--Taylor expansions, \citet{ZhaoZheng2020Drift} consider putting SS-GP priors over $a$ of the \matern family, so that the derivatives of $a$ naturally appear as the state components of $a$ (see, Section~\ref{sec:deep-matern}). In this way, computing the covariance matrices of the derivatives of $a$ is no longer needed.

\begin{remark}
	Note that the SS-GP approach requires to treat $X(t_1), X(t_2), \ldots ,\allowbreak X(t_T)$ as time variables and sort their data $x_{1:T}=\lbrace x_1,x_2,\ldots,x_T \rbrace$ in temporal order. 
\end{remark}

In Figure~\ref{fig:drift-est}, we show a representative result from~\citet{ZhaoZheng2020Drift}, where the SS-GP approach is employed to approximate the drift functions of two SDEs. In particular, the solutions are obtained by using the It\^{o}-1.5 discretisation, and an unscented Kalman filter and an RTS smoother. For more details regarding the experiments the reader is referred to~\citet{ZhaoZheng2020Drift}.

\section{Probabilistic spectro-temporal signal analysis}
\label{sec:spectro-temporal}
Let $z\colon \T \to \R$ be a periodic signal. In signal processing, it is often of interest to approximate the signal by Fourier expansions of the form
\begin{equation}
	z(t) \approx \alpha_0 + \sum^{N}_{n=1} \big[\alpha_n \cos(2 \, \pi \, \mathring{f}_n \, t) + \beta_n \sin(2 \, \pi \, \mathring{f}_n \, t)\big],
\end{equation}
where $\big\lbrace \mathring{f}_n\colon n=1,2,\ldots,N \big\rbrace$ stand for the frequency components, and $N$ is a given expansion order. When $z$ satisfies certain conditions~\citep{Katznelson2004}, the representation in the equation above converges as $N\to\infty$ (in various modes). 

Let use denote $y_k\coloneqq y(t_k)$ and suppose that we have a set of measurement data $y_{1:T}=\lbrace y_k\colon k=1,2,\ldots,T \rbrace$ of the signal at time instances $t_1, t_2, \ldots, t_T\in\T$. In order to quantify the truncation and measurement errors, we introduce Gaussian random variables $\xi_k \sim \mathrm{N}(0, \Xi_k)$ for $k=1,2,\ldots,T$ and let
\begin{equation}
	\begin{split}
		Y_k = \alpha_0 + \sum^{N}_{n=1} \big[\alpha_n \cos(2 \, \pi \, \mathring{f}_n \, t_k) + \beta_n \sin(2 \, \pi \, \mathring{f}_n \, t_k)\big] + \xi_k
		\label{equ:spectro-temporal-y}
	\end{split}
\end{equation}
represent the random measurements of $z$ at $t_k$. The goal now is to estimate the coefficients $\lbrace \alpha_0, \alpha_n, \beta_n \colon n=1,2,\ldots, N \rbrace$ from the data $y_{1:T}$. We call this problem the \emph{spectro-temporal estimation} problem.

One way to proceed is by using the MLE method~\citep{Bretthorst1988}, but~\citet{QiYuan2002, ZhaoZheng2018KF, ZhaoZheng2020KFECG} show that we can also consider this spectro-temporal estimation problem as a GP regression problem. More precisely, the modelling assumption is that
\begin{equation}
	\begin{split}
		\alpha_0(t) &\sim \mathrm{GP}\big(0, C^0_\alpha(t, t')\big),\\
		\alpha_n(t) &\sim \mathrm{GP}\big(0, C^n_\alpha(t, t')\big),\\
		\beta_n(t) &\sim \mathrm{GP}\big(0, C^n_\beta(t, t')\big),
		\label{equ:spectro-temporal-gp-priors}
	\end{split}
\end{equation}
for $n=1,2,\ldots, N$, and that the measurements follow
\begin{equation}
	Y_k = \alpha_0(t_k) + \sum^{N}_{n=1} \big[\alpha_n(t_k) \, \cos(2 \, \pi \, \mathring{f}_n \, t_k) + \beta_n(t_k) \, \sin(2 \, \pi \, \mathring{f}_n \, t_k)\big] + \xi_k,\nonumber
\end{equation}
for $k=1,2,\ldots,T$. This results in a standard GP regression problem therefore, the posterior distribution of coefficients $\lbrace \alpha_0, \alpha_n, \beta_n \colon n=1,2,\ldots, N \rbrace$ have a close-form solution. However, solving this GP regression problem is, in practice, infeasible when the expansion order $N$ and the number of measurements $T$ are large. This is due to the fact that one needs to compute $2\,N+1$ covariance matrices of dimension $T\times T$ and compute their inverse.

\citet{ZhaoZheng2018KF} propose to solve this spectro-temporal GP regression problem under the state-space framework, that is, by replacing the GP priors in Equation~\eqref{equ:spectro-temporal-gp-priors} with their SDE representations. Since SS-GPs have already been extensively discussed in previous sections, we omit the resulting state-space spectro-temporal estimation formulations. However, the details can be found in Section~\ref{sec:ssgp} and in~\citet{ZhaoZheng2018KF}.

The computational cost of the state-space spectro-temporal estimation method is substantially cheaper than that of standard batch GP methods. Indeed, Kalman filters and smoothers only need to compute one $E$-dimensional covariance matrix at each time step (see, Algorithm~\ref{alg:kfs}) instead of those required by batch GP methods. The dimension $E$ is equal to the sum of all the state dimensions of the SS-GPs $\lbrace \alpha_0, \alpha_n, \beta_n \colon n=1,2,\ldots, N \rbrace$. 

\citet{ZhaoZheng2020KFECG} further extend the state-space spectro-temporal estimation method by putting quasi-periodic SDE priors~\citep{Solin2014} over the Fourier coefficients instead of the Ornstein--Uhlenbeck SDE priors used in~\citet{ZhaoZheng2018KF}. This consideration generates a time-invariant version of the measurement model in Equation~\eqref{equ:spectro-temporal-y}, thus, one can apply steady-state Kalman filters and smoothers (SS-KFSs) in order to achieve lower computational costs. The computational cost is further reduced because SS-KFSs do not need to compute the $E$-dimensional covariances of the state in their filtering and smoothing loops. Instead, the state covariances in SS-KFSs are replaced by a pre-computed steady covariance matrix obtained as the solution of its discrete algebraic Riccati equation (DARE). Moreover, solving the DARE is independent of data/measurements, which is especially useful when the model is known or fixed. However, SS-KFSs may not always be computationally efficient when $N \gg T$, since solving an $E$-dimensional DARE can be demanding when $E$ is large.

\citet{ZhaoZheng2018KF, ZhaoZheng2020KFECG} show that the state-space spectro-temporal estimation method can be a useful feature extraction mechanism for detecting atrial fibrillation from electrocardiogram signals. More specifically, the spectro-temporal method estimates the spectrogram images of atrial fibrillation signals. These images are then fed to a deep convolutional neural network classifier which is tasked with recognising atrial fibrillation manifestations.

Since the measurement noises $\lbrace \xi_k\colon k=1,2,\ldots,T \rbrace$ in Equation~\eqref{equ:spectro-temporal-y} encode the truncation and measurement errors, it is also of interest to estimate them. This is done in~\citet{GaoRui2019ALKS}, where the variances $\Xi_k$ of $\xi_k$ for $k=1,2,\ldots, T$ are estimated under the alternating direction method of multipliers.

\begin{figure}[t!]
	\centering
	\includegraphics[width=.99\linewidth]{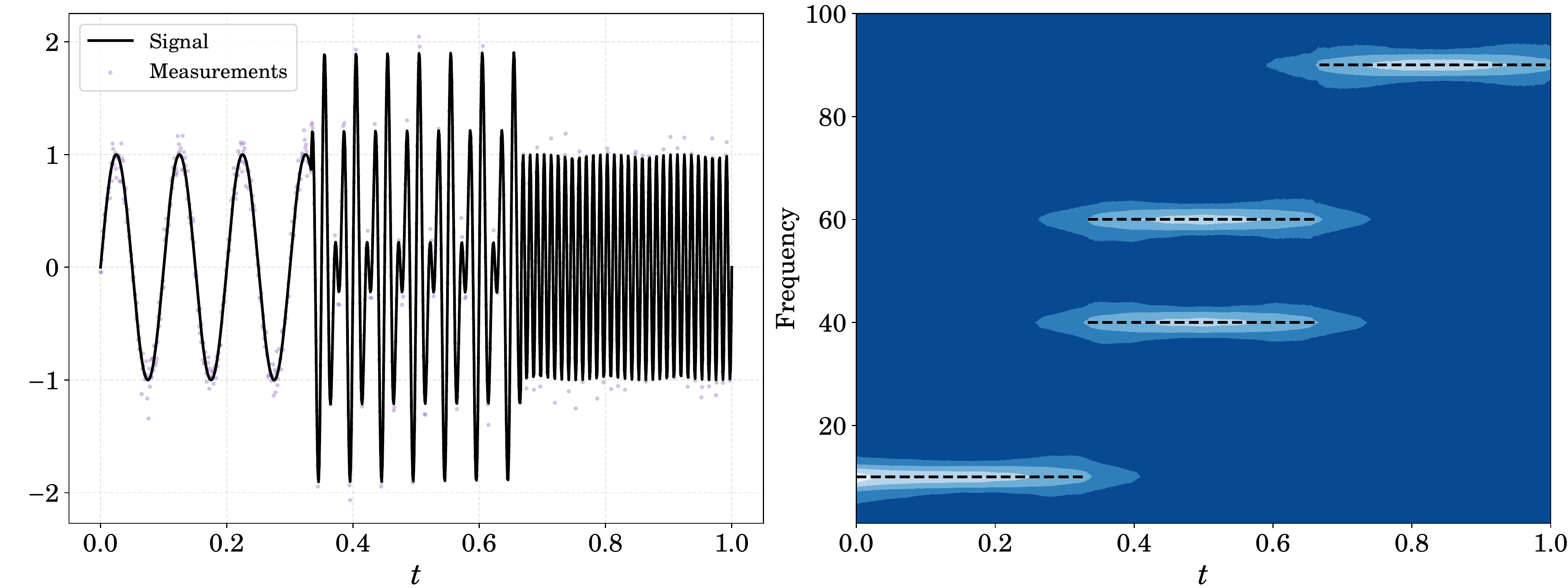}
	\caption{Spectrogram (right, contour plot) of a sinusoidal signal (left) generated by Kalman filtering and RTS smoothing using the method in Section~\ref{sec:spectro-temporal}. Dashed black lines (right) stand for the ground truth frequencies.}
	\label{fig:spectro-temporal-demo}
\end{figure}

Figure~\ref{fig:spectro-temporal-demo} illustrates an example of using the state-space spectro-temporal estimation method to estimate the spectrogram of a sinusoidal signal with multiple frequency bands. 

\section{Signal modelling with SS-DGPs}
\label{sec:apps-ssdgp}
In this section, we apply SS-DGPs for modelling gravitational waves and human motion (i.e., acceleration). We consider these as SS-DGP regression problems, where the measurement models are assumed to be linear with respect to the SS-DGPs with additive Gaussian noises. As for their priors, we chose the Mat\'{e}rn $\nu=3 \, / \, 2$ SS-DGP in Example~\ref{example:ssdgp-m32}, except that the parent GPs $U^2_1$ and $U^3_1$ use the Mat\'{e}rn $\nu=1 \, / \, 2$ representation.

\subsection*{Modelling gravitational waves}
Gravitational waves are curvatures of spacetime caused by the movement of objects with mass~\citep{Maggiore2008}. Since the time Albert Einstein predicted the existence of gravitational waves theoretically from a linearised field equation in 1916~\citep{EinsteinGW1937, Hill2017}, much effort has been done to observe their presence~\citep{Blair1991}. In 2015, the laser interferometer gravitational-wave observatory (LIGO) team first observed a gravitational wave from the merging of a black hole binary~\citep[event GW150914,][]{LIGO2016}. This wave/signal is challenging for standard GPs to fit because the frequency of the signal changes over time. It is then of our interest to see if SS-DGPs can fit this gravitational wave signal.

\begin{figure}[t!]
	\centering
	\includegraphics[width=.99\linewidth]{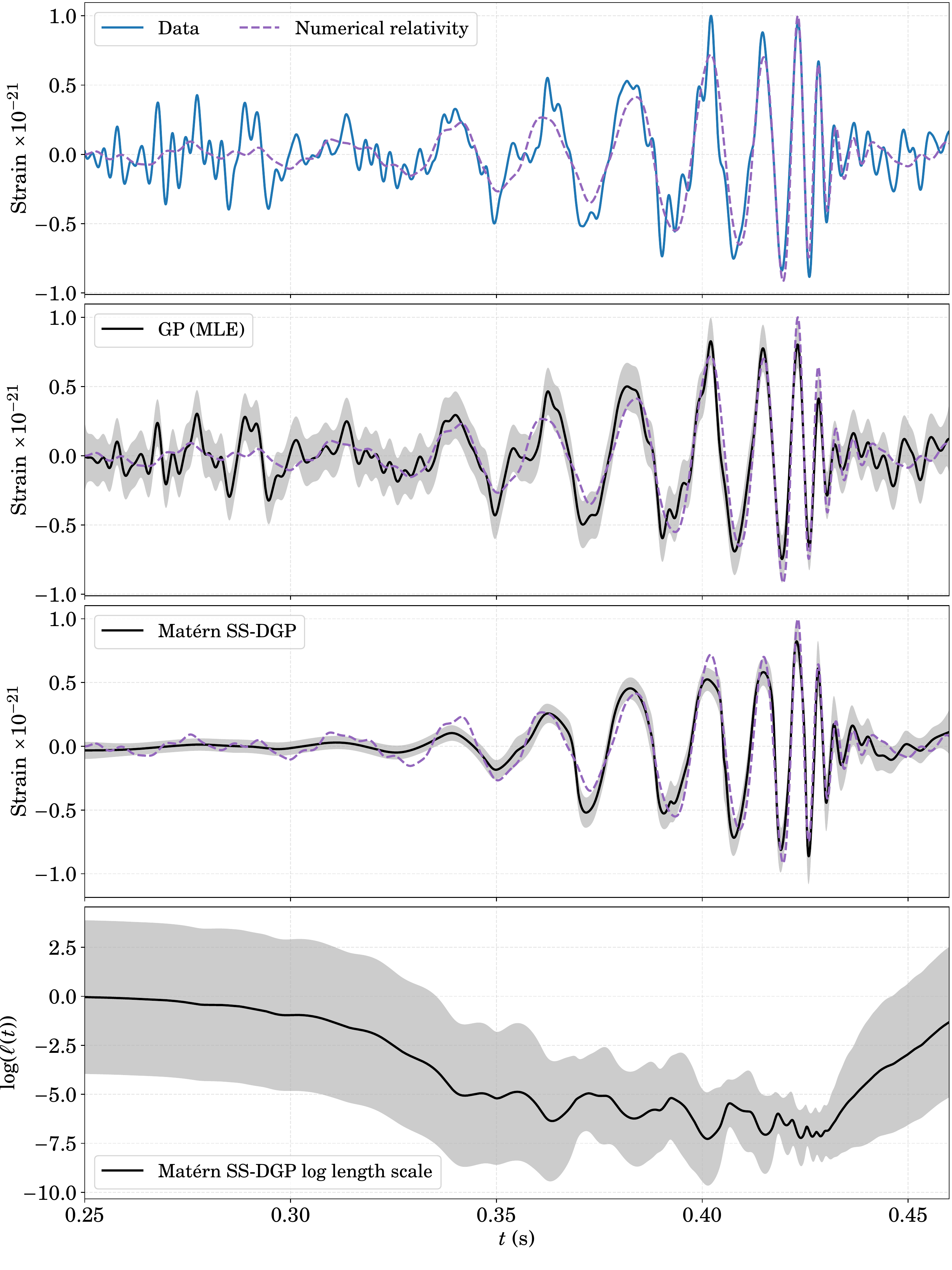}
	\caption{Mat\'{e}rn $\nu=3\, \, /\, 2$ SS-DGP regression (solved by cubature Kalman filter and smoother) for the gravitational wave in event GW150914 (Hanford, Washington). The shaded area stands for 0.95 confidence interval. Details about the data can be found in~\citet{Zhao2020SSDGP}.}
	\label{fig:gravit-wave-ssdgp}
\end{figure}

Figure~\ref{fig:gravit-wave-ssdgp} plots the SS-DGP fit for the gravitational wave observed in the event GW150914. In the same figure, we also show the fit from a \matern $\nu=3\,/\,2$ GP as well as a waveform (which is regarded as the ground truth) computed from the numerical relativity (purple dashed lines) for comparison. Details about the experiment and data are found in~\citet{Zhao2020SSDGP}.

Figure~\ref{fig:gravit-wave-ssdgp} shows that the GP fails to give a reasonable fit to the gravitational wave because the GP over-adapts the high-frequency section of the signal around $0.4$~s. On the contrary, the SS-DGP does not have such a problem, and the fit is closer to the numerical relativity waveform compared that of the GP. Moreover, the estimated length scale (in log transformation) can interpret the data in the sense that the length scale value decreases as the signal frequency increases. 

\subsection*{Modelling human motion}
We apply the regularised SS-DGP (R-SS-DGP) presented in Section~\ref{sec:l1-r-dgp} to fit an accelerometer recording of human motion. The reason for using R-SS-DGP here is that the recording (see, the first row of Figure~\ref{fig:imu-r-ssdgp}) is found to have some sharp changes and artefacts. Hence, we aim at testing if we can use sparse length scale and magnitude to describe such data. The collection of accelerometer recordings and the experiment settings are detailed in~\citet{Hostettler2018} and~\citet{Zhao2021RSSGP}, respectively.

\begin{figure}[t!]
	\centering
	\includegraphics[width=.99\linewidth]{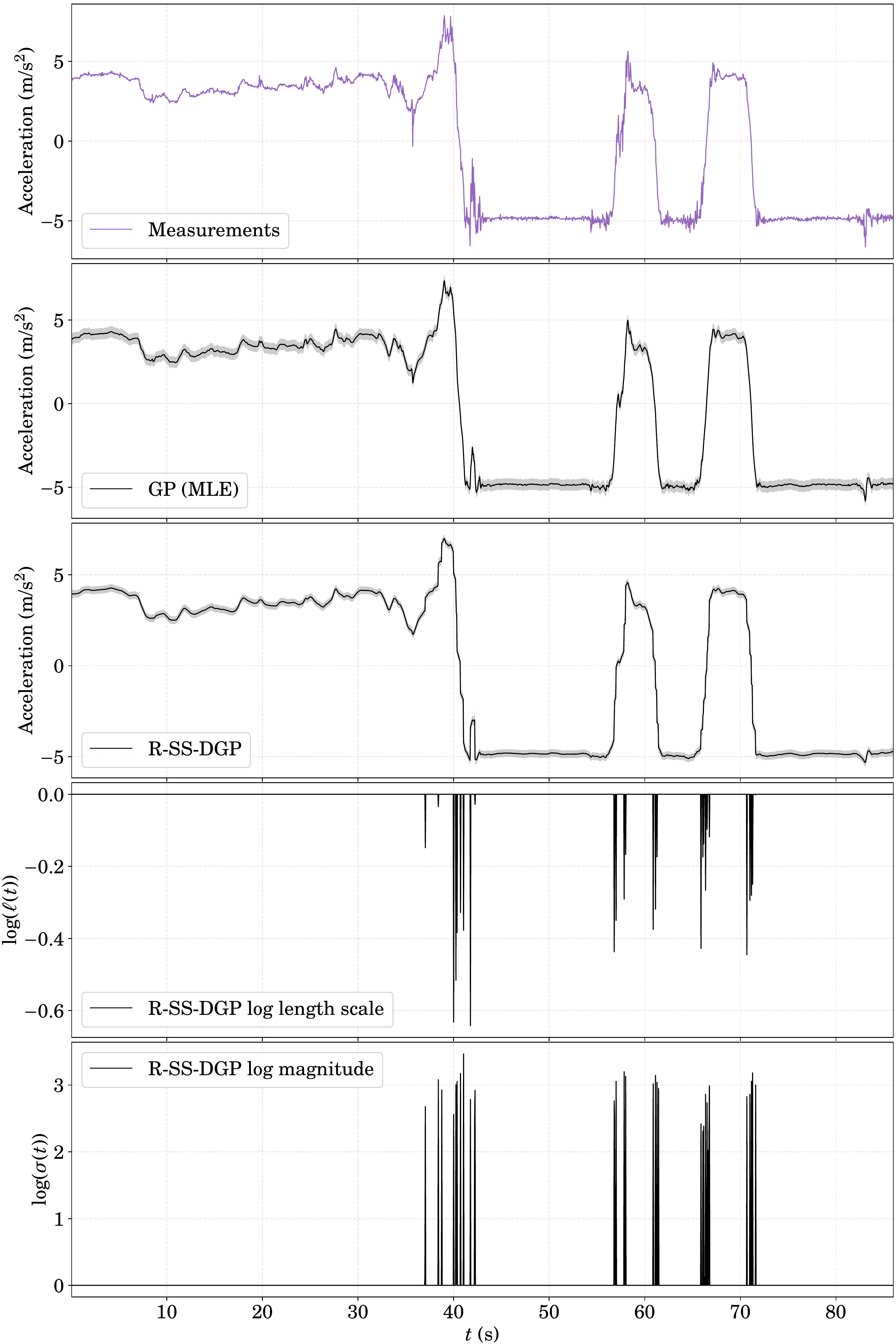}
	\caption{Human motion modelling with an R-SS-DGP. The GP here uses a \matern $\nu=3\,/\,2$ covariance function. Shaded area stands for 0.95 confidence interval.}
	\label{fig:imu-r-ssdgp}
\end{figure}

A demonstrative result is shown in Figure~\ref{fig:imu-r-ssdgp}. We see that the fit of R-SS-DGP is smoother than that of GP. Moreover, the posterior variance of R-SS-DGP is also found to be reasonably smaller than GP. It is also evidenced from the figure that the GP does not handle the artefacts well, for example, around times $t=55$~s and $62$~s. Finally, we find that the learnt length scale and magnitude (in log transformation) are sparse, and that they can respond sharply to the abrupt signal changes and artefacts.

\section{Maritime situational awareness}
\label{sec:maritime}
Another area of applications of (deep) GPs is  autonomous maritime navigation. In~\citet{Sarang2020}, we present a literature review on the sensor technology and machine learning methods for autonomous vessel navigation. In particular, we show that GP-based methods are able to analyse ship trajectories~\citep{Rong2019}, detect navigation abnormality~\citep{Kowalska2012, Smith2014}, and detect/classify vessels~\citep{XiaoZ2017}.

\begin{figure}[t!]
	\centering
	\includegraphics[width=.99\linewidth]{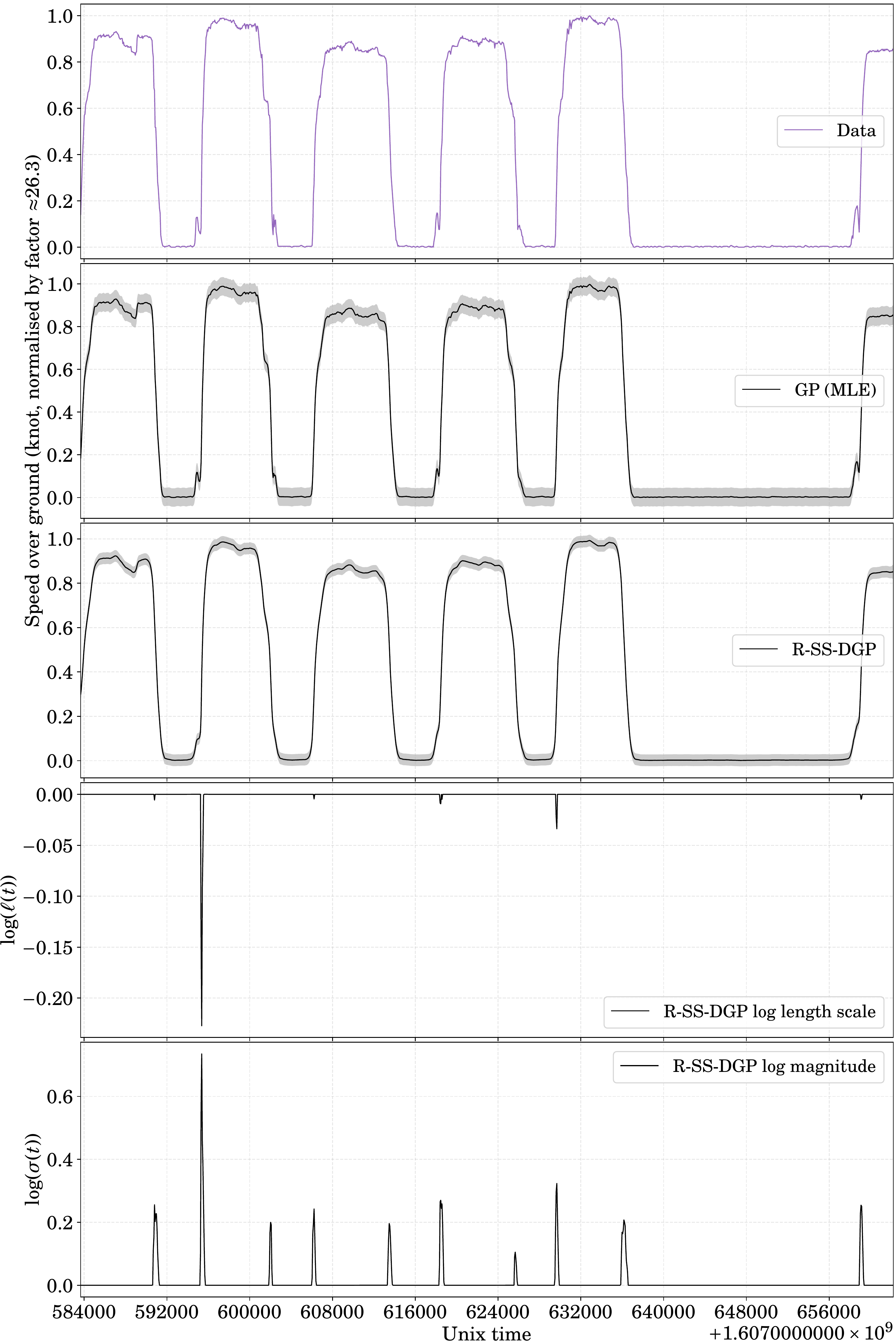}
	\caption{Modelling AIS recording (speed over ground) of MS Finlandia with an R-SS-DGP. The GP here uses a \matern $\nu=3\,/\,2$ covariance function. Shaded area stands for 0.95 confidence interval.}
	\label{fig:ais-ssdgp}
\end{figure}

In Figure~\ref{fig:ais-ssdgp}, we present an example for fitting an automatic identification system (AIS) recording by using an R-SS-DGP. The recording is taken from MS Finlandia (Helsinki--Tallinn) by Fleetrange Oy on December 10, 2020. We see from the figure that the fit of R-SS-DGP is smoother than that of GP. Moreover, the learnt length scale and magnitude parameters are flat and jump at the acceleration/deceleration points.

\chapter{Summary and discussion}
\label{chap:summary}
In this chapter we present a concise summary of Publications I--VII as well as discussion on a few unsolved problems and possible future extensions.

\section{Summary of publications}
This section briefly summaries the contributions of Publications~I--VII and highlights their significances.

\subsection*{Publication~\cp{paperTME} (Chapter~\ref{chap:tme})}
This paper proposes a new class of non-linear continuous-discrete Gaussian filters and smoothers by using the Taylor moment expansion (TME) scheme to predict the means and covariances from SDEs. The main significance of this paper is that the TME method can provide asymptotically exact solutions of the predictive mean and covariances required in the Gaussian filtering and smoothing steps. Secondly, the paper analyses the positive definiteness of TME covariance approximations and thereupon presents a few sufficient conditions to guarantee the positive definiteness. Lastly, the paper analyses the stability of TME Gaussian filters.

\subsection*{Publication~\cp{paperSSDGP} (Chapter~\ref{chap:dssgp})}
This paper introduces state-space representations of a class of deep Gaussian processes (DGPs). More specifically, the paper defines DGPs as vector-valued stochastic processes over collections of conditional GPs, thereupon, the paper represents DGPs in hierarchical systems of the SDE representations of their conditional GPs. The main significance of this paper is that the resulting state-space DGPs (SS-DGPs) are Markov processes, so that the SS-DGP regression problem is computationally cheap (i.e., linear with respect to the number of measurements) by using continuous-discrete filtering and smoothing methods. Secondly, the paper identifies that for a certain class of SS-DGPs the Gaussian filtering and smoothing methods fail to learn the posterior distributions of their state components. Finally, the paper features a real application of SS-DGPs in modelling a gravitational wave signal. 

\subsection*{Publication~\cp{paperKFSECG} (Section~\ref{sec:spectro-temporal})}
This paper is an extension of Publication~\cp{paperKFSECGCONF}. In particular, the quasi-periodic SDEs are used to model the Fourier coefficients instead of the Ornstein--Uhlenbeck ones used in Publication~\cp{paperKFSECGCONF}. This consideration leads to state-space models for which the measurement representations are time-invariant therefore, one can use steady-state Kalman filters and smoothers to solve the spectro-temporal estimation problem with lower computational cost compared to Publication~\cp{paperKFSECGCONF}. This paper also expands the experiments for atrial fibrillation detection by taking into account more classifiers. 

\subsection*{Publication~\cp{paperDRIFT} (Section~\ref{sec:drift-est})}
This paper is concerned with the state-space GP approach for estimating unknown drift functions of SDEs from partially observed trajectories. This approach is significant mainly in terms of computation, as the computational complexity scales linearly in the number of measurements. In addition, the state-space GP approach allows for using high-order It\^{o}--Taylor expansions in order to give accurate SDE discretisations without the necessity to compute the covariance matrices of the derivatives of the GP prior.

\subsection*{Publication~\cp{paperKFSECGCONF} (Section~\ref{sec:spectro-temporal})}
This paper introduces a state-space probabilistic spectro-temporal estimation method and thereupon applies the method for detecting atrial fibrillation from electrocardiogram signals. The so-called probabilistic spectro-temporal estimation is a GP regression-based model for estimating the coefficients of Fourier expansions. The main significance of this paper is that the state-space framework allows for dealing with large sets of measurements and high-order Fourier expansions. Also, the combination of the spectro-temporal estimation method and deep convolutional neural networks shows efficacy for classifying a class of electrocardiogram signals.

\subsection*{Publication~\cp{paperMARITIME} (Section~\ref{sec:maritime})}
This paper reviews sensor technologies and machine learning methods for autonomous maritime vessel navigation. In particular, the paper lists and reviews a number of studies that use deep learning and GP methods for vessel trajectory analysis, ship detection and classification, and ship tracking. The paper also features a ship detection example by using a deep convolutional neural network. 

\subsection*{Publication~\cp{paperRNSSGP} (Section~\ref{sec:l1-r-dgp})}
This paper solves $L^1$-regularised DGP regression problems under the alternating direction method of multipliers (ADMM) framework. The significance of this paper is that one can introduce regularisation (e.g., sparseness or total variation) at any level of the DGP component hierarchy. Secondly, the paper provides a general framework that allows for regularising both batch and state-space DGPs. Finally, the paper presents a convergence analysis for the proposed ADMM solution of $L^1$-regularised DGP regression problems.

\section{Discussion}
Finally, we end this thesis with discussion on some unsolved problems and possible future extensions.

\subsection*{Positive definiteness analysis for high-order and high-dimensional TME covariance approximation}
Theorem~\ref{thm:tme-cov-pd} provides a sufficient condition to guarantee the positive definiteness of TME covariance approximations. However, the use of Theorem~\ref{thm:tme-cov-pd} soon becomes infeasible as the expansion order $M$ and the state dimension $d$ grow large. In practice, it can be easier to check the positive definiteness numerically when $d$ is small. 

\subsection*{Practical implementation of TME}
A practical challenge with implementing TME consists in the presence of derivative terms in $\A$ (see, Equation~\eqref{equ:generator-ito}). This in turn implies that the iterated generator $\A^M$ further requires the computation of derivatives of the SDE coefficients up to order $M$. While the derivatives of $\A$ are easily computed by hand, the derivatives in $\A^M$ require more consideration as they involve numerous applications of the chain rule, not to mention the multidimensional operator $\Am$ in Remark~\ref{remark:multidim-generator}.

While in our current implementation we chose to use symbolic differentiation (for ease of implementation as well as portability across languages), several things can be said against using it. Symbolic differentiation explicitly computes full Jacobians, where only vector-Jacobian/Jacobian-vector products would be necessary. This induces an unnecessary overhead that grows with the dimension of the problem. Also, symbolic differentiation is usually independent of the philosophy of modern differentiable programming frameworks and the optimisation for parallelisable hardware (e.g., GPUs), hence they may incur a loss of performance on these.

Automatic differentiation tools, for instance, TensorFlow and JaX are amenable to computing the derivatives in $\Am$. Furthermore, they provide efficient computations for Jacobian-vector/vector-Jacobian products. We hence argue that these tools are worthwhile for performance improvement in the future\footnote{By the time of the pre-examination of this thesis, the TME method is now implemented in JaX as an open source library (see, Section~\ref{sec:codes}).}.

\subsection*{Generalisation of the identifiability analysis}
The identifiability analysis in Section~\ref{sec:identi-problem} is limited to SS-DGPs for which the GP elements are one-dimensional. This dimension assumption is used in order to derive Equation~\eqref{equ:vanish-cov-eq1} in closed-form. However, it is of interest to see whether we can generalise Lemma~\ref{lemma:vanishing-prior-cov} for SS-DGPs that have multidimensional GP elements.

The abstract Gaussian filter in Algorithm~\ref{alg:abs-gf} assumes that the prediction steps are done exactly. However, this assumption may not always be realistic because Gaussian filters often involve numerical integrations to predict through SDEs, for example, by using sigma-point methods. Hence, it is important to verify if Lemma~\ref{lemma:vanishing-prior-cov} still holds when one computes the filtering predictions by some numerical means.

\subsection*{Spatio-temporal SS-DGPs}
SS-DGPs are stochastic processes defined on temporal domains. In order to model spatio-temporal data, it is necessary to generalise SS-DGPs to take values in infinite-dimensional spaces~\citep{Giuseppe2014}. A path for this generalisation is to leverage the stochastic partial differential equation (SPDE) representations of spatio-temporal GPs. To see this, let us consider an $\mathbb{H}$-valued stochastic process $U \colon \T \to \mathbb{H}$ governed by a well-defined SPDE
\begin{equation}
	\diff U(t) = A \, U(t) \diff t + B \diff W(t) \nonumber
\end{equation}
with some boundary and initial conditions, where $A\colon \mathbb{H} \to \mathbb{H}$ and $B\colon \mathbb{W} \to \mathbb{H}$ are linear operators, and $W\colon \T \to \mathbb{W}$ is a $\mathbb{W}$-valued Wiener process. Then we can borrow the idea presented in Section~\ref{sec:ssdgp} to form a spatio-temporal SS-DGP by hierarchically composing such SPDEs of the form above.

A different path for generalising SS-DGPs is shown by~\citet{Emzir2020}. Specifically, they build deep Gaussian fields based on the SPDE representations of \matern fields~\citep{Whittle1954, Lindgren2011}. However, we should note that this approach gives random fields instead of spatio-temporal processes.

\renewcommand{\bibname}{References}
\bibliographystyle{plainnat}
\bibliography{refs}

\errata

\addpublication{Zheng Zhao, Toni Karvonen, Roland Hostettler, and Simo S\"{a}rkk\"{a}}{Taylor moment expansion for continuous-discrete Gaussian filtering}{IEEE Transactions on Automatic Control}{Volume 66, Issue 9, Pages 4460--4467}{December}{2020}{Zheng Zhao, Toni Karvonen, Roland Hostettler, and Simo S\"{a}rkk\"{a}}{paperTME}
\addcontribution{Zheng Zhao wrote the article and produced the results. The stability analysis is mainly due to Toni Karvonen. Roland Hostettler gave useful comments. Simo S\"{a}rkk\"{a} contributed the idea.}
\adderrata{In Example 7, the coefficient $\Phi_{x,2}$ should multiply with a factor $2$.}

\addpublication{Zheng Zhao, Muhammad Emzir, and Simo S\"{a}rkk\"{a}}{Deep state-space Gaussian processes}{Statistics and Computing}{Volume 31, Issue 6, Article number 75, Pages 1--26}{September}{2021}{Zheng Zhao, Muhammad Emzir, and Simo S\"{a}rkk\"{a}}{paperSSDGP}
\addcontribution{Zheng Zhao wrote the article and produced the results. Muhammad Emzir and Simo S\"{a}rkk\"{a} gave useful comments.}

\addpublication{Zheng Zhao, Simo S\"{a}rkk\"{a}, and Ali Bahrami Rad}{Kalman-based spectro-temporal ECG analysis using deep convolutional networks for atrial fibrillation detection}{Journal of Signal Processing Systems}{Volume 92, Issue 7, Pages 621--636}{April}{2020}{Zheng Zhao, Simo S\"{a}rkk\"{a}, and Ali Bahrami Rad}{paperKFSECG}
\addcontribution{Zheng Zhao wrote the article and produced the results. Ali Bahrami Rad helped with the experiments. Simo S\"{a}rkk\"{a} came up with the spectro-temporal idea.}

\addpublication[conference]{Zheng Zhao, Filip Tronarp, Roland Hostettler, and Simo S\"{a}rkk\"{a}}{State-space Gaussian process for drift estimation in stochastic differential equations}{Proceedings of the 45th IEEE International Conference on Acoustics, Speech and Signal Processing (ICASSP)}{Barcelona, Spain, Pages 5295--5299}{May}{2020}{IEEE}{paperDRIFT}
\addcontribution{Zheng Zhao wrote the article and produced the results. Filip Tronarp provided codes for the iterated posterior linearisation filter. Roland Hostettler gave useful comments. Idea was due to Simo S\"{a}rkk\"{a}.}

\addpublication[conference]{Zheng Zhao, Simo S\"{a}rkk\"{a}, and Ali Bahrami Rad}{Spectro-temporal ECG analysis for atrial fibrillation detection}{Proceedings of the IEEE 28th International Workshop on Machine Learning for Signal Processing (MLSP)}{Aalborg, Denmark, 6 pages}{September}{2018}{IEEE}{paperKFSECGCONF}
\addcontribution{Zheng Zhao wrote the article and produced the results. Ali Bahrami Rad helped with the experiments. Simo S\"{a}rkk\"{a} came up with the spectro-temporal idea.}

\addpublication[accepted]{Sarang Thombre, Zheng Zhao, Henrik Ramm-Schmidt, Jos\'{e} M. Vallet García, Tuomo Malkam\"{a}ki, Sergey Nikolskiy, Toni Hammarberg, Hiski Nuortie, M. Zahidul H. Bhuiyan, Simo S\"{a}rkk\"{a}, and Ville V. Lehtola}{Sensors and AI techniques for situational awareness in autonomous ships: a review}{IEEE Transactions on Intelligent Transportation Systems}{20 pages}{September}{2020}{IEEE}{paperMARITIME}
\addcontribution{Zheng Zhao wrote the reviews of AI techniques and produced corresponding results.}

\addpublication[submitted]{Zheng Zhao, Rui Gao, and Simo S\"{a}rkk\"{a}}{Hierarchical Non-stationary temporal Gaussian processes with $L^1$-regularization}{Statistics and Computing}{20 pages}{May}{2021}{Zheng Zhao, Rui Gao, and Simo S\"{a}rkk\"{a}}{paperRNSSGP}
\addcontribution{Zheng Zhao wrote the article and produced the results. Rui Gao contributed the convergence analysis. Simo S\"{a}rkk\"{a} gave useful comments.}

\includepdf[pages=-]{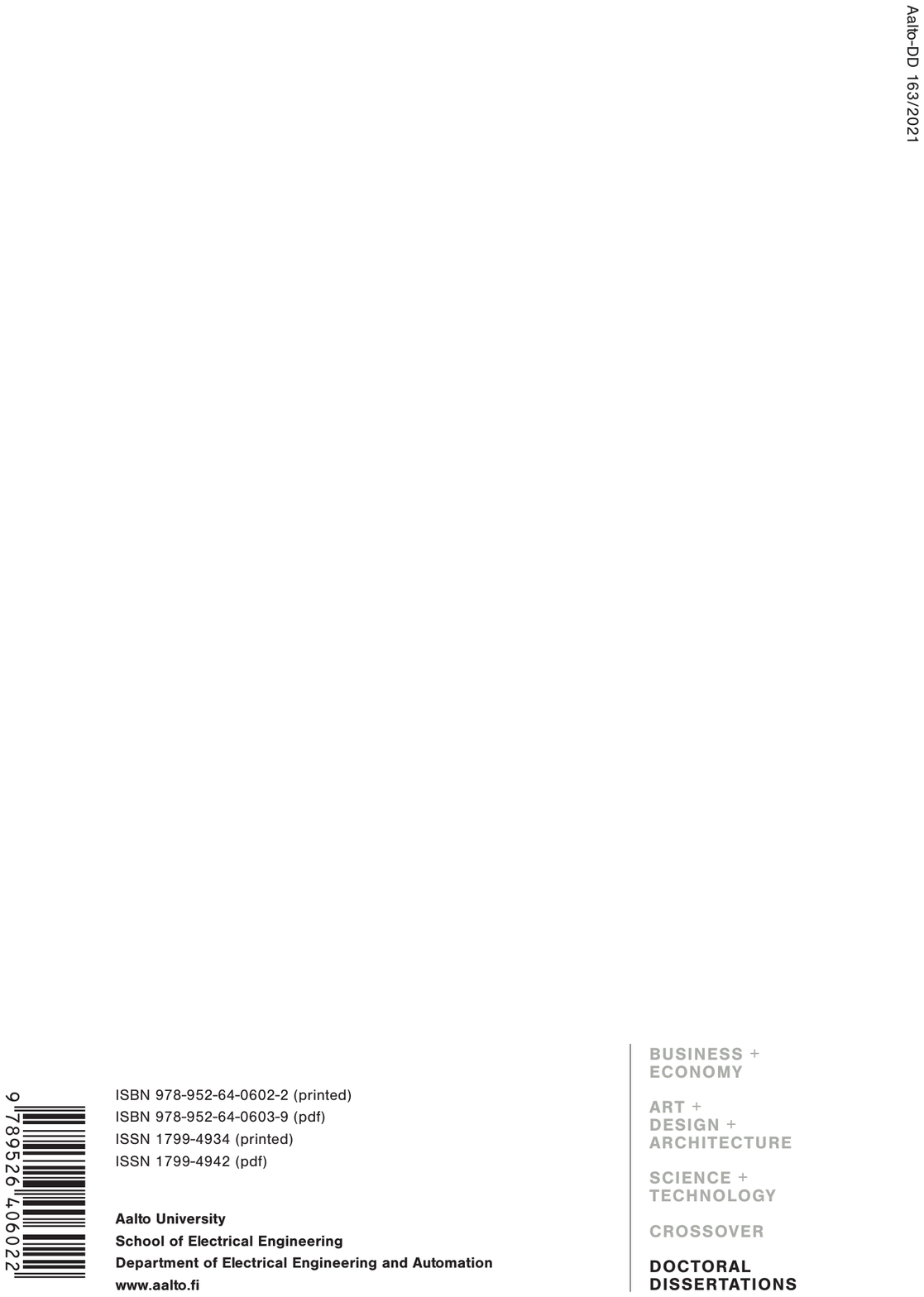}

\end{document}